\documentclass[xcolor=svgnames,11pt,reqno]{amsart}
\usepackage[T1]{fontenc}
\usepackage[round]{natbib}
\usepackage{amsmath}
\usepackage{amssymb}
\usepackage{amsthm}
\usepackage{fancyhdr}
\usepackage{enumitem}
\usepackage{comment}
\usepackage{relsize}
\usepackage{nicefrac}
\usepackage{bm}
\usepackage{mathrsfs}
\usepackage{graphicx}
\usepackage[utf8]{inputenc}
\usepackage{cancel}
\usepackage{dsfont}
\usepackage{mathtools}
\usepackage{tikz,xcolor}
\usepackage{dirtytalk}
\usetikzlibrary{cd,decorations.pathreplacing,matrix,arrows,positioning,automata,shapes,shadows,calc,fadings,decorations,through,intersections}
\usepackage{float} 
\usepackage{subcaption}
\usepackage{multirow}
\usepackage[left=2.5cm,top=2.8cm,right=2.5cm,bottom=2.5cm,head=3cm,headsep=1cm, foot=3cm]{geometry}
\usepackage{hyperref}
\usepackage[skip=2mm,indent=12pt]{parskip}
\definecolor{darkblue}{rgb}{0.1,0.1,0.9}
\definecolor{darkred}{rgb}{0.9,0.1,0.1}

\hypersetup{colorlinks, allcolors= darkblue}

\newtheorem{theorem}{Theorem}[section]
\newtheorem{definition}[theorem]{Definition}
\newtheorem{proposition}[theorem]{Proposition}
\newtheorem{assumption}[theorem]{Assumption}
\newtheorem{corollary}[theorem]{Corollary}
\newtheorem{lemma}[theorem]{Lemma}
\newtheorem{remark}[theorem]{Remark}

\newcommand{\cL}{\mathcal{L}}

\newcommand{\cP}{\mathcal{P}}

\newcommand{\cV}{\mathcal{V}}
\newcommand{\cI}{\mathcal{I}}
\newcommand{\cB}{\mathcal{B}}

\newcommand{\cC}{\mathcal{C}}

\newcommand{\cR}{\mathcal{R}}
\newcommand{\bbR}{\mathbb{R}}

\newcommand{\cO}{\mathcal{O}}

\newcommand{\R}{\mathbb{R}}

\usepackage{booktabs}
\usepackage{siunitx}
\usepackage[colorinlistoftodos,textsize=small]{todonotes}
\usepackage{cleveref}
\usepackage{caption}
\usepackage{anyfontsize}
\usepackage{yhmath}
\usepackage{comment}

\renewcommand{\l}{\left}
\renewcommand{\r}{\right}
\renewcommand{\tilde}{\widetilde}

\DeclareMathOperator*{\argsup}{\arg\sup}

\allowdisplaybreaks
\makeatletter

\newcommand{\Rmnum}[1]{\expandafter\@slowromancap\romannumeral #1@}

\title[Incentive Pareto Efficiency under Adverse Selection]{Incentive Pareto Efficiency in\vspace{0.2cm}\\Monopoly Insurance Markets with Adverse Selection}

\author[Maria Andraos and Mario Ghossoub]{Maria Andraos\vspace{0.1cm}\\ University of Waterloo\vspace{0.8cm}\\
Mario Ghossoub\vspace{0.1cm}\\ University of Waterloo\vspace{0.8cm}\\ \today}

\address{{\bf Maria Andraos}: University of Waterloo -- Department of Statistics and Actuarial Science -- 200 University Ave.\ W.\ -- Waterloo, ON, N2L 3G1 -- Canada\vspace{0.1cm}}\email{\href{mailto:mandraos@uwaterloo.ca}{mandraos@uwaterloo.ca}\vspace{0.3cm}}

\address{{\bf Mario Ghossoub}: University of Waterloo -- Department of Statistics and Actuarial Science -- 200 University Ave.\ W.\ -- Waterloo, ON, N2L 3G1 -- Canada\vspace{0.1cm}}
\email{\href{mailto:mario.ghossoub@uwaterloo.ca}{mario.ghossoub@uwaterloo.ca}\vspace{0cm}}

\thanks{ \vspace{0.2cm}}

\thanks{ \vspace{0.2cm}}

\thanks{Mario Ghossoub acknowledges financial support from the Natural Sciences and Engineering Research Council of Canada (NSERC Grant No.\ 2024-03744).}

\begin{document}

\maketitle

\begin{abstract}
We study a monopolistic insurance market with hidden information, where the agent's type $\theta$ is private information that is unobservable to the insurer, and it is drawn from a continuum of types. The hidden type affects both the loss distribution and the risk attitude of the agent. Within this framework, we show that a menu of contracts is incentive efficient if it maximizes social welfare function, subject to incentive compatibility and individual rationality constraints. This holds for general utility functionals. In the special case of Yaari Dual Utility, we provide two partial converse statements to this result, and we give a semi-explicit characterization of optimal solutions to the social welfare maximization problem. We do this under two different settings: (i) the first assumes that types are ordered in a way such that larger values of $\theta$ correspond to more risk-averse types who face stochastically larger losses; whereas (ii) the second assumes that larger values of $\theta$ correspond to less risk-averse types who face stochastically larger losses. In both settings, the structure of optimal menus of contracts depends on the level of the social welfare weight, and we examine several properties thereof. 
\end{abstract} 

\sloppy

\bigskip
\section{Introduction}
In insurance markets, contracts are written between two parties who generally do not share the same information about the underlying risk. In these markets, information asymmetry arises naturally because the agent typically knows more about their own exposure or behavior than the insurer can observe, as initially noted by \cite{allais53b}. Two major obstacles for a smooth running of the insurance mechanism are moral hazard and adverse selection, which attracted the attention of economists. Both problems have been extensively studied for their implications on contract design and market efficiency.

Moral hazard arises when the outcome of a contract is partly influenced by the agent’s actions, and the insurer cannot, without incurring costs, observe or verify to what extent the reported losses are attributable to the agent’s behavior. Specifically, \textit{ex ante} moral hazard occurs when the agent's unobservable actions affect the probability of a loss before it occurs. This has been studied by \cite{pauly1978overinsurance}, \cite{marshall1976moral}, \cite{holmstrom1979moral} and \cite{shavell1979moral}, among others. \textit{Ex post} moral hazard, on the other hand, occurs when the agent can misreport or influence the realized magnitude of the loss after it occurs. This was first pointed out by \cite{spence1978insurance} and later studied by \cite{townsend1979optimal}, for instance.

Our focus in this paper is on adverse selection, where the agent possesses private information about their risk characteristics and may use this hidden information to their own advantage. The insurer offers a menu of contracts designed in such a way that each agent type selects the contract intended for them, thereby revealing the agent's private information through their contract choice. This self selection mechanism must satisfy incentive compatibility, ensuring that each agent prefers the contract designed for their own type over those intended for other types. \cite{stiglitz1976equilibrium} and \cite{stiglitz1977monopoly} study insurance markets where the agent’s type, high risk or low risk, is the private information. They both assume that the agent is risk averse with expected utility preferences.  
In particular, \cite{stiglitz1976equilibrium} consider a competitive market with multiple risk-neutral insurers, and they show that under information asymmetry, a separating equilibrium may arise. Different risk types are offered different insurance contracts tailored to their own characteristics. Low-risk agents receive partial coverage at a lower premium, while high-risk agents obtain full coverage but pay a higher premium.
In contrast, \cite{stiglitz1977monopoly} considers a monopoly market with a single risk-neutral insurer offering a non-linear pricing menu subject to individual rationality and incentive compatibility constraints. He shows that under information asymmetry the equilibrium is also separating. In this setting, low-risk types may prefer not to purchase any coverage but if they do, they receive partial coverage; whereas high-risk types receive full coverage. 
\cite{chade2012optimal} extend the work of \cite{stiglitz1977monopoly} by moving beyond the two-type framework, to a setting with a continuum of types. They show that in equilibrium, the monopolist insurer expects a strictly positive profit. The highest risk type receives full coverage, the lowest risk type is indifferent between insurance and no insurance, and all other types receive partial coverage.
\cite{gershkov2023optimal} reexamine the classic monopoly insurance problem under adverse selection of \cite{stiglitz1977monopoly}, allowing for a continuum of privately known types, a type-dependent loss distribution hidden from the insurer, and Yaari's dual utility for the agent's preferences (\cite{yaari1987dual}) rather than expected utility. This dual utility of Yaari is represented by a Choquet integral with respect to a distorted probability, where the distortion function represents the risk attitude (risk aversion) of the policyholder, and it is assumed to be known by the insurer. The monopolist risk neutral insurer's problem is formulated as a constrained optimal contracting problem of expected-profit maximization subject to incentive compatibility and individual rationality constraints. They show that the optimal menu of contracts takes the form of a layered deductible indemnity schedule, under a regularity condition. Moreover, under specific technical conditions, the optimal menus consist of either linear deductible contracts or of upper-limit contracts. Consistent with the monopolist setting, they also show that, under asymmetric information, the insurer earns strictly positive profit. 
Recently, \cite{ghossoub2025optimal} consider a monopoly insurance market in which the insurer is risk neutral and profit maximizing. The agent's preferences are given by Yaari's dual utility, and the agent's risk aversion level (or risk attitude) is his private information. Hence, in contrast to \cite{gershkov2023optimal}, the insurer observes the loss distribution but is unable to observe the agent's risk attitude. They formulate the insurer's problem as designing an incentive compatible and individually rational menu of contracts that maximizes expected profit. They show that the optimal menu consists of layered deductible contracts, insurance coverage and premia are monotone in the level of risk aversion, the most risk averse agent receives full coverage, and the insurer earns strictly positive profit.

Pareto efficiency under asymmetric information has been studied extensively in the literature. 
For example, early contributions by \cite{prescott1984pareto}, \cite{jerez2003dual}, and \cite{bisin2006efficient} analyze constrained Pareto efficiency in environments with incentive compatibility, primarily in settings with finitely many types. 
\cite{ghossoub2025optimal} label this efficiency as incentive Pareto optimality. They show that any individually-rational and incentive-compatible menu that maximizes a social welfare function is incentive Pareto optimal, thereby providing a sufficient condition for incentive efficiency. Crucially, the necessity part of the equivalence between social welfare maximization and incentive Pareto efficiency was left unaddressed. This is arguably the more interesting, and the more complex result, which we establish in this paper.

In this paper, we consider a monopolistic insurance market, in which the agent's type is private information, hidden from the insurer and drawn from a continuum $\Theta$ of types. The agent faces a type-dependent loss, with a continuous distribution function that is unknown to the insurer. Additionally, the agent's utility functional $U$ is a function of the type $\theta$, in that $\theta$ is a parameter of the agent's utility evaluation of their welfare. Consequently, the agent's type affects both the riskiness of the agent (the loss distribution) and the risk-attitude of the agent (e.g., their risk aversion, through a parameterization of the utility functional). We extend the notion of incentive efficiency (or incentive Pareto optimality) introduced in \cite{ghossoub2025optimal} in the context of Yaari's Dual Utility to arbitrary type-dependent utility functionals for both the policyholder and the insurer.

Our first result shows that a menu of contracts is incentive Pareto efficient if it maximizes a social welfare functional, subject to individual rationality and incentive compatibility, in the presence of hidden information and a continuum of types. This result holds for arbitrary utility functionals of both the insurer and the agent, thereby providing the theoretical foundation for subsequent analysis. In the case of Yaari's Dual Utilities, we also provide two converse results.

We then characterize optimal solutions to the social welfare problem in the specific case of Yaari's Dual Utility. Specifically, in this special case of our general setup, we assume that both the insurer's and agent's preferences are represented by Yaari Dual Utility functionals, expressed as a Choquet integral with respect to a distorted probability. In this case, the insurer can observe neither the agent's risk attitude nor their type-dependent loss distribution. We consider two orderings of the type space. In the first case, higher types are more risk averse and face stochastically larger losses. In the second case, higher types are less risk averse and face stochastically larger losses. In both cases, each type of agent is assumed to be weakly more risk averse than the monopolistic insurer.
We show in Theorem \ref{th:solution_characterization_IPO} that, depending on the social weight level and under some technical conditions, the optimal menu of contracts exhibits one of two distinct forms: either a layered marginal retention structure, or full coverage (zero retention).

Additionally, in both aforementioned orderings of the type space, the optimal welfare-maximizing menu displays some desirable monotonicity properties. Specifically, in the first ordering of the type space, higher types facing stochastically larger losses receive more coverage at higher premia. In the separating regions of coverage, efficiency at the top holds: full coverage is provided to the highest type, if the highest type's loss distribution is nontrivial and the insurer is strictly less risk averse than the highest type. This echoes the results of \cite{chade2012optimal}, \cite{gershkov2023optimal}, and \cite{ghossoub2025optimal}. The insurer absorbs the surplus from the lowest type leaving them indifferent between insuring and not insuring, and higher types derive lower utilities from the optimal menu. On the insurer's side, the utility depends on the degree of loss transfer. Serving higher types who face stochastically larger losses does not necessarily yield higher utility for the insurer. In particular, Proposition \ref{prop:insurerutilityoptimalmenu} characterizes how the insurer's utility from the optimal menu varies across types. Similar results hold for the second kind of ordering of the type space.

The rest of the paper is organized as follows. Section \ref{sec:insurancemarket} introduces the insurance market model. In Section \ref{sec:IPO}, we define incentive-efficient menus of contracts and establish a social-welfare sufficient condition for incentive efficiency, in the general case of type-dependent utility functionals. Section \ref{sec:DU} provides a characterization of optimal welfare-maximizing menus in the Yaari Dual Utility framework, under different assumptions on the ordering of the type space. Section \ref{sec:conclusion} concludes. Proofs and related analysis are given in the \hyperlink{LinkToAppendix}{Appendices}.

\bigskip
\section{The Insurance Market}
\label{sec:insurancemarket}

Let $(S, \Sigma, \mathbb{P})$ be a probability space, and denote by $B(\Sigma)$ the space of bounded, real-valued, and $\Sigma$-measurable functions. We consider an insurance market in which an agent is facing an insurable loss, modeled as an element of $B(\Sigma)$, and seeking coverage from a monopolist insurer, in return for a premium payment.

The agent has a type denoted by $\theta$, which is private information that is unobservable to the insurer. We assume that the agent's type $\theta$ is drawn from a continuum $\Theta = [\underline{\theta}, \bar \theta]$ of types. Let $\cB(\Theta)$ denote the Borel sigma-algebra on the type space $\Theta$, and equip the measurable space of types $(\Theta, \cB(\Theta))$ with the Lebesgue measure $\cL$. 

We assume that the loss faced by the agent is type dependent. Specifically, for $\theta \in \Theta$, the type-$\theta$ agent faces a nonnegative loss $L_\theta \in B(\Sigma)$, which can be covered by an indemnity function $I_\theta (L_\theta)$, in exchange for a nonnegative premium payment $p_\theta \in \bbR_+$. The loss $L_\theta$ takes values in $[0, \bar L_\theta]$, for some $\bar L_\theta < +\infty$. For simplicity, one can assume that for each $\theta \in \Theta$, the random variable $L_\theta$ takes values in the interval $[0, \bar L]$, where $\bar L := \underset{\theta \in \Theta}{\sup} \ \bar L_\theta < +\infty$ is the uniform upper bound.

To rule out potential \textit{ex post} moral hazard, we impose the customary restriction that the market only offers indemnities that satisfy the no-sabotage condition of \cite{CarlierDana}.

\begin{assumption} \label{ass:feasible_I}
We restrict the set of admissible indemnities to the following set of 1-Lipschitz and non-decreasing functions:
$$
\cI =  \{ I: [0, \bar L] \to [0, \bar L], I(0) =0 , 0 \leq I(l_1) - I(l_2) \leq l_1 - l_2, \ \forall \, 0 \leq l_2 \leq l_1 \leq \bar L \  \}.
$$
\end{assumption}

Not knowing the agent's type, the insurer sets out to design a menu of contracts, from which the agent can select one single contract. 

\begin{definition}
\label{def:contract_menu}
A contract is a pair $(I,p)\in \cI\times \bbR_+$, where $I\in\cI$ is a feasible indemnity and $p\geq 0$ is the premium paid by the policyholder to the insurer. A menu of contracts is a collection
$$
(I_{\theta}, p_{\theta})_{\theta \in \Theta},
$$
such that $(I_{\theta},p_{\theta})\in\cI\times\bbR_+$, for all $\theta\in\Theta$, the map $\theta \mapsto p_\theta$ is measurable, and the map $(\theta, l) \mapsto I_\theta(l)$ is jointly Borel measurable on $\Theta \times [0, \bar L]$.
\end{definition}

\medskip

Preferences in this market are represented by functionals $U,V : \Theta \times \cI \times \bbR_+ \to \R$, where for a given triplet $\left(\theta, I, p\right) \in \Theta \times \cI \times \bbR_+$, 
$$U\left(\theta, I, p\right)$$

\noindent denotes the end-of-period utility of a type-$\theta$ agent after purchasing the contract $(I, p)$; and 
$$V\left(\theta, I, p\right)$$

\noindent denotes the insurer's utility from providing the contract $(I,p)$ to a type-$\theta$ agent. All throughout, we make the following normalization:
$$V(\theta, 0, 0) = 0, \ \ \forall \, \theta \in \Theta.$$

\medskip

For notational convenience, we write
$$
U_{\theta}(I, p) := U(\theta, I, p) \ \ \text{and} \ \ 
V_{\theta}(I, p) := V(\theta, I, p),
$$

\noindent where $\theta$ captures type dependence in the loss riskiness and in the agent's risk characteristics.

\bigskip 
\section{Efficiency under asymmetric information}
\label{sec:IPO}
In full information settings, classical Pareto efficiency ensures that no one can be made better off without making someone else worse off. However, under asymmetric information, this classical concept of efficiency is no longer appropriate, unless incentive compatibility is imposed. This is because without incentive compatibility, a type-$\theta$ agent might misreport their type and select a contract intended for other types. In this section, we discuss the notion of incentive Pareto optimality previously examined by \cite{ghossoub2025optimal}, and we provide a social-welfare characterization thereof. 

\medskip
\subsection{Incentive Pareto Optimality}
Let $\mu$ be a probability measure on the measurable space of types $(\Theta, \cB(\Theta))$ representing the distribution of agent types in the market. That is, $\mu(B)$ denotes the proportion of agent types lying in $B$, for any measurable set $B \in \cB(\Theta)$. 

\medskip

\begin{assumption}
We assume that the probability measure $\mu$ is absolutely continuous with respect to the Lebesgue measure $\cL$ with Radon-Nikodym derivative $q$. That is, 
$$
\mu(B) = \int_B q(\theta) \, d\theta, \ \text{for all $B \in \cB(\Theta)$ }.
$$
\end{assumption}

\medskip

\noindent The cumulative distribution function over types is defined by,
$$
Q(\theta) = \mu( [\underline{\theta}, \theta] ), \ \forall \theta \in \Theta,
$$

\noindent with corresponding density function $q$, with respect to Lebesgue measure.

\medskip

\begin{assumption}\label{ass:integrability}
For any menu of contracts $(I_{\theta}, p_{\theta})_{\theta \in \Theta}$, the mappings $\theta \mapsto U_\theta(I_\theta,p_\theta)$ and $\theta \mapsto V_\theta(I_\theta,p_\theta)$ are in $L^1 (\Theta, \mu)$.
\end{assumption}

Assumption \ref{ass:integrability} is a technical condition that ensures that the agent’s and the insurer's utilities are integrable over the type space $\Theta$, for any menu of contracts. Consequently, aggregate utilities are well-defined as Bochner integrals. We refer to Appendix \ref{AppBochner} for a detailed discussion of Bochner spaces. 

\medskip

\begin{definition} \label{def:IR}
A menu of contracts $(I_{\theta}, p_{\theta} )_{\theta \in \Theta}$ is said to be individually rational IR if both of the following hold.
\begin{enumerate}
\item[(P1)] Each agent type is incentivized to participate in the market. That is,
$$
U_\theta (I_{\theta}, p_{\theta}) \geq U_\theta( L_{\theta}, 0), \, \text{for each $\theta \in \Theta$},
$$
where $U_\theta (L_{\theta}, 0)$ denotes the utility of the type-$\theta$ agent in the absence of insurance.

\smallskip

\item[(P2)] The insurer is incentivized to participate in the market. That is, 
$$\int_{\Theta} V_{\theta} (I_{\theta}, p_{\theta}) \, d\mu(\theta) \geq \int_{\Theta} V_{\theta} (0,0) \, d \mu(\theta) = 0,$$
where $V_{\theta} (0,0)$ denotes the insurer's utility when no insurance is provided to the type-$\theta$ agent. 
\end{enumerate}
\end{definition}

\noindent We denote by $\cI\cR$ the set of all individual rational menus. 

\medskip

\begin{definition} \label{def:IC}
A menu of contracts $(I_{\theta}, p_{\theta} )_{\theta \in \Theta}$ is said to be incentive compatible IC if no type $\theta$ can benefit from choosing the contract of another type $\theta^{\prime}$. That is, 
$$
U_\theta (I_{\theta}, p_{\theta} ) \geq U_\theta( I_{\theta ^{\prime} }, p_{\theta^{\prime}}), \ \ \text{for each $\theta, \theta ^{\prime} \in \Theta$}.
$$
\end{definition}
\noindent Let $\cI\cC$ be the set of all incentive compatible menus.

\medskip

\begin{definition} 
\label{def:IPO}
A menu $ \left( I^*_{\theta}, p^*_{\theta} \right)_{\theta \in \Theta} \in \cI\cR \cap \cI\cC$ is said to be \textit{incentive efficient} or \textit{incentive Pareto optimal} (IPO), if there does not exist another menu $(I_{\theta}, p_{\theta})_{\theta \in \Theta}  \in \cI\cR \cap \cI\cC$ such that the following two conditions hold:
\begin{enumerate}
\item For $\mu$-almost every $\theta \in \Theta$, 
$$U_{\theta}(I_{\theta}, p_{\theta}) \geq U_{\theta}(I^*_{\theta}, p^*_{\theta}),$$
and in addition
$$\int_{\Theta} V_{\theta} (I_{\theta}, p_{\theta}) \, d\mu (\theta) \geq \int_{\Theta} V_{\theta} (I^*_{\theta}, p^*_{\theta}) \, d \mu(\theta).$$

\medskip

\item At least one of the two following conditions holds: 
$$\int_{\Theta} V_{\theta}(I_{\theta}, p_{\theta}) \, d \mu(\theta) > \int_{\Theta} V_{\theta}(I^*_{\theta}, p^*_{\theta}) \, d \mu(\theta),$$ 

\noindent or 
$$ \mu \left(\left\{ \theta \in \Theta \,; \, U_{\theta} (I_{\theta}, p_{\theta}) >  U_{\theta}(I^*_{\theta}, p^*_{\theta}) \,\right\} \right) >0.$$ 
\end{enumerate}
\end{definition}

\medskip

\noindent We denote by $\cI \cP \cO \subseteq \cI\cR \cap \cI \cC$ the set of all incentive efficient menus.

\medskip
\subsection{Social Welfare Maximization}\label{sub:construction}
We now establish the link between incentive efficient menus of contracts and social welfare maximization.

\medskip

\begin{theorem}
\label{th:IPOiff}
If there exists a probability measure $\eta$ on the measurable space of types $(\Theta, \cB(\Theta))$ that is equivalent to $\mu$, and some $\alpha \in (0,1)$ such that a menu of contracts $( I^*_{\theta} , p^*_{\theta})_{\theta \in \Theta} $ is optimal for the problem
\begin{equation} \label{eq:IPO_sup}
\underset{(I_{\theta}, p_{\theta})_{\theta \in \Theta} \in \cI\cR \cap \cI\cC }  { \sup} \left \{ \alpha \int_{\Theta} U_{\theta} (I_{\theta}, p_{\theta} ) \, d\eta(\theta) + (1-\alpha) \int_{\Theta} V_{\theta} (I_{\theta}, p_{\theta} ) \, d\mu (\theta)\right\},
\end{equation}

\noindent then $(I^*_{\theta} , p^*_{\theta})_{\theta \in \Theta} $ is incentive efficient.
\end{theorem}

\begin{proof}
The proof can be found in Appendix \ref{App:proofIPOiff}.
\end{proof}

We denote by $W_{\eta, \alpha} \big( (I_{\theta}, p_{\theta})_{\theta \in \Theta} \big)$ the social welfare function that combines the agent's and insurer's aggregate utilities under some welfare weight $\alpha \in (0,1)$ and a probability measure $\eta$:
$$
W_{\eta, \alpha} \big( (I_{\theta}, p_{\theta})_{\theta \in \Theta} \big) = \alpha \int_{\Theta} U_{\theta} (I_{\theta}, p_{\theta} ) \, d\eta(\theta) + (1-\alpha) \int_{\Theta} V_{\theta} (I_{\theta}, p_{\theta} ) \, d\mu(\theta).
$$

Theorem \ref{th:IPOiff} provides a sufficient welfare-maximization condition for incentive Pareto optimality. This is a general result that holds for any well-defined utility functionals $U$ and $V$. In Appendix \ref{AppendixPOsupport}, we provide a partial converse results in the case of Dual Utilities.

\medskip
\subsection{Retention Functions}
In this section, we represent menus of contracts using retention functions, rather than indemnification functions. This reformulation will be used in the remainder of this paper.

\medskip

For $\theta \in \Theta$, the end-of-period wealth of the type-$\theta$ agent is given by
$$
-p_{\theta} - L_{\theta} + I_{\theta}(L_{\theta}) = - p_{\theta} - R_{\theta} (L_{\theta}),
$$

\noindent where $R_{\theta}(L_{\theta}) := L_{\theta} - I_{\theta} (L_{\theta}) \geq 0$ is the loss retained by the type-$\theta$ agent, that is the part of the agent's loss $L_\theta$ that is not covered by the insurer. The insurer's end-of-period wealth after receiving $p_{\theta}$ from the type-$\theta$ agent in exchange for $I_{\theta}(L_{\theta})$, is given by
$$
p_{\theta} - I_{\theta}(L_{\theta}) = p_{\theta} -( L_{\theta} - R_{\theta}(L_{\theta}) ).
$$

\medskip 

\begin{remark}
Since Lipschitz-continuous functions are absolutely continuous, an indemnity function $I$ belongs to $\cI$ if and only if the associated retention function $R$ belongs to the set
$$
\cR = 
\left\{ 
R: [0, \bar L] \rightarrow [0, \bar L] ; \ R(0) =0, \ 0 \leq \frac{\partial R(l)}{\partial l} \leq 1,  \ l \in [0, \bar L], \text{ a.e.}
\right\}.
$$
\end{remark}

Definitions \ref{def:IR}, \ref{def:IC}, and \ref{def:IPO} can be restated using a retention-form menu of contracts $(R_{\theta}, p_{\theta})_{\theta \in \Theta}$, where $(R_\theta,p_\theta)\in\cR\times\bbR_+$ for all $\theta\in\Theta$. Since $I_{\theta} (L_{\theta}) = L_{\theta} - R_{\theta} (L_{\theta})$ for each $\theta \in \Theta$, we can write
\begin{align} 
&U_{\theta} (I_{\theta}, p_{\theta}) = U_{\theta} ( L_{\theta} - R_{\theta} (L_{\theta}) ,p_{\theta}) := \tilde U_{\theta}(R_{\theta}, p_{\theta}), \ \text{and}\label{eq:Utilde} \\ 
&
V_{\theta}(I_{\theta}, p_{\theta}) = V_{\theta} (L_{\theta} - R_{\theta} (L_{\theta}) , p_{\theta}) := \tilde V_{\theta}(R_{\theta}, p_{\theta}).
\label{eq:Vtilde}
\end{align}

\medskip

\begin{remark} \label{re:ItoR}
Theorem \ref{th:IPOiff} can be equivalently stated in terms of menus of contracts of the form $(R_{\theta}, p_{\theta})_{\theta \in \Theta} \in \cI\cR \cap \cI\cC$ using utility functionals $\tilde U_{\theta}(R_{\theta}, p_{\theta})$ and $\tilde V_{\theta}(R_{\theta}, p_{\theta})$ defined in \eqref{eq:Utilde} and \eqref{eq:Vtilde}, respectively. 
\end{remark}

\bigskip
\section{The Case of Dual Utilities} 
\label{sec:DU}

So far, the characterization of incentive efficiency in Theorem \ref{th:IPOiff}, which can be restated using menus of retentions  by Remark \ref{re:ItoR}, has been established for general utility functionals. In this section, we specialize this general framework to Yaari's Dual Utilities, and we provide a crisper characterization of the structure of these efficient menus.

\medskip
\subsection{Dual Utility Framework}
The Dual Utility of \cite{yaari1987dual} is defined as a Choquet integral with respect to a distorted probability.

\medskip

\begin{definition}
For a given random variable $X$, the dual utility of $X$ is given by:
\begin{align*}
DU( X ) 
&=
\int X dg \circ \mathbb{P}  
:= \int_{ - \infty }^0 \bigg( g \big( 1 - \mathbb{P} ( X\leq x ) \big) -1 \bigg) dx + \int_0^{ + \infty } g\big( 1- \mathbb{P}  (X \leq x ) \big) dx,
\end{align*}

\noindent where  $g : [0, 1] \rightarrow [0, 1]$ is a distortion function, that is, an increasing function with $g(0) =0 $ and $g(1)=1$. 
\end{definition}

\medskip

\begin{definition}\label{dominance}
Let $g_1$ and $g_2$ be two distortion functions. We say that $g_1$ dominates $g_2$ if:
$$
g_1(t) \geq g_2(t), \ \text{for all $t \in [0,1]$}.
$$
\end{definition}

\medskip

Unlike Expected Utility Theory, where risk aversion is captured by the curvature of the utility function, in Rank-Dependent Utility (RDU -- e.g.,  \cite{quiggin2012generalized}), both the utility function and the distortion function contribute to risk aversion (e.g., \cite{ChewKarniSafra1987}). Yaari's Dual Utility is a special case of RDU, in which the utility function is linear and risk aversion is captured entirely by the distortion function. Hence, in our setting, strong risk aversion is equivalent to the distortion function $g$ being convex, and weak risk aversion requires $g(x) \leq x$, for all $x \in [0,1]$. See, for instance, \cite{quiggin2012generalized}, \cite{yaari1987dual}, \cite{chateauneuf1994risk}, or \cite{ChewKarniSafra1987}.

\medskip

Moreover, it follows from \cite{quiggin2012generalized} and \cite{ghossoub2021comparative} that weak risk aversion in both RDU and DU can be characterized by the dominance relation between probability weighting functions. Specifically, if $g_1$ and $g_2$ are two distortion functions with associated Dual Utilities ${DU}_1$ and ${DU}_2$, and if $g_1$ dominates $g_2$ in the sense of Definition \ref{dominance}, then ${DU}_2$ is weakly more risk averse than ${DU}_1$. 

\medskip

In this section, we make the following assumptions on the utility functionals.

\medskip 

\begin{assumption}
For each $\theta \in \Theta$, the type-$\theta$ agent has preferences that admit a representation in terms of a Yaari Dual Utility:
$$
DU_{\theta}(\cdot) = \int \cdot \,\,d \ g_{\theta}\circ \mathbb{P}, \ \text{where $g_{\theta}$ denotes the type $\theta$'s distortion function.} 
$$

\noindent Similarly, the monopolistic insurer's preferences admit a representation in terms of the following Yaari Dual Utility functional:
$$
DU^{In}(\cdot) = \int \cdot \ d\ g^{In}\circ \mathbb{P}, \ \text{where $g^{In}$ denotes the insurer's distortion function.} 
$$
\end{assumption}

\medskip

\begin{assumption} \label{ass:distortion_insurer_agent}
For each $t \in [0,1]$, and for all $\theta \in \Theta$, 
$$
g^{In} (t) \geq g_{\theta}(t).
$$
\end{assumption}

Assumption \ref{ass:distortion_insurer_agent} states that the insurer's distortion function $g^{In}$ dominates each type's distortion function $g_{\theta}$ for all $\theta \in \Theta$. Consequently, $DU_{\theta}$ for each type-$\theta$ agent is weakly more risk averse than $DU^{In}$ of the monopolistic insurer. 

\medskip

Dual utilities are translation invariant, meaning that for any random variable $X$ and any constant $c$, $DU(X+c) = DU(X) + c$. The end-of-period utility of a type-$\theta$ agent is therefore given by:
$$
U_{\theta} (R_{\theta}, p_{\theta}) 
= {DU}_{\theta} (-p_{\theta} - R_{\theta}(L_{\theta}) ) 
= -p_{\theta} + {DU}_{\theta} ( -R_{\theta} (L_{\theta})).
$$

\noindent Since $- R_{\theta}(L_{\theta}) \leq0$,
\begin{align*}
U_{\theta} (R_{\theta}, p_{\theta})
&= - p_{\theta} + \int_{-\infty}^0 \left[ g_{\theta} \big(1- \mathbb{P}(-R_{\theta}(L_{\theta}) \leq x) \big) -1\right] \, dx \\
&=  - p_{\theta} - \int_0^{+ \infty} \left[1- g_{\theta} \big(\mathbb{P}(R_{\theta}(L_{\theta}) \leq l) \big)\right] \, dl \\
&=- p_{\theta} - \int_0^{ \bar L}  \left[ 1 - g_{\theta} \big( F_{\theta}(l)  \big)  \right] \, \frac{ \partial R_{\theta}(l) }{ \partial l } \, dl,
\end{align*}

\noindent where $F_\theta(l) := \mathbb{P}(L_\theta \leq l)$ denotes the cumulative loss distribution function, for a given $\theta \in \Theta$. 

\medskip

In the case of no insurance, the dual utility of a type-$\theta$ agent is given by:
$$
U_{\theta} (L_{\theta}, 0) = DU_{\theta} (-L_{\theta}) =  - \int_0^{ \bar L }  \left [ 1 - g_{\theta}  \left( F_{\theta}(l) \right) \right] \ dl.
$$

\medskip

Additionally, by translation invariance, the insurer's end-of-period utility from providing a contract $(R_{\theta}, p_{\theta})$ to the type-$\theta$ agent is given by 
$$
V_{\theta} (R_{\theta}, p_{\theta})
= {DU}^{In}(p_{\theta} -L_{\theta} + R_{\theta}(L_{\theta}) ) 
= p_{\theta} + {DU}^{In} ( -L_{\theta} + R_{\theta}(L_{\theta}) ),
$$

\noindent Since $-L_\theta + R_\theta (L_\theta) = - I_\theta (L_\theta) \leq 0$, we obtain
\begin{align*}
V_{\theta} (R_{\theta}, p_{\theta})
&= p_{\theta} - \int_0^{+ \infty} \left[1 - g^{In} (\mathbb{P} (L_{\theta} - R_{\theta}(L_{\theta}) \leq l ))\right] \, dl \\
&= p_{\theta} - \int_0^{I_{\theta}(\bar L_{\theta})} \left[1 - g^{In} (\mathbb{P} (I_{\theta}(L_{\theta}) \leq l ))\right] \, dl \\
&= p_{\theta} - \int_0^{\bar L} \left[ 1 - g^{In} (F_{\theta}(l))\right] \left( 1 - \frac{\partial R_{\theta}(l)}{\partial l }\right) \, dl.
\end{align*}

\medskip

\begin{remark}
The insurer does not observe the agent's realized type and therefore does not know which loss distribution and distortion function apply to that agent. However, the family $\{F_\theta,g_\theta\}_{\theta\in\Theta}$ and the prior distribution of types are assumed to be common knowledge.
\end{remark}

\medskip

\begin{proposition} \label{IR_characterization}
A menu of contracts $(R_{\theta}, p_{\theta})_{\theta \in \Theta}$ is individually rational if and only if it satisfies:
\begin{enumerate}
\item $\displaystyle\int_\Theta V_\theta (R_\theta, p_\theta) \,d\mu(\theta) \geq0$; and,
\medskip
\item $p_{\theta}  \leq \displaystyle\int_0 ^{ \bar L} \left[  1- g_{\theta}( F_{\theta}(l) )  \right] \left[ 1- \frac{ \partial R_{\theta} (l) }{\partial l} \right] \, dl$, for all $\theta \in \Theta$.
\end{enumerate}
\end{proposition}

\medskip

Proposition \ref{IR_characterization} is an immediate implication of the definition of individual rationality. Specifically, a menu of contracts is individually rational if and only if the insurer's aggregate utility is non-negative and the associated premium does not exceed a certain upper bound at which the agent is indifferent.

\medskip
\subsection{Type Ordering under Dual Utility}\label{sub:typeordering}
In this setting, the agent's type affects their loss distribution and risk attitude. 
The following assumptions impose an ordering on the type space $\Theta$. 

\medskip

\begin{assumption}\label{Ass:cdf_family}
Let $L_{\theta}$ be the loss faced by a type-$\theta$ agent, with cumulative distribution function $F_\theta$. 
\begin{enumerate}
\item  The family of cumulative distribution functions $\{F_{\theta}\}_{\theta \in \Theta}$ is uniformly Lipschitz continuous in $\theta$, with common Lipschitz constant $c' < + \infty$. 

\medskip

\item \label{ass:F_i2}Type-dependent losses $L_{\theta}$ are ordered in the sense of first order stochastic dominance. Specifically, for $\theta_1 < \theta_2$, we have $L_{\theta_1} \preccurlyeq _{FOSD}   L_{\theta_2}$, that is, $F_{\theta_1}(l) \geq F_{\theta_2}(l)$, for all $l \in [0, \bar L]$. Equivalently, 
$$
\frac{\partial F_{\theta} (l) }{ \partial \theta} \leq 0, \ \forall \, l.
$$
\end{enumerate}
\end{assumption}

\noindent Assumption \ref{Ass:cdf_family}-\eqref{ass:F_i2} states that larger types face a stochastically larger loss, in the sense of first-order dominance.

\medskip

\begin{assumption}\label{Ass:distortion_family}
We assume that the following holds:
\smallskip
\begin{enumerate}
\item  $\{ g_{\theta} (t) \}_{\theta \in \Theta}$ is uniformly Lipschitz continuous in $t \in [0,1]$ with common Lipschitz constant $ \delta < + \infty$. That is, for each $\theta \in \Theta$, 
$$
g ^\prime _\theta (t) \leq \delta, \ \forall \, t \in [0,1].
$$

\medskip

\item $\{ g_{\theta} \}_{\theta \in \Theta}$ is uniformly Lipschitz continuous in $\theta$, with common Lipschitz constant $c < + \infty$. 

\medskip

\item \label{ass:g_i3} The type space $\Theta$ is ordered such that:
$$
\frac{\partial g_{\theta} (t) }{\partial \theta} \leq 0,\ \text{ for $t \in (0,1)$. }
$$
\end{enumerate}
\end{assumption}

\noindent Assumption \ref{Ass:distortion_family}-\eqref{ass:g_i3} states that the distortion function $g_{\theta}$ is pointwise smaller for larger values of $\theta$. If $\theta_1, \theta_2 \in \Theta$ are such that $\theta_1 \leq \theta_2$, then $g_{\theta_1}(t) \geq g_{\theta_2}(t)$ for $t \in [0,1]$. This means that the type $\theta_2$-agent is weakly more risk averse than the type $\theta_1$-agent.  

\medskip

Assumption \ref{Ass:cdf_family}-\eqref{ass:F_i2} and Assumption \ref{Ass:distortion_family}-\eqref{ass:g_i3} state that the type space $\Theta$ is ordered such that higher types (larger values of $\theta$) are more (weakly) risk averse and face stochastically larger losses.

\medskip

\begin{remark} \label{Re:chain_rule}
Note that $g_{\theta} \big(  F_{\theta} (l)  \big) $ can be written as the composed function $ \big(  g_{\theta} \circ F_{\theta}  \big) (l)$, for $\theta \in \Theta$ and $l \in [0, \bar L_{\theta}] \subseteq [0, \bar L]$. Hence,
\begin{align*}
\frac{\partial }{ \partial \theta } \left [ g_{\theta} \big(  F_{\theta} (l)  \big)  \right ] 
=
\frac{\partial  g_{\theta} }{ \partial \theta } \big(  F_{\theta} (l)  \big) 
+ g^{\prime}_{\theta}\big(  F_{\theta} (l)  \big)  \frac{\partial  F_{\theta}(l) }{ \partial \theta },
\end{align*}

\noindent where $g^{\prime}_{\theta}\big(  F_{\theta} (l)  \big)  := \frac{ \partial g_{\theta} }{ \partial t } \big(  F_{\theta} (l)  \big)  \bigg|_{t = F_{\theta}(l) } $. 
\end{remark}

The composed function $g_{\theta} \circ F_{\theta}$ is monotone in $\theta$ for all $l$:
$$
\frac{\partial }{ \partial \theta } \left [ g_{\theta} \big(  F_{\theta} (l)  \big)  \right ]  \leq 0,
$$
which follows from Assumption \ref{Ass:cdf_family} and Assumption \ref{Ass:distortion_family}, and since $g_{\theta}(t)$ is increasing in $t$ for all $\theta$. Hence, as $\theta$ increases, the composition $g_{\theta} \big( F_{\theta}(l)\big)$ decreases, for all $l$. If $\theta_1, \theta_2 \in \Theta$ are such that $\theta_1 \leq \theta_2$, then $ g_{\theta _1} \big( F_{\theta_1}(l)\big) \geq g_{\theta_2} \big( F_{\theta_2}(l)\big)$, for all $l$. In other words, higher types, who are more risk averse, assign lower distorted cumulative distribution functions to the loss, meaning that they distort their own perceived loss distributions more pessimistically.

\medskip
\subsection{Solution Characterization Under Dual Utility}
\label{sec:solutionIPO}

It follows from Theorem \ref{th:IPOiff} and Remark \ref{re:ItoR} that if $\eta$ is a probability measure on $(\Theta, \cB(\Theta))$ equivalent to $\mu$, and $\alpha \in (0,1)$, then solutions $(R^*_{\theta}, p^*_{\theta})_{\theta \in \Theta}$ to the problem
\begin{equation} \label{eq:IPO_sup_retention}
\underset{(R_{\theta}, p_{\theta})_{\theta \in \Theta} \in \cI\cR \cap \cI\cC }  { \sup} \left \{ \alpha \int_{\Theta}  U_{\theta} (R_{\theta}, p_{\theta} ) \, d\eta(\theta) + (1-\alpha) \int_{\Theta} V_{\theta} (R_{\theta}, p_{\theta} ) \, d\mu(\theta) \right\}
\end{equation}

\noindent are incentive efficient. Throughout the remainder of the paper, the feasible set is understood to consist of admissible menus with nonnegative premia, as in Definition \ref{def:contract_menu}.

\medskip

Here, we aim to characterize the optimal solutions to Problem \eqref{eq:IPO_sup_retention}, under Yaari's Dual Utility framework. We start by presenting preliminary results about individual rationality and incentive compatibility. The proofs of all results are provided in Appendix \ref{Appendixproofs}.

\medskip

\begin{proposition}\label{IC_characterization}
If a menu of contracts $(R_{\theta}, p_{\theta})_{\theta \in \Theta}$ is incentive compatible, then for any $\theta \in \Theta$, the premium $p_{\theta}$ is of the following form:
\begin{align} \label{eq:premium}
p_{\theta} \ 
&=
p_{\underline{\theta} } + \int_0^{ \bar L}  \left[ 1 - g_{ \underline{ \theta }} \big( F_{\underline{ \theta} } (l)  \big)   \right] \,\, \frac{ \partial R_{\underline{ \theta }}(l) }{ \partial l } \, dl  - \int_{\underline{\theta}} ^ {\theta} \int_0 ^{\bar L}  \left[ \frac{\partial g_s}{ \partial s} (F_s(l)) + g'_s(F_s(l)) \frac{\partial F_s(l)}{\partial s} \right] \frac{\partial R_s(l)}{\partial l } \, dl \,  ds  \nonumber \\
& \quad -
\int_0^{ \bar L }  \left[ 1 - g_{\theta} \big( F_{\theta}(l)  \big)   \right] \,\, \frac{ \partial R_{\theta}(l) }{ \partial l } \, dl. 
\end{align}
\end{proposition}

\medskip

\begin{definition} \label{def:submodular}
A collection of retention functions $\{R_{\theta}\}_{\theta \in \Theta}$ is submodular if $\frac{\partial R_{\theta} (l) }{\partial l } $ is non-increasing in $\theta$, for all $l \in [0, \bar L]$.
\end{definition}

As the agent's type increases (representing more risk aversion), higher types are willing to pay higher premia for more coverage than less risk averse types are unwilling to pay. 
Submodularity of retention functions provides a natural alignment between the agent's risk attitude and the structure of coverage. 
Definition \ref{def:submodular} says that if $\theta_1, \theta_2 \in \Theta$ are such that $\theta_1 \leq \theta_2$, then 
$$ \frac{ \partial R_{\theta_1} (l)}{ \partial l} \geq \frac{ \partial R_{\theta_2} (l)}{ \partial l}, \ \forall l.$$

\noindent This ensures that coverage becomes progressively more generous as risk aversion increases, higher types receive greater coverage, transferring a larger portion of loss to the insurer and retaining less to themselves. Consequently, each type pays a premium consistent with their own preferences.  

\medskip

\begin{proposition} \label{prop:IC_iff}
Consider a submodular collection of retention functions $\{ R_{\theta}\}_{\theta \in \Theta}$. Then a menu of contracts $(R_{\theta}, p_{\theta})_{\theta \in \Theta}$ is in $\cI \cC$ if and only if $\{p_{\theta}\}_{\theta \in \Theta}$ satisfies \eqref{eq:premium}.
\end{proposition}

The following proposition shows that an incentive compatible menu is individually rational if and only if the contract offered to the lowest type is individually rational.

\begin{proposition} \label{prop:IR_lowest_type}
If $(R_{\theta}, p_{\theta})_{\theta \in \Theta} \in \cI \cC$ is such that $\int_\Theta V_\theta (R_\theta, p_\theta)\, d\mu(\theta) \geq 0$, then $(R_{\theta}, p_{\theta})_{\theta \in \Theta} \in \cI \cR$ if and only if, for the lowest type, $(R_{\underline {\theta} } , p_{ \underline{ \theta } } )$ satisfies the agent's participation (P1) of Definition \ref{def:IR}.
\end{proposition}

\medskip

\begin{corollary} \label{cor:IRandIC_iff}
Assume that the collection of retention functions $\{ R_{\theta}\}_{\theta \in \Theta}$ is submodular. Then $(R_{\theta}, p_{\theta})_{\theta \in \Theta} \in \cI \cR \cap \cI\cC $ if and only if both of the following conditions hold:
\begin{enumerate}
\item The premia $\{ p_{\theta}\}_{\theta \in \Theta}$ satisfy \eqref{eq:premium}, with
$$
p_{\underline{\theta}}
\leq
\int_0 ^{ \bar L } \left[ \,\, 1- g_{ \underline{ \theta}}( F_{ \underline{ \theta} }(l) ) \,\, \right] \left[ 1- \frac{ \partial R_{ \underline{ \theta } } (l) }{\partial l} \right] \,\, dl .
$$

\smallskip

\item The insurer's participation (P2) of Definition \ref{def:IR} is satisfied. That is,
$$
\int_\Theta V_\theta (R_\theta, p_\theta)\, d\mu(\theta) \geq 0. 
$$
\end{enumerate}
\end{corollary}

Let $Q_{\eta}$ and $\bar Q_{\eta}$ denote respectively the cumulative and decumulative distribution functions over types induced by the probability measure $\eta$. That is,
$$
Q_{\eta}(\theta) := \eta ([ \underline{\theta}, \theta ]),  \ \text{and} \ \bar Q_{\eta}(\theta) := \eta ( ( \theta, \bar \theta ]), \ \forall \,\theta \in \Theta .
$$

\medskip

\begin{assumption} \label{ass:eta2}
The densities $q$ and $q_\eta$ are continuous on $\Theta$, and $q$ is strictly positive on $\Theta$.
Moreover, for all $\theta \in [\underline\theta, \bar\theta)$, we have 
$$
\frac{q(\theta)}{\bar Q(\theta)} \geq \frac{ q_{\eta} (\theta) }{ \bar Q_{\eta} (\theta) }.
$$
\end{assumption}

Assumption \ref{ass:eta2} states that the hazard rate over types under $\mu$ is greater than or equal to the hazard rate over types under $\eta$. 
Consequently, the distribution over types under $\mu$ is smaller in the hazard rate order than the one under $\eta$. Moreover, we can show that Assumption \ref{ass:eta2} implies that $\frac{ \bar Q_{\eta} (\theta) }{\bar Q(\theta)}$ is non-decreasing in $\theta$. Indeed,  
$$
\left( \frac{ \bar Q_{\eta} (\theta) }{\bar Q(\theta)} \right) ^{\prime} = \frac{- q_{\eta} (\theta) \bar Q(\theta) + q(\theta) \bar Q_{\eta} (\theta)}{\bar Q (\theta) ^2} = \frac{\bar Q_{\eta}(\theta) }{ \bar Q(\theta) } \cdot \left[ \frac{q(\theta)}{\bar Q(\theta)} - \frac{q_{\eta} (\theta) }{ \bar Q_{\eta}(\theta)}\right] \geq 0 \,, \,\, \forall \, \theta \in [\underline\theta, \bar\theta).
$$

\medskip

\begin{remark} \label{re:Q_eta}
Note that, at $\theta = \underline{\theta}$,
\begin{equation} \label{eq:frac_at_underbar_theta}
\frac{ \bar Q_{\eta} (\theta) }{\bar Q(\theta)} \bigg|_{\theta = \underline{\theta}} = 1 .
\end{equation}

\noindent Moreover, at $\theta = \bar \theta$, and L'Hospital's rule gives
\begin{equation} \label{eq:frac_at_bar_theta}
\frac{ \bar Q_{\eta} (\theta) }{\bar Q(\theta)} \bigg|_{\theta = \bar \theta }  = \underset{ \theta \to \bar \theta } {\lim} \frac{ \bar Q_{\eta} (\theta) }{\bar Q(\theta)} =  \underset{ \theta \to \bar \theta } {\lim} \frac{ - q_{\eta} (\theta) }{-q(\theta)}   = \frac{ q_{\eta} ( \bar \theta) }{ q( \bar \theta) }\,.
\end{equation}

\noindent Since $ \frac{ \bar Q_{\eta} (\theta) }{\bar Q(\theta)} $ is non-decreasing in $\theta$, it then follows that
\begin{equation}\label{eq:frac>1}
\frac{ q_{\eta} (\bar\theta) }{ q(\bar \theta)}  \geq 1. 
\end{equation}
\end{remark}

\medskip

\begin{proposition}\label{prop:social_welfare_expression}
Consider a menu $(R_\theta, p_\theta)_{\theta \in \Theta} \in \cI\cR\cap \cI\cC$, and suppose that Assumption \ref{ass:eta2} holds. Then the social welfare function is given by:
\begin{align}\label{eq:social_welfare}
W_{\eta, \alpha} \left( (R_{\theta}, p_{\theta})_{\theta \in \Theta} \right) 
&= (1-2\alpha)  \left[ p_{\underline{\theta}} +  \int_0^{ \bar L}  \left[ 1 - g_{ \underline{ \theta }} \big( F_{\underline{ \theta} } (l)  \big)   \right] \, \frac{ \partial R_{\underline{ \theta }}(l) }{ \partial l } \, dl \right] 
-
\int_{\Theta} \int_0^{\bar L} J_{\theta, \eta} (l) \frac{ \partial R_{\theta}(l) }{ \partial l } \, dl \,d\mu(\theta) \nonumber \\
&\qquad- 
(1-\alpha) \int_{\Theta} \int_0^{\bar L} \left[ 1 - g^{In}(F_{\theta}(l))  \right] \,dl \, d\mu(\theta),
\end{align}

\noindent where 
\begin{align}
\label{eq:Jeta}
J_{\theta, \eta}(l) 
&:= (1-\alpha) \left[  g^{In} (F_{\theta}(l))-  g_{\theta} (F_{\theta}(l)) \right] \nonumber\\
&\quad+
\left(\frac{\bar Q(\theta) }{ q (\theta) } \right) \left[ \frac{\partial g_{\theta}}{ \partial \theta} (F_{\theta}(l)) + g'_{\theta}(F_{\theta}(l)) \frac{\partial F_{\theta}(l)}{\partial \theta} \right]  \left[ (1-\alpha) - \alpha \frac{ \bar Q_{\eta}(\theta)}{\bar Q(\theta)}\right] .
\end{align}
\end{proposition}
\begin{proof}
The proof is provided in Appendix \ref{App:prop:social_welfare_expression}.
\end{proof}

\medskip
\begin{lemma}\label{le:alpha0}
Suppose that Assumption \ref{ass:eta2} holds, and let
$$
\alpha_0 := \frac{q (\bar \theta)}{q_{\eta}(\bar \theta) + q( \bar \theta)}.
$$
Then $ 0 < \alpha_0 \leq \frac{1}{2}$, and the following hold.

\smallskip
\begin{enumerate}
\item[(i)] If $\alpha \in (0, \alpha_0)$, then 
$$ 1-\alpha -  \alpha \frac{\bar Q_{\eta}(\theta)}{ \bar Q(\theta)} >0, \ \forall \, \theta \in \Theta. $$

\smallskip
\item[(ii)] If $\alpha \in [\alpha_0, \frac{1}{2}]$, then there exists $\theta_\alpha \in \Theta$ such that: 
$$\frac{\bar Q_\eta(\theta_\alpha)}{\bar Q(\theta_\alpha)} = \frac{1-\alpha}{\alpha}.$$
\end{enumerate}
\end{lemma}

\begin{proof}
The proof is provided in Appendix \ref{App:le:alpha0}. 
\end{proof}

The following result provides a characterization of optimal solutions to Problem \eqref{eq:IPO_sup_retention}.

\medskip

\begin{theorem} \label{th:solution_characterization_IPO}
Suppose that Assumption \ref{ass:eta2} holds and let $\alpha_0$ be defined as in Lemma \ref{le:alpha0}. Consider the following three cases of social weight $\alpha$. 

\medskip

\begin{enumerate}
\item \label{alpha1} The case where $\alpha \in \left(0, \alpha_0\right)$. Define a marginal retention $r^*_\theta(l)$ by 
\begin{equation} 
\label{eq:optimal_marginal_retention_Sec:PO}
r^*_\theta(l) := 
\begin{cases}
0 & J_{\theta, \eta}(l) >0 , \\
\in [0,1] & J_{\theta, \eta}(l) = 0 , \\
1 & J_{\theta, \eta}(l) < 0 ,	
\end{cases}
\end{equation}

\noindent where the function $J_{ \theta, \eta }(l)$ is given in \eqref{eq:Jeta}.
Define  $R^*_\theta$ by $R^*_\theta(l): =\int_0^l r^*_\theta(s) \,ds$, and define the premium schedule $p^*_\theta$ by
\begin{align}\label{eq:optimal_premium}
p^*_{\theta} 
&:= \int_0^{ \bar L}  \left[ 1 - g_{ \underline{\theta}  } \big( F_{\underline{\theta} }(l)  \big)   \right] \, dl 
- \int_{\underline{\theta}} ^ {\theta} \int_0 ^{\bar L}  \left[ \frac{ \partial g_s}{\partial s } (F_s(l)) + g'_s(F_s(l)) \frac{\partial F_s(l)}{\partial s} \right] \frac{\partial R^*_s(l)}{\partial l } \ dl \, ds  \nonumber  \\
&\quad -
\int_0^{ \bar L }  \left[ 1 - g_{\theta} \big( F_{\theta}(l)\big) \right] \, \frac{ \partial R^*_{\theta}(l) }{ \partial l } \, dl.
\end{align}

\noindent Suppose that the following hold:
\begin{enumerate}
\item[(i)] $J_{\theta,\eta}(l)$ is measurable and non-decreasing in $\theta$, for each $l$.
\item[(ii)] The value of $r^*_\theta(l)$ on the set $\{J_{\theta,\eta}(l)=0\}$ is chosen so that $r^*_\theta(l)$ is measurable and non-increasing in $\theta$.
\item[(iii)] $\int_\Theta V_\theta(R^*_\theta,p^*_\theta)\,d\mu(\theta) \geq 0.$
\end{enumerate}

\medskip

Then for a given $\eta$ and $\alpha$, the collection $\{R^*_\theta\}_{\theta \in \Theta}$ is submodular, and $\{p^*_\theta\}_{\theta \in \Theta}$ is non-negative. 
Moreover,  $(R^*_\theta,p^*_\theta)_{\theta\in\Theta}$ is an optimal solution for Problem \eqref{eq:IPO_sup_retention}.

\medskip 

\item\label{alpha2} The case where $\alpha \in \left[ \alpha_0, \frac{1}{2} \right]$.
There exists $\theta_\alpha\in\Theta$ satisfying $\frac{\bar Q_{\eta} (\theta_{\alpha}) }{\bar Q (\theta_{\alpha}) } = \frac{1-\alpha}{\alpha}$.
\begin{enumerate}
\item[(a)] For $\theta < \theta_{\alpha}$, let $R^*_\theta(l) :=\int_0^l r^*_\theta(s) \,ds$, where $r^*_\theta(l)$ satisfies \eqref{eq:optimal_marginal_retention_Sec:PO} and $J_{\theta, \eta}(l)$ is given by \eqref{eq:Jeta}.

\medskip

\item[(b)] For every $\theta\geq \theta_\alpha$, define $R^*_\theta$ by $ R^*_{ \theta}(l) =0$ for all $l \in [0, \bar L]$.
\end{enumerate}

\smallskip
Define the premium schedule $p^*_\theta$ by \eqref{eq:optimal_premium}. Suppose that the following two conditions hold for $\theta<\theta_\alpha$:
\begin{enumerate}
\item [(i)] $J_{\theta,\eta}(l)$ is measurable and non-decreasing in $\theta$, for each $l$.
\item[(ii)] The value of $r^*_\theta(l)$ on the set $\{J_{\theta,\eta}(l)=0\}$ is chosen so that $r^*_\theta(l)$ is measurable and non-increasing in $\theta$.
\end{enumerate}

\medskip

\noindent Then for a given $\eta$ and $\alpha$, the collection $\{R^*_\theta\}_{\theta \in \Theta}$ is submodular, and the collection $\{p^*_\theta\}_{\theta \in \Theta}$ is non-negative. Moreover, if $\int_\Theta V_\theta(R^*_\theta,p^*_\theta)\,d\mu(\theta) \geq 0$, then $(R^*_\theta, p^*_\theta)_{\theta \in \Theta}$ is an optimal solution to Problem \eqref{eq:IPO_sup_retention}. 
\medskip

\item\label{alpha3} The case where $\alpha\in\left(\frac{1}{2},1\right)$.
Suppose that the following condition holds:
\begin{equation}
\label{eq_case3_condition}
\int_\Theta\int_0^{\bar L}
\left[
1-g^{In}(F_\vartheta(l))
\right]\,dl\,d\mu(\vartheta)
\leq
\int_0^{\bar L} \left[1 - g_{\underline\theta}(F_{\underline\theta}(l))\right]  \, dl.
\end{equation}

\noindent Then there exists an optimal solution
$(R_\theta^*,p_\theta^*)_{\theta\in\Theta}$ to Problem \eqref{eq:IPO_sup_retention}, such that
$R_\theta^*(l)=0, \, \forall l\in[0,\bar L], \, \forall\theta\in\Theta$,
and 
\begin{equation}
\label{eq:high_alpha_binding_premium_theorem}
p^*_\theta = \int_\Theta\int_0^{\bar L}
\left[
1-g^{In}(F_\vartheta(l))
\right]\,dl\,d\mu(\vartheta) \geq 0,
\ \forall \, \theta\in\Theta.
\end{equation}

\noindent This is the unique optimal full-coverage pooling menu, and it binds the insurer's participation constraint.
\end{enumerate}
\end{theorem}

\begin{proof}
The proof is provided in Appendix \ref{App:thIPOsolution}.
\end{proof}

\medskip

Theorem \ref{th:solution_characterization_IPO} shows that the structure of incentive-efficient menus depends on the value of the social weight $\alpha$. 
Specifically, when $\alpha \leq \frac{1}{2}$ and $\theta < \theta_\alpha$, that is when $\alpha$ satisfies cases \eqref{alpha1} and \eqref{alpha2}-(a), a separating equilibrium emerges. The optimal retention function is submodular and is characterized by the marginal retention in \eqref{eq:optimal_marginal_retention_Sec:PO} ensuring that higher types receive more coverage.
If $\alpha$ satisfies case \eqref{alpha2}-(b), full coverage is offered to agent types $\theta \geq \theta_{\alpha}$.
Moreover, for both cases \eqref{alpha1} and \eqref{alpha2}, the optimal premia satisfy \eqref{eq:optimal_premium} and are non-decreasing in $\theta$. That is, optimal coverage is more expensive for higher types, who are more risk averse and face stochastically larger losses. 
Finally, when $\alpha$ satisfies case \eqref{alpha3}, that is, if a high social weight is placed on the agent's welfare, then the optimal menu is a full-coverage pooling menu. 
Moreover, the quantity 
$$\int_\Theta\int_0^{\bar L}\left[1-g^{In}(F_\vartheta(l))\right]\,dl\,d\mu(\vartheta)$$ 
is the insurer's aggregate certainty-equivalent cost of providing full coverage, averaged across types; whereas the quantity 
$$\int_0^{\bar L}\left[1-g_{\underline\theta}(F_{\underline\theta}(l))\right]\,dl$$
is the maximum premium that the lowest type is willing to pay for full coverage, i.e., the lowest type's reservation premium. Condition \eqref{eq_case3_condition} means that the insurer's break-even premium for full coverage is affordable for the lowest type. That is, the minimum premium needed to make the insurer willing to provide full coverage is not larger than the maximum premium the lowest type is willing to pay for full coverage. This is essentially a market-making condition under which the common optimal premium $p^*$ across types binds the insurer's participation constraint, and is equal to the insurer's average cost of providing full coverage.

\medskip

For cases \eqref{alpha1} and \eqref{alpha2} of Theorem \ref{th:solution_characterization_IPO}, the insurer's participation constraint must be satisfied so that the optimal characterized menu $(R^*_\theta, p^*_\theta)_{\theta \in \Theta}$ is individually rational. 
The following lemma provides sufficient conditions ensuring the satisfaction of the insurer's participation constraint.

\medskip

\begin{lemma}
\label{le:IRimplication}
Let $\alpha \in (0, \frac{1}{2}]$ and let $(R^*_\theta, p^*_\theta)_{\theta \in \Theta}$ be a solution characterized in Theorem \ref{th:solution_characterization_IPO}, whose premia are given by \eqref{eq:optimal_premium}. 

\begin{enumerate}
\item If $r^*_\theta(l)\in [0,1]$ for a.e. $l \in [0, \bar L]$, then the insurer's participation constraint is satisfied if the following condition holds
\begin{align}\label{eq:condition_IR_partialcoverage}
\int_\Theta \int_0^{\bar L} \left[ 1 - g^{In}(F_\theta(l))\right] \left( 1- r^*_\theta(l) \right) \,dl \, d\mu(\theta) 
&\leq \int_\Theta p^*_\theta \, d\mu(\theta).
\end{align}

\medskip

\item If $r^*_\theta \equiv 0$ and condition \eqref{eq_case3_condition} holds, then the insurer's participation constraint is satisfied. 

\medskip
\item If $r^*_\theta \equiv 1$, then $ V_\theta (R^*_\theta , p^*_\theta)= 0$, for all $\theta \in \Theta$. Hence the insurer's participation constraint binds. That is, $\displaystyle\int_\Theta V_\theta (R^*_\theta, p^*_\theta) \, d\mu(\theta) = 0$.
\end{enumerate}
\end{lemma}

\begin{proof}
The proof can be found in Appendix \ref{App:proofleIRimplication}.
\end{proof}

Lemma \ref{le:IRimplication} discusses the conditions under which the optimal menu $(R^*_\theta, p^*_\theta)_{\theta \in \Theta}$ characterized in Theorem \ref{th:solution_characterization_IPO} satisfies the insurer's participation constraint (P2) of Definition \ref{def:IR}, for $\alpha \in (0, \frac{1}{2}]$ and depending on whether the optimal marginal retention corresponds to full, partial, or zero coverage.
When coverage is partial, condition \eqref{eq:condition_IR_partialcoverage} requires that the insurer's certainty-equivalent cost of providing partial coverage does not exceed the aggregate optimal premium $\int_\Theta p^*_\theta \, d\mu(\theta)$, collected across types. This aggregate optimal premium admits the representation given in Lemma \ref{le:IRimplication}-(1), where the quantity 
$$\int_0^{\bar L} [1 - g_{\underline{\theta}}(F_{\underline{\theta}}(l))] \, dl$$ 
corresponds to the maximum premium that the lowest type is willing to pay for full coverage; whereas the quantity 
$$\int_\Theta\int_0^{\bar L} [1 - g_\theta(F_\theta(l))] \frac{\partial R^*_\theta(l)}{\partial l} \, dl \, d\mu(\theta)$$ 
corresponds to the aggregate dual utility generated by the retained loss borne by the agent under the partial coverage menu. The remaining terms of the aggregate premium compare the utility of the lowest type with the aggregate utility across types under the optimal menu.
In contrast, when full coverage is offered to the agent, the insurer's aggregate certainty-equivalent cost of providing full coverage must not exceed the lowest type's willingness to pay for full coverage, so that the insurer is willing to participate in the market. 
If the agent retains the entire loss, then the insurer earns zero utility and is indifferent between participating and not participating in the market.

\medskip
\subsection{On the Monotonicity of the Function $\theta \mapsto J_{ \theta, \eta}(l)$}
\label{sec:solutionIPOmon}
Consider now cases \eqref{alpha1} and \eqref{alpha2}-(a), where $\alpha \leq \frac{1}{2}$ and $\theta<\theta_\alpha$. 
To achieve the separating layered equilibrium described by the marginal retention functions given in \eqref{eq:optimal_marginal_retention_Sec:PO}, we require the function $J_{\theta, \eta}(l)$ defined in \eqref{eq:Jeta} to be non-decreasing in $\theta$. We examine in this section some sufficient conditions for this monotonicity.

\medskip

\begin{proposition}
Consider cases \eqref{alpha1} or \eqref{alpha2}-(a) of the social weight $\alpha$. The function $J_{\theta,\eta}(l)$ is non-decreasing in $\theta$ for all $l  \in [0, \bar L]$, if the following conditions hold for  $\theta \in [\underline{\theta},\bar \theta)$.
\begin{enumerate}
\item \label{C1} $ 0 \leq \left( \frac{\bar Q(\theta) }{ q(\theta) }  \right) ' \leq 1$;

\medskip

\item \label{C2} The function $ \theta \mapsto F_{\theta}$ is convex in $\theta$ for all $l$, that is, $\frac{ \partial ^2 F_{\theta}(l) }{ \partial \theta^2 } \geq 0$;

\medskip

\item \label{C3} The function $ \theta \mapsto g_{\theta} (t)$ is convex in $\theta$ for all $t$, that is, for $F_{\theta}(l) \in [0,1]$, $ \frac{ \partial ^2 g_{\theta} }{ \partial \theta^2 } (F_{\theta}(l)) \geq 0 $;

\medskip

\item \label{C4} The function $t \mapsto g_{\theta}(t)$ is convex in $t$ for all $\theta \in \Theta$, that is, for $F_{\theta}(l) \in [0,1]$, $ g^{\prime \prime}_{\theta} (F_{\theta}(l)) \geq 0$;

\medskip

\item \label{C5} The function $g: (\theta, t) \mapsto g_\theta(t)$ is submodular such that \begin{equation}\label{eq:gsubmodular}
\frac{\partial ^2 g_{\theta}}{\partial \theta \partial t} \leq 0,
\end{equation} 
and satisfies the following:
$$
2 \,\frac{\bar Q(\theta)}{q(\theta)} \, \frac{\partial ^2 g_{\theta}}{\partial \theta \partial t}(F_{\theta}(l)) \leq - g^{In \prime}(F_{\theta}(l)).
$$ 
\end{enumerate}
\end{proposition}
\begin{proof}
First, note that in this region, $1- \alpha -\alpha \frac{\bar Q_{\eta}(\theta) }{\bar Q (\theta) } \geq 0$. The partial derivative of $J_{\theta, \eta}(l)$ with respect to $\theta$ is given by:
\begin{align*}
&\frac{\partial J_{\theta, \eta}(l)}{\partial \theta} 
= \frac{\partial g_{\theta}}{ \partial \theta} (F_{\theta}(l)) \, \left[ (1-\alpha) \left( \left(  \frac{\bar Q(\theta) }{q(\theta) } \right)^{\prime } - 1 \right)  - \alpha \left( \frac{\bar Q(\theta) }{q(\theta) } \,\frac{\bar Q_{\eta}(\theta) }{\bar Q (\theta) } \right) ^{\prime} \right] \\
&\quad +
(1-\alpha) \frac{\partial F_{\theta}(l)}{\partial \theta} \, \left[ g^{\prime}_{\theta}(F_{\theta}(l))  \left( \left( \frac{\bar Q(\theta) }{ q (\theta) } \right)^{\prime} -  1 \right) \right]
+ (1-\alpha) \frac{\partial F_{\theta}(l)}{\partial \theta} \,   \left[ {\color{blue} 2} \, \frac{\bar Q(\theta) }{ q (\theta) }\frac{ \partial ^2 g _{\theta} }{ \partial \theta \,\, \partial t} (F_{\theta}(l)) +  g^{In \ \prime}(F_{\theta}(l))\right] \\
&\quad -
\alpha \frac{\partial F_{\theta}(l)}{\partial \theta} \, \left[ g^{\prime}_{\theta}(F_{\theta}(l)) \left( \frac{\bar Q(\theta) }{q(\theta) } \,\frac{\bar Q_{\eta}(\theta) }{\bar Q (\theta) } \right) ^{\prime} +{\color{blue} 2}\, \frac{\bar Q_{\eta}(\theta) }{\bar Q (\theta) } \frac{\bar Q(\theta) }{ q (\theta) }\frac{ \partial ^2 g _{\theta} }{ \partial \theta \,\, \partial t} (F_{\theta}(l)) \right] \\
&\quad + 
\left[ 1- \alpha -\alpha \frac{\bar Q_{\eta}(\theta) }{\bar Q (\theta) } \right] \frac{\bar Q(\theta) }{q(\theta) }\, \left[   g^{\prime}_{\theta} (F_{\theta}(l))  \frac{\partial^2 F_{\theta}(l)}{\partial \theta^2}    +
\frac{ \partial ^2 g_{\theta} }{ \partial \theta^2 } (F_{\theta}(l))  +
g^{\prime \prime} _{\theta} (F_{\theta}(l)) \left( \frac{\partial F_{\theta}(l)}{\partial \theta} \right)^2  \right].
\end{align*}

Using the monotonicity implications of Assumptions \ref{Ass:cdf_family}, \ref{Ass:distortion_family}, and \ref{ass:eta2}, if Conditions \ref{C1} to \ref{C5} hold on $[\underline{\theta}, \bar \theta)$, then $\frac{\partial J_{\theta, \eta}(l)}{\partial \theta} \geq 0$, and hence $\theta \mapsto J_{\theta, \eta}(l)$ is non-decreasing in $\theta$ on $[\underline{\theta}, \bar \theta)$, for all $l \in [0, \bar L]$. Moreover,
$$
J_{\theta, \eta}(l) \bigg|_{\theta = \bar \theta} := \underset{\theta \to \bar \theta}{\lim} J_{\theta, \eta}(l) =
(1-\alpha)
\left[g^{In}(F_{\bar\theta}(l))
-g_{\bar\theta}(F_{\bar\theta}(l))
\right].
$$

Since $\theta \mapsto J_{\theta, \eta}(l)$ is continuous on $\Theta$ and non-decreasing on $[\underline{\theta},\bar \theta)$, then it follows that
$\theta \mapsto J_{\theta, \eta}(l)$ is non-decreasing in $\theta$ on $\Theta$.
\end{proof}

\medskip

The ratio $ \frac{ \bar Q ( \theta) }{ q(\theta)}$ represents the inverse hazard rate and measures how many agent types are left above $\theta$ relative to the density of the type-$\theta$ agent. Condition \ref{C1} implies that the distribution over types is heavy tailed, placing more probability on larger $\theta$ values.  
Moreover, the growth of the inverse hazard rate is bounded and does not increase too quickly as $\theta$ increases. As a result, the population of higher types does not thin out too rapidly.

\medskip

Conditions \ref{C2} and \ref{C3} are respectively saying that as $\theta$ increases, the loss distribution becomes riskier, and higher types become more risk averse at a decreasing rate.
Condition \ref{C4} reflects strong risk aversion through the convexity of the distortion function $g_{\theta}(t)$ in $t$ for all $\theta$.   
Condition \ref{C5} implies that the function $g$ is submodular. 
In other words, higher types assign lower marginal weight to favorable probabilities, distorting them more pessimistically and reflecting greater risk aversion. 

\medskip
\subsection{Properties of the Optimal Menu} \label{sub:properties_optimalmenu}
\begin{proposition}
\label{prop:insurerutilityoptimalmenu}
Suppose that the function $J_{\theta, \eta}(l)$ given in \eqref{eq:Jeta} is non-decreasing in $\theta$, for all $l$. Let $\alpha \in (0, \frac{1}{2}]$ and let $(R^*_\theta, p^*_\theta)_{\theta \in \Theta}$ be an optimal solution to Problem \eqref{eq:IPO_sup_retention}, as characterized in Theorem \ref{th:solution_characterization_IPO}, whose premia satisfy \eqref{eq:optimal_premium}. Then the following hold.
\medskip
\begin{enumerate}
\item If $r^*_\theta(l) \in [0,1]$ for a.e.\ $l \in [0, \bar L]$, then $\theta \mapsto V_\theta (R^*_\theta, p^*_\theta)$ is non-decreasing whenever the following condition holds:
\begin{align}\label{eq:inequality_condition}
&\int_\theta^{\theta'}
\int_0^{\bar L}
g^{In\prime}(F_s(l))
\left(-\frac{\partial F_s(l)}{\partial s}\right)
(1-r_{\theta'}^*(l))
\,dl\,ds\nonumber  \\ &\leq
\int_0^{\bar L}
[g^{In}(F_\theta(l))-g_\theta(F_\theta(l))]
(r_\theta^*(l)-r_{\theta'}^*(l))\,dl+
\int_\theta^{\theta'}
\int_0^{\bar L}
\frac{\partial }{\partial s} (g_s \circ F_s)(l)
(r_{\theta'}^*(l)-r_s^*(l))
\,dl\,ds.
\end{align}

\smallskip 
\item If $r^*_\theta \equiv 0$, then the mapping $\theta \mapsto V_\theta (R^*_\theta, p^*_\theta)$ is non-increasing.

\smallskip

\item If $r^*_\theta \equiv 1$, then $V_\theta (R^*_\theta, p^*_\theta) =0 $. 
\end{enumerate}
\end{proposition}
\begin{proof}
The proof can be found in Appendix \ref{App:proof_Vthetaoptimum}. 
\end{proof}

\medskip

Proposition \ref{prop:insurerutilityoptimalmenu} studies the insurer's utility at the optimum, and shows that the monotonicity of this utility with respect to the agent's type $\theta$ depends on the optimal marginal retention. In particular,  when partial coverage is provided, the insurer has a higher utility at the optimum from higher types, who are more risk averse, if \eqref{eq:inequality_condition} holds. If full coverage is provided to the agent at the optimum, the insurer's utility decreases with the agent's type $\theta$. An implication of this is that for higher types, who are more risk averse and face stochastically larger losses, the insurer's utility is lower than for lower risk types, as intuition would suggest. If, in contrast, zero coverage is provided to the agent, i.e., the agent retains the entire loss, we saw in Lemma \ref{le:IRimplication} that the insurer's utility is equal to zero, and hence the participation constraint binds leaving them indifferent between participating or not in the market. Consequently, the insurer's utility is neither increasing nor decreasing in the agent's type.

\medskip

\begin{proposition}
\label{prop:properties_optimalmenu}
Let $(R^*_\theta,p^*_\theta)_{\theta\in\Theta}$ be an optimal solution to Problem \eqref{eq:IPO_sup_retention}, as characterized in Theorem \ref{th:solution_characterization_IPO} and consider the separating regions of coverage where $\alpha \in(0, \frac{1}{2}]$. Suppose that the sufficient conditions of Lemma \ref{le:IRimplication} ensuring the satisfaction of the insurer's participation constraint at the optimum hold, and that the function $J_{\theta, \eta}(l)$ given in \eqref{eq:Jeta} is non-decreasing in $\theta$, for all $l$. Then the following properties of the optimal menu hold.

\smallskip
\begin{enumerate}
\item\label{Prop1} $R^*_\theta$ is non-increasing with $\theta$ for every $l \in [0 ,\bar L]$, and $p^*_\theta$ is non-decreasing in $\theta$.

\medskip

\item\label{Prop2} For $\theta = \bar \theta$, $R^*_{\bar \theta}(l)=0$ for all $l \in [0,\bar L]$ if the following conditions hold:
\begin{enumerate}
\item[(a)]$F_{\bar\theta}(l)\in(0,1)$, for almost every $l \in (0, \bar L)$, 
\medskip
\item [(b)] $g^{In}(t) > g_{\bar \theta} (t)$, for all $t \in (0, 1)$.
\end{enumerate}

\medskip

\item\label{Prop3} For $\theta = \underline{\theta}$, $U_{\underline{\theta}}(R^*_{\underline{\theta}}, p^*_{\underline{\theta}}) = U_{\underline{\theta}}(L_{\underline{\theta}}, 0)$. 

\medskip

\item\label{Prop4} 
The function $\theta \mapsto U_\theta (R^*_\theta, p^*_\theta)$ is non-increasing. Moreover, it is convex if the following hold:
\begin{enumerate}
\item $g$ is submodular, as in \eqref{eq:gsubmodular};
\medskip
\item $g_\theta(t)$ is convex in $\theta$ for all $t \in [0,1]$;
\medskip
\item $g_\theta(t)$ is convex in $t$ for all $\theta \in \Theta$;
\medskip
\item $F_\theta$ is convex in $\theta$ for all $l \in [0,\bar L]$. 
\end{enumerate}
\end{enumerate}
\end{proposition}

\begin{proof}
The proof can be found in Appendix \ref{App:proofpropoptimalmenu}. 
\end{proof}

\medskip

Proposition \ref{prop:properties_optimalmenu} provides some important properties of the optimal solution to Problem \eqref{eq:IPO_sup_retention} when $\alpha \in (0, \frac{1}{2}]$. Specifically, the optimal retention decreases with the agent's type $\theta$, whereas the optimal premium is non-decreasing with $\theta$ as established in Theorem \ref{th:solution_characterization_IPO}. That is, higher types, who are more risk averse and face stochastically larger losses, receive more coverage for a larger premium. Additionally, the highest risk type (who is the most risk averse and faces the largest loss) receives full coverage at every possible loss level $l$, under additional conditions ensuring strict positivity of the virtual value function at $\bar \theta$. Specifically, the first condition requires the highest type's loss distribution to be nontrivial, while the second requires the insurer to be strictly less risk averse than the highest type agent. This property, commonly referred to as \textit{efficiency at the top}, was shown to hold by \cite{chade2012optimal}, and later by \cite{gershkov2023optimal} and \cite{ghossoub2025optimal}. In contrast, the lowest risk type (the least risk averse) is indifferent in participating at the optimum, and the monopolist insurer absorbs all of the surplus from this agent.  

\medskip

Proposition \ref{prop:properties_optimalmenu}  also shows that the agent's utility at the optimum decreases with the agent's type $\theta$. Moreover, if $g$ is submodular as in \eqref{eq:gsubmodular}, if both the agent's distortion and the cumulative loss distribution are convex in type, and if the agent's distortion is convex in $t$ for all types, then the agent's utility is also convex in type. 
While the monotonicity of the agent’s utility is consistent with the findings of previous literature (e.g., \cite{chade2012optimal},  \cite{gershkov2023optimal}), \cite{ghossoub2025optimal} assume that the agent's type affects only their risk attitude, and they show that the agent's utility decreases with risk type. Moreover, it is convex in type if convexity of the agent's distortion in type is satisfied. 

\medskip
\subsection{Special Cases} \label{sec:specialcases}
We can clearly notice the importance of the social weight $\alpha$ in the characterization of optimal solutions to Problem \eqref{eq:IPO_sup_retention} given in Theorem \ref{th:solution_characterization_IPO}.
Particularly, when $\alpha$ is close to zero, the objective places almost all weight on the insurer's aggregate utility, whereas if $\alpha$ is close to $1$, then in this case, the objective primarily reflects the agent's aggregate utility. 
Hence, we examine the maximization of aggregate utilities of the insurer alone, and the agent alone, respectively. The proofs of all results of this section can be found in Appendix \ref{Appendixproofs}.  

\medskip

\subsubsection{Insurer's Welfare Maximization}
We consider the following problem. 
\begin{equation} 
\label{eq:sup_a0}
\underset{(R_{\theta}, p_{\theta})_{\theta \in \Theta} \in \cI\cR \cap \cI\cC }  { \sup}  \ \int_{\Theta} V_{\theta} (R_{\theta}, p_{\theta} ) \, d\mu(\theta).
\end{equation}

\noindent This reduces Problem \eqref{eq:IPO_sup_retention} to the maximization of the insurer's total utility subject to the incentive compatibility and individual rationality constraints, and no longer defines an IPO solution.

\medskip

The following proposition provides a necessary condition for incentive compatibility, characterizing the insurer's total utility.

\medskip

\begin{proposition}\label{prop:insureraggutility}
If $(R_{\theta}, p_{\theta})_{\theta \in \Theta} \in \cI\cC$, then the monopolistic insurer's total utility is given by:
\begin{align}\label{eq:barV}
\int_{\Theta} V_{\theta}(R_{\theta}, p_{\theta}) \, d\mu(\theta) 
&= p_{\underline{\theta} } + \int_0^{ \bar L}  \left[ 1 - g_{ \underline{ \theta }} \big( F_{\underline{ \theta} } (l)  \big)   \right] \,\, \frac{ \partial R_{\underline{ \theta }}(l) }{ \partial l } \,\, dl - \int_{\Theta} \int_0^{\bar L} \left[ 1 - g^{In} (F_{\theta}(l))\right] \, dl \, d\mu(\theta) \nonumber \\
&\quad -
\int_{\Theta} \int_0^{\bar L} J^I_{\theta}(l)  \frac{\partial R_\theta(l)}{\partial l } \,  dl \, d\mu(\theta),
\end{align}

\noindent where
\begin{equation} 
\label{eq:Jtheta}
J^I_{\theta}(l) = 
\left[   g^{In} (F_{\theta}(l)) - g_{\theta} \big( F_{\theta}(l)  \big) \right] + \frac{\bar Q(\theta)}{q(\theta)} \left[ \frac{\partial g_{\theta}}{ \partial \theta } (F_\theta(l)) + g'_\theta(F_\theta (l))\frac{\partial F_\theta(l)}{\partial \theta} \right] .        
\end{equation}
\end{proposition}

\medskip

\begin{theorem}
\label{th:solution_characterization}
Define a marginal retention function $r_\theta^*$ by
$$
r_\theta^*(l)=
\begin{cases}
0, & J_\theta^I(l)>0,\\
\in[0,1], & J_\theta^I(l)=0,\\
1, & J_\theta^I(l)<0.
\end{cases}
$$

\noindent where the function $J^I_{\theta}(l)$ is given in \eqref{eq:Jtheta}. Define $R_\theta^*$ by
$R_\theta^*(l):=\int_0^l r_\theta^*(s)\,ds$,
and define $p_\theta^*$ by 
\begin{align*}
p^*_{\theta} 
&= \int_0^{ \bar L}  \left[ 1 - g_{ \underline{\theta}  } \big( F_{\underline{\theta} }(l)  \big)   \right] \, dl  - \int_{\underline{\theta}} ^ {\theta} \int_0 ^{\bar L}  \left[ \frac{ \partial g_s}{\partial s } (F_s(l)) + g'_s(F_s(l)) \frac{\partial F_s(l)}{\partial s}  \right] \frac{\partial R^*_s(l)}{\partial l } \, dl \, ds\\
&\quad -
\int_0^{ \bar L }  \left[ 1 - g_{\theta} \big( F_{\theta}(l)  \big)   \right] \, \frac{ \partial R^*_{\theta}(l) }{ \partial l } \, dl.  
\end{align*}

\noindent Suppose that the following conditions hold:
\begin{enumerate}
\item[(i)]$J^I_{\theta}(l)$ is measurable and non-decreasing in $\theta$, for each $l$.
\item [(ii)]The value of $r^*_\theta(l)$ on the set $\{J^I_{\theta}(l)=0\}$ is chosen so that $r^*_\theta(l)$ is measurable and non-increasing in $\theta$.
\end{enumerate}

\smallskip

\noindent Then the collection $\{R^*_{\theta}\}_{\theta \in \Theta} $ is submodular, and $\{p_\theta^*\}_{\theta\in\Theta}$
is non-negative. Moreover, if $\int_\Theta V_\theta (R^*_\theta, p^*_\theta) \, d\mu(\theta) \geq 0$, then $(R^*_\theta, p^*_\theta)_{\theta \in \Theta} \in \cI\cR\cap \cI\cC$ is an optimal solution to Problem \eqref{eq:sup_a0}.
\end{theorem}

Theorem \ref{th:solution_characterization} shows that the optimal menu of contracts consists of a collection of layered retention functions, assuming that $J^I_{\theta}(l)$ is monotone and that the insurer's participation constraint is satisfied. Particularly, the sufficient conditions of Lemma \ref{le:IRimplication} ensure that the insurer's participation constraint holds.
Moreover, when $J^I_{\theta}(l)>0$, full coverage is provided. $J^I_{\theta}(l)<0$ corresponds to no coverage of this loss level, and when $J^I_{\theta}(l)=0$, the optimal retention allows for some flexibility, as long as feasibility is maintained. 
It remains to examine the sufficient conditions that ensure the monotonicity of $J^I_\theta(l)$. First, the partial derivative of $J^I_\theta(l)$ with respect to $\theta$ is given by:
\begin{align*}
\frac{\partial J^I_{\theta}(l) }{\partial \theta} 
&=
\left( \left(  \frac{\bar Q(\theta) }{q(\theta) } \right)^{\prime } - 1 \right) \left[
\frac{\partial g_{\theta}}{ \partial \theta} (F_{\theta}(l)) + g^{\prime}_{\theta}(F_{\theta}(l)) \frac{\partial F_{\theta}(l)}{\partial \theta}  \right] \\
&\quad + \frac{\partial F_{\theta}(l)}{\partial \theta} \, \left[ 2\, \frac{\bar Q(\theta) }{ q (\theta) }\frac{ \partial ^2 g _{\theta} }{ \partial \theta \,\, \partial t} (F_{\theta}(l)) +  g^{In \ \prime}(F_{\theta}(l))\right] \\
&\quad + 
\frac{\bar Q(\theta) }{q(\theta) }\, \left[   g^{\prime}_{\theta} (F_{\theta}(l))  \frac{\partial^2 F_{\theta}(l)}{\partial \theta^2}    +
\frac{ \partial ^2 g_{\theta} }{ \partial \theta^2 } (F_{\theta}(l))  +
g^{\prime \prime} _{\theta} (F_{\theta}(l)) \left( \frac{\partial F_{\theta}(l)}{\partial \theta} \right)^2 \ \right].
\end{align*}

\noindent Conditions \ref{C1} to  \ref{C5} of the monotonicity of $J_{\theta, \eta} (l)$ are also sufficient for $J^I_{\theta}(l)$ to be non-decreasing in $\theta$ for all $l$.

\medskip

\subsubsection{Agent's Welfare Maximization}

Consider the following problem:
\begin{equation} 
\label{eq:sup_a1}
\underset{(R_{\theta}, p_{\theta})_{\theta \in \Theta} \in \cI\cR \cap \cI\cC }  { \sup}  \,\int_{\Theta} U_{\theta} (R_{\theta}, p_{\theta} ) \,d\mu(\theta), 
\end{equation}

\noindent which reduces Problem \eqref{eq:IPO_sup_retention} to the maximization of the agent's aggregate utility subject to the incentive compatibility and individual rationality constraints, and no longer defines an IPO solution.

\medskip
\begin{proposition}
\label{prop:objective_agent}
If $(R_{\theta}, p_{\theta})_{\theta \in \Theta} \in \cI\cC$, then the agent's aggregate utility is given by:
$$
\int_\Theta U_\theta(R_\theta, p_\theta)\,d\mu (\theta)
= 
- \left[ p_{\underline{\theta}} +  \int_0^{ \bar L}  \left[ 1 - g_{ \underline{ \theta }} \big(F_{\underline{ \theta} } (l)\big)\right] \, \frac{ \partial R_{\underline{ \theta }}(l) }{ \partial l }  dl \right] -  \int_{\Theta} \int_0^{\bar L} J^A_\theta(l) \frac{\partial R_{\theta}(l)}{\partial l } \, dl \, d\mu(\theta), 
$$

\noindent where $ J^A_\theta(l) = - \frac{ \bar Q(\theta)}{q(\theta)}\left[ \frac{\partial g_{\theta}}{ \partial \theta} (F_{\theta}(l)) + g'_{\theta}(F_{\theta}(l))\frac{\partial F_{\theta}(l)}{\partial \theta} \right]\geq0$. 
\end{proposition}

\medskip
\begin{proposition}
\label{prop:solution_characterization_a1}
Suppose that condition \eqref{eq_case3_condition} holds. Then there exists an optimal solution $(R^*_\theta,p^*_\theta)_{\theta\in\Theta}\in\cI\cR\cap\cI\cC$ to Problem \eqref{eq:sup_a1} whose retention is obtained by pointwise maximization and satisfies $R^*_\theta(l)=0$ for all $l$ and for all $\theta\in\Theta$, and the common pooling premium is given by:
$$
p^*=\int_\Theta\int_0^{\bar L}\left[1-g^{In}(F_\theta(l))\right]\,dl\,d\mu(\theta)\geq 0.
$$

\noindent Moreover, the insurer's participation constraint binds. 
\end{proposition}

It follows from Proposition \ref{prop:solution_characterization_a1} that $R^*_\theta(l)=0$ for each $\theta \in \Theta$ and $l \in [0, \bar L]$. Particularly, \eqref{eq:sup_a1} describes the maximization of the agent's aggregate utility subject to individual rationality and incentive compatibility. In this case, the agent's welfare is maximized by full insurance, meaning zero retention. However, even in this agent-welfare maximization problem, the premium schedule must still satisfy the insurer's participation constraint because the feasible set is $\cI\cR\cap\cI\cC$. Therefore the premium cannot generally be set equal to zero. If full coverage is optimal, incentive compatibility forces a pooling premium, and if the insurer's participation constraint binds, the common premium must cover the insurer's distorted cost of full coverage.

\medskip
\subsection{Alternative Ordering} \label{sec:ordering2}
The results obtained so far rely on the type ordering assumptions of Subsection \ref{sub:typeordering}, under which higher types are more risk averse and face stochastically larger losses. 
Alternatively, one may consider a setting in which higher types face stochastically larger losses but are less risk averse. We formalize this by the following assumption.

\medskip

\begin{assumption}\label{ass:ordering2}
The distortion function $g_\theta$ is non-decreasing in $\theta$. That is,
$$
\frac{\partial g _{\theta} }{ \partial \theta }(t) \geq 0, \ \forall \, t \in (0,1).
$$

\noindent Moreover, losses $L_\theta$ for $\theta \in \Theta$ are ordered in the first order stochastic dominance sense, such that for $\theta_1, \theta_2 \in \Theta$ with $\theta_1<\theta_2$, 
$
L_{\theta_1} \preccurlyeq _{FOSD}  L_{\theta_2}, \ \text{or equivalently,} \ F_{\theta_1}(l) \geq F_{\theta_2}(l).
$
That is, 
$$
\frac{\partial F_\theta (l)}{\partial \theta} \leq 0, \ \forall \, l \in [0, \bar L]. 
$$
\end{assumption}

\medskip

\begin{assumption}\label{ass:lossdominatesaversion}
$$
g^\prime _\theta (F_\theta (l)) \left|\frac{\partial F_\theta (l)}{\partial \theta} \right| \geq \left| \frac{\partial g _{\theta} }{ \partial \theta } (F_\theta (l))\right|, \ \forall \, l \in [0, \bar L]. 
$$
\end{assumption}

Assumption \ref{ass:lossdominatesaversion} ensures that as the agent's type increases, the effect of facing larger losses dominates the reduction in risk aversion.
Moreover, it follows that
$$
\frac{\partial }{ \partial \theta } \left [ g_{\theta} \big(  F_{\theta} (l)  \big)  \right ] 
=
\frac{\partial  g_{\theta} }{ \partial \theta } \big(  F_{\theta} (l)  \big) + g^{\prime}_{\theta}\big(  F_{\theta} (l)  \big)  \frac{\partial  F_{\theta}(l) }{ \partial \theta } 
\leq 0.
$$

\medskip

In this setting, the characterization of individually rational and incentive compatible menus follows exactly as in Subsection \ref{sec:solutionIPO}. 
We aim to characterize solutions to Problem \eqref{eq:IPO_sup_retention} under the alternative type ordering assumptions. 

\medskip

\begin{proposition}
\label{prop:ordering2}
Suppose that $\eta$ satisfies Assumption \ref{ass:eta2}, and let $\alpha_0$ be defined as in Lemma \ref{le:alpha0}.  Consider the following cases. 
\medskip
\begin{enumerate}
\item \label{alpha1_o2} The case where $\alpha \in \left(0, \alpha_0 \right)$. Define $R^*_\theta$ by 
$R^*_\theta(l) = \int_0^l r^*_\theta(s) \, ds$,
where the marginal retention $r^*_\theta$ satisfies \eqref{eq:optimal_marginal_retention_Sec:PO} and $J_{\theta, \eta}(l)$ satisfies \eqref{eq:Jeta}. Define $p^*_\theta$ by \eqref{eq:optimal_premium}, and suppose that the following conditions hold.
\begin{enumerate}
\item[(i)] $J_{\theta,\eta}(l)$ is measurable and non-decreasing in $\theta$, for each $l$.
\item[(ii)] The value of $r^*_\theta(l)$ on the set $\{J_{\theta,\eta}(l)=0\}$ is chosen so that $r^*_\theta(l)$ is measurable and non-increasing in $\theta$.
\item[(iii)] $\int_\Theta V_\theta(R^*_\theta,p^*_\theta)\,d\mu(\theta) \geq 0.$
\end{enumerate}

\medskip

Then for a given $\eta$ and $\alpha$, the collection $\{R^*_\theta\}_{\theta \in \Theta}$ is submodular, and $\{p^*_\theta\}_{\theta \in \Theta}$ is non-negative. 
Moreover,  $(R^*_\theta,p^*_\theta)_{\theta\in\Theta}$ is an optimal solution for Problem \eqref{eq:IPO_sup_retention}.

\medskip 

\item\label{alpha2_o2} The case where $\alpha \in \left[\alpha_0 , \frac{1}{2} \right]$.
There exists $\theta_\alpha\in\Theta$ satisfying $\frac{\bar Q_{\eta} (\theta_{\alpha}) }{\bar Q (\theta_{\alpha}) } = \frac{1-\alpha}{\alpha}$.
\begin{enumerate}
\item[(a)] For $\theta < \theta_{\alpha}$, let $R^*_\theta(l) :=\int_0^l r^*_\theta(s) \,ds$, where $r^*_\theta(l)$ satisfies \eqref{eq:optimal_marginal_retention_Sec:PO} and $J_{\theta, \eta}(l)$ is given by \eqref{eq:Jeta}. 

\medskip

\item[(b)] For every $\theta\geq \theta_\alpha$, define $R^*_\theta$ by $ R^*_{ \theta}(l) =0$ for all $l \in [0, \bar L]$.
\end{enumerate}

\smallskip
Define the premium schedule $p^*_\theta$ by \eqref{eq:optimal_premium}. Suppose that the following two conditions hold for $\theta<\theta_\alpha$:
\begin{enumerate}
\item [(i)] $J_{\theta,\eta}(l)$ is measurable and non-decreasing in $\theta$, for each $l$.
\item[(ii)] The value of $r^*_\theta(l)$ on the set $\{J_{\theta,\eta}(l)=0\}$ is chosen so that $r^*_\theta(l)$ is measurable and non-increasing in $\theta$.
\end{enumerate}

\noindent Then for a given $\eta$ and $\alpha$, the collection $\{R^*_\theta\}_{\theta \in \Theta}$ is submodular, and the collection $\{p^*_\theta\}_{\theta \in \Theta}$ is non-negative. Moreover, if $\int_\Theta V_\theta(R^*_\theta,p^*_\theta)\,d\mu(\theta) \geq 0$, then $(R^*_\theta, p^*_\theta)_{\theta \in \Theta}$ is an optimal solution to Problem \eqref{eq:IPO_sup_retention}. 
\medskip

\item\label{alpha3_o2} The case where $\alpha\in\left(\frac{1}{2},1\right)$.
Suppose that condition \eqref{eq_case3_condition} holds, then there exists an optimal solution
$(R_\theta^*,p_\theta^*)_{\theta\in\Theta}$ to Problem \eqref{eq:IPO_sup_retention}, such that
$R_\theta^*(l)=0, \, \forall l\in[0,\bar L], \, \forall\theta\in\Theta$, and the premia $\{p^*_\theta\}_{\theta \in \Theta}$ satisfy \eqref{eq:high_alpha_binding_premium_theorem}.
Moreover, this is the unique optimal full-coverage pooling menu, and it binds the insurer's participation constraint.
\end{enumerate}
\end{proposition}

\medskip

Under Assumptions \ref{ass:ordering2} and \ref{ass:lossdominatesaversion}, the composite function $g_\theta (F_\theta(l))$ remains non-increasing in $\theta$ for all $l \in [0, \bar L]$. Hence the proof of Theorem \ref{th:solution_characterization_IPO} applies unchanged to Proposition \ref{prop:ordering2}. 
In particular, Proposition \ref{prop:ordering2} shows that when $\alpha$ satisfies cases \eqref{alpha1_o2} and \eqref{alpha2_o2}-(a), a separating layered equilibrium emerges. Optimal retentions are submodular ensuring that higher types, facing larger losses, receive more coverage, despite being less risk averse. This is because the larger loss faced by a high-type agent dominates their lower risk aversion, so the agent still requires more coverage, as captured by Assumption \ref{ass:lossdominatesaversion}.
If $\alpha \leq \frac{1}{2}$ and $\theta \geq \theta_\alpha$, full coverage is offered to the agent and optimal premia satisfy \eqref{eq:optimal_premium}. 
On the other hand, if $\alpha > \frac{1}{2}$, there exists a unique optimal pooling menu that provides full coverage to all types. Provided that the insurer’s break-even premium for full insurance does not exceed the lowest type's willingness to pay for full coverage, incentive compatibility implies a common premium $p^*$ across types, equal to the insurer’s aggregate certainty-equivalent cost of providing full coverage. In this case, the insurer’s participation constraint binds leaving them indifferent between participating in the market or not. That is, $ \int_\Theta V_\theta (0, p^*) \,d\mu(\theta)= 0$. 

\medskip

The sufficient conditions of Lemma \ref{le:IRimplication} ensuring the insurer’s participation constraint continue to apply in this setting. Next, we provide sufficient conditions that ensure the monotonicity of $J_{\theta, \eta}(l)$, when $\alpha$ satisfies cases \eqref{alpha1_o2} and \eqref{alpha2_o2}-(i) .

\medskip
\begin{proposition}
When $\alpha$ satisfies cases \eqref{alpha1_o2} or \eqref{alpha2_o2}-(a), the function
$J_{\theta, \eta}(l)$ is non-decreasing in $\theta$ for all $l \in [0,\bar L]$ if the following conditions hold for $\theta \in [\underline{\theta}, \bar \theta)$.
\smallskip
\begin{enumerate}
\item \label{C21} $ 0 \leq \left( \frac{\bar Q(\theta) }{ q(\theta) }  \right) ' \leq 1$;

\medskip
\item \label{C22} The function $ \theta \mapsto F_{\theta}$ is convex in $\theta$ for all $l$, that is, $\frac{ \partial ^2 F_{\theta}(l) }{ \partial \theta^2 } \geq 0$;

\medskip
\item \label{C23} The function $ \theta \mapsto g_{\theta} (t)$ is convex in $\theta$ for all $t$, that is, for $F_{\theta}(l) \in [0,1]$, $ \frac{ \partial ^2 g_{\theta} }{ \partial \theta^2 } (F_{\theta}(l)) \geq 0 $;

\medskip
\item \label{C24} The function $t \mapsto g_{\theta}(t)$ is convex in $t$ for all $\theta \in \Theta$, that is, for $F_{\theta}(l) \in [0,1]$, $ g^{\prime \prime} _{\theta} (F_{\theta}(l)) \geq 0$;

\smallskip
\item \label{C25} The function $g:(\theta, t) \mapsto g_\theta (t)$ satisfies 
$$
\frac{\partial^2 g_\theta}{\partial \theta \ \partial t}(t)\leq 0,
$$
such that the following hold:
$$
2\, \frac{\bar Q(\theta)}{q(\theta)}  \frac{\partial^2 g_\theta}{\partial \theta \ \partial t}(F_\theta (l)) \leq - g^{In \ \prime} \big(F_\theta(l) \big), \ \text{and}
$$
$$
\left(\frac{\bar Q_{\eta}(\theta) }{\bar Q (\theta) } \frac{\bar Q(\theta)}{q(\theta)} \right) ^{\prime}
\left[ \frac{\partial g_{\theta}}{ \partial \theta} (F_{\theta}(l)) + g^{\prime}_{\theta}\big(  F_{\theta} (l)  \big) \frac{\partial F_{\theta}(l)}{\partial \theta} \right] 
\leq - 2\,
\frac{\bar Q_\eta(\theta)}{\bar Q(\theta)} \frac{\bar Q(\theta)}{q(\theta)}  \frac{\partial F_\theta (l)}{\partial \theta}\frac{\partial^2 g_\theta}{\partial \theta \ \partial t}(F_\theta (l)).
$$
\end{enumerate} 
\end{proposition}

\begin{proof}
Recall that in this region, we have $1- \alpha -\alpha \frac{\bar Q_{\eta}(\theta) }{\bar Q (\theta) } \geq 0$. The partial derivative of $J_{\theta, \eta}(l)$ with respect to $\theta$ is given by:
\begin{align*}
\frac{\partial J_{\theta, \eta}(l)}{\partial \theta} 
&=  (1-\alpha) \left( \left(  \frac{\bar Q(\theta) }{q(\theta) } \right)^{\prime } - 1 \right) \left[ \frac{\partial g_{\theta}}{ \partial \theta} (F_{\theta}(l)) + g^{\prime}_{\theta}\big(  F_{\theta} (l)  \big) \frac{\partial F_{\theta}(l)}{\partial \theta} \right]\\
&\quad
+ (1-\alpha) \frac{\partial F_{\theta}(l)}{\partial \theta}  \left[ 2\, \frac{\bar Q(\theta) }{ q (\theta) }\frac{ \partial ^2 g _{\theta} }{ \partial \theta \,\, \partial t} (F_{\theta}(l)) +  g^{In \ \prime}(F_{\theta}(l))\right] \\
&\quad -
\alpha \left[ \left(\frac{\bar Q_{\eta}(\theta) }{\bar Q(\theta) }\frac{\bar Q(\theta)}{q(\theta)} \right) ^{\prime}
\left[ \frac{\partial g_{\theta}}{ \partial \theta} (F_{\theta}(l)) + g^{\prime}_{\theta}\big(  F_{\theta} (l)  \big) \frac{\partial F_{\theta}(l)}{\partial \theta} \right] + 2 \,\frac{\partial F_{\theta}(l)}{\partial \theta} \frac{\bar Q_\eta(\theta)}{\bar Q(\theta)}\frac{\bar Q(\theta)}{q(\theta)}\frac{ \partial ^2 g _{\theta} }{ \partial \theta \partial t} (F_{\theta}(l))
\right] \\
&\quad + 
\left[ 1- \alpha -\alpha \frac{\bar Q_{\eta}(\theta) }{\bar Q (\theta) } \right] \frac{\bar Q(\theta) }{q(\theta) }\cdot \left[   g^{\prime}_{\theta} (F_{\theta}(l))  \frac{\partial^2 F_{\theta}(l)}{\partial \theta^2}    +
\frac{ \partial ^2 g_{\theta} }{ \partial \theta^2 } (F_{\theta}(l))  +
g^{\prime \prime} _{\theta} (F_{\theta}(l)) \left( \frac{\partial F_{\theta}(l)}{\partial \theta} \right)^2  \right].
\end{align*}

Using the monotonicity implications of Assumptions \ref{ass:eta2} and \ref{ass:ordering2}, if conditions \ref{C21} to \ref{C25} hold on $[\underline{\theta}, \bar \theta)$, then $\frac{\partial J_{\theta, \eta}(l)}{\partial \theta} \geq 0$, and hence $\theta \mapsto J_{\theta, \eta}(l)$ is non-decreasing in $\theta$ on $[\underline{\theta}, \bar \theta)$, for all $l \in [0, \bar L]$. Moreover,
$$
J_{\theta, \eta}(l) \bigg|_{\theta = \bar \theta} := \underset{\theta \to \bar \theta}{\lim} J_{\theta, \eta}(l) =
(1-\alpha)
\left[g^{In}(F_{\bar\theta}(l))
-g_{\bar\theta}(F_{\bar\theta}(l))
\right].
$$

Since $\theta \mapsto J_{\theta, \eta}(l)$ is continuous on $\Theta$ and non-decreasing on $[\underline{\theta},\bar \theta)$, then it follows that
$\theta \mapsto J_{\theta, \eta}(l)$ is non-decreasing in $\theta$ on $\Theta$.
\end{proof}

\medskip

Condition \ref{C21} ensures that the population of higher types does not thin out too rapidly. As $\theta$ increases, Condition \ref{C22} ensures that the loss distribution becomes riskier at a decreasing rate, and Condition \ref{C23} ensures that higher types become less risk averse at an increasing rate.
Condition \ref{C24} guarantees strong risk aversion for each type. Condition \ref{C25} implies submodularity of the function $g$, meaning that as $\theta$ increases, the marginal distortion decreases. 
While Conditions \ref{C21} to \ref{C24} remain unchanged under the alternative type ordering compared to the sufficient conditions of monotonicity in Subsection \ref{sec:solutionIPOmon}, Condition \ref{C25} must be strengthened to preserve monotonicity.

\medskip

The results obtained in Subsection \ref{sub:properties_optimalmenu} on the optimal menu of contracts $(R^*_\theta, p^*_\theta)_{\theta \in \Theta}$, continue to hold in this setting. The insurer's participation constraint has distinct implications on the optimal menu of contracts, depending on the level of the social weight $\alpha$ and the value of the optimal marginal retention.
In particular, for small values of the social weight $\alpha$, when the optimal premium satisfies \eqref{eq:optimal_premium}, Lemma \ref{le:IRimplication} implies that if the agent retains the entire loss, the insurer becomes indifferent between participating in the market or not. Otherwise, the insurer’s participation constraint imposes an upper bound on the insurer’s aggregate certainty-equivalent cost of coverage.
Moreover, we see in Proposition \ref{prop:insurerutilityoptimalmenu} that, when full coverage is provided to the agent for all types $\theta \in \Theta$, the insurer's utility becomes lower for higher types, who are less risk averse and face stochastically larger losses. If partial coverage is provided, then the insurer's utility increases with the agent's type if \eqref{eq:inequality_condition} holds.
On the other hand, for larger values of the social weight, namely when $\alpha > \frac{1}{2}$, and if condition \eqref{eq_case3_condition} holds, the unique optimal full-coverage pooling menu leaves the insurer indifferent only in the aggregate, as shown in Proposition \ref{prop:ordering2}.

Proposition \ref{prop:properties_optimalmenu} demonstrates that higher types of the agent, who are less risk averse and face stochastically larger losses, receive more coverage in exchange for higher premia. The highest type $\bar \theta$, who is the least risk averse but faces the largest loss, receives full coverage at every loss level, if the highest type's loss distribution is nontrivial and the insurer is strictly less risk averse than the highest type. The lowest type $\underline{\theta}$, who is the most risk averse but faces the smallest loss, is indifferent between participating in the market and not participating. The agent's utility decreases with the type, meaning that higher types, i.e., less risk-averse types facing stochastically larger losses, receive lower utilities at the optimum. Additionally, the agent's utility is convex in types, if the same conditions of Proposition \ref{prop:properties_optimalmenu}-(4) hold.

\newpage
\section{Conclusion}
\label{sec:conclusion}

This paper examines a monopolistic insurance market with hidden information, where the agent's risk attitude and loss distribution are private information, and the agent's type is drawn from a continuum. Within this framework, we study the concept of incentive Pareto optimality, which extends the classical Pareto efficiency to settings of information asymmetry, and is constrained by requirements of incentive compatibility and individual rationality on optimal menus of contracts.

\medskip

Our first result shows that, for general utility functionals, if a menu of insurance contracts maximizes a social welfare function, subject to individual rationality and incentive compatibility constraints, then it is incentive efficient. Furthermore, in the special case of Yaari Dual Utilities, two partial converse results hold, under additional technical conditions. Under Yaari's Dual Utility, we characterize optimal menus of contracts that solve the social welfare maximization problem, and we show that under two distinct assumptions on the ordering of the type space, and with some regularity conditions, the optimal contract can either provide full coverage, or exhibit a layered structure of marginal retention functions, depending on the level of the social weight. 

\medskip

In addition, in the separating/layered regions, the optimal retention and the optimal premium are both monotone in the agent's type, with higher types receiving more coverage at the optimum in exchange for higher premium payments. Efficiency at the top holds under some strictness conditions, whereby full coverage is provided to the highest type. The insurer extracts all of the surplus from the lowest type agent, who is indifferent between participating and not participating at the optimum. Moreover, we show that, as the agent faces stochastically larger losses, their utility from the optimal menu decreases.
However, in the full-coverage pooling region where $\alpha > \frac{1}{2}$, incentive compatibility forces a common premium that binds the insurer's participation constraint if the insurer's aggregate certainty-equivalent cost of full coverage does not exceed the lowest type's reservation premium.
We also study the variation of the insurer's utility across types at the optimum, when a separating equilibrium holds. Particularly, if the agent is offered full coverage, the insurer benefits more from lower types of the agent that face smaller losses. If, on the other hand, the agent retains the entire loss, the insurer is indifferent in participating in the market. Finally, when partial coverage is offered, the insurer benefits more from higher types if a certain condition holds.

\newpage

\setlength{\parskip}{0.5ex}
\hypertarget{LinkToAppendix}{\ }
\appendix
\numberwithin{equation}{section}
\renewcommand{\theequation}{\thesection.\arabic{equation}}
\vspace{-0.4cm}

\section{Mathematical Background}
\label{AppMathBackground}

Throughout this appendix, let $(\Theta,\Omega,\mu)$ be a finite measure space, and let $(X,\|\cdot\|_X)$ be a Banach space.

\medskip
\subsection{Bochner Spaces}
\label{AppBochner}

We begin by recalling the basic notions of strong measurability, Bochner integrability, and Bochner $L^p$-spaces.

\medskip

\begin{definition}
A function $u:\Theta\to X$ is said to be \emph{strongly measurable} if there exists a sequence of simple functions $(u_n)_{n\ge 1}$ from $\Theta$ into $X$ such that
$$
u_n(\theta)\to u(\theta)
\ \ \hbox{for $\mu$-a.e.\ }\theta\in\Theta.
$$
\end{definition}

\medskip

\begin{definition}
A strongly measurable function $u:\Theta\to X$ is said to be \emph{Bochner integrable} if
$$
\int_\Theta \|u(\theta)\|_X\,d\mu(\theta)<\infty.
$$
In that case, its Bochner integral is denoted by
$$
\int_\Theta u(\theta)\,d\mu(\theta)\in X.
$$
\end{definition}

\medskip

\begin{definition}
Let $1 \leq p < \infty$. The Bochner space $L^p(\Theta;X)$ is the set of all strongly measurable functions $u:\Theta\to X$ such that
$$
\int_\Theta \|u(\theta)\|_X^p\,d\mu(\theta)<\infty,
$$

\noindent where two functions are identified whenever they agree $\mu$-a.e. The space $L^p(\Theta;X)$ is endowed with the norm
$$
\|u\|_{L^p}
:=
\left(\int_\Theta \|u(\theta)\|_X^p\,d\mu(\theta)\right)^{1/p}.
$$
\end{definition}

\medskip

\begin{remark}
When $X=\bbR^m$ is finite-dimensional, strong measurability is equivalent to ordinary measurability of the coordinate functions. In particular, if
$$
u=(u_1,\dots,u_m):\Theta\to\bbR^m,
$$
then
$$
u\in L^p(\Theta;\bbR^m)
\quad\Longleftrightarrow\quad
u_i\in L^p(\Theta,\mu), \ \, \forall \, i=1,\dots,m.
$$

\medskip

\noindent Moreover, the Bochner integral is computed componentwise:
$$
\int_\Theta u(\theta)\,d\mu(\theta)
=
\left(
\int_\Theta u_1(\theta)\,d\mu(\theta),
\dots,
\int_\Theta u_m(\theta)\,d\mu(\theta)
\right).
$$
\end{remark}

\medskip

\begin{theorem}
For every Banach space $X$ and every $1\le p<\infty$, the space $L^p(\Theta;X)$ is a Banach space.
\end{theorem}

\begin{proof}
See, for example, \cite[Chap.\ 2]{diestel1977uhl}.
\end{proof}

\medskip

In the body of the paper, we work with $X=\bbR^2$, equipped with the Euclidean norm $\|\cdot\|_{\bbR^2}$. Thus, for $1<p<\infty$, the space $L^p(\Theta;\bbR^2)$ is the Banach space of all $\bbR^2$-valued strongly measurable maps $u=(u_1,u_2)$ such that
$$
\int_\Theta \|u(\theta)\|_{\bbR^2}^p\,d\mu(\theta)<\infty.
$$

\medskip
\subsection{Reflexivity and Weak Compactness}

We next recall the notions of duality, reflexivity, and weak compactness used in the paper. For any normed space $N$, we denote by $N^*$ its continuous dual and by $N^{**}$ its bidual.

\medskip

\begin{definition}
Let $N$ be a normed space. The canonical map $J_N: N \to N^{**}$ is defined by
$$
J_N(x)(f):=f(x),
\ \ \forall \, x\in N,\ \forall \, f\in N^*.
$$

\noindent The space $N$ is said to be \emph{reflexive} if $J_N$ is surjective.
\end{definition}

\medskip

\begin{remark}
The following are standard observations:
\medskip
\begin{enumerate}
\item Every finite-dimensional normed space is reflexive.
\medskip
\item Every reflexive normed space is complete, and hence is a Banach space.
\end{enumerate}
\end{remark}

\medskip

\noindent Since $\bbR^2$ is finite-dimensional, it follows that $\bbR^2$ is reflexive.

\medskip

\begin{theorem}
\label{thm:Lp_reflexive_appendix}
Let $X$ be a reflexive Banach space, and let $1<p<\infty$. Then $L^p(\Theta;X)$ is reflexive.
\end{theorem}

\begin{proof}
See, for example, \cite[Chap.\ 4]{diestel1977uhl}.
\end{proof}

\medskip

\noindent Applying Theorem \ref{thm:Lp_reflexive_appendix} to $X=\bbR^2$, it follows that $L^p(\Theta;\bbR^2)$ is reflexive whenever $1<p<\infty$.

\medskip

The following is a standard result, often referred to as Kakutani's theorem. 

\medskip

\begin{theorem}
\label{thm:reflexive_iff_weakly_compact_ball}
A Banach space $X$ is reflexive if and only if its closed unit ball is weakly compact.
\end{theorem}

\begin{proof}
\cite[Theorem V.4.2]{conway2019course}.
\end{proof}

\medskip

\begin{corollary}
\label{cor:closed_bounded_convex_weakly_compact}
If $X$ is a reflexive Banach space, then every closed, bounded, and convex subset of $X$ is weakly compact.
\end{corollary}

\begin{proof}
Let $A\subset X$ be closed, bounded, and convex. Then there exists some $\lambda > 0$ such that
$$
A \subset \lambda \, B_X,
$$
where $B_X$ denotes the closed unit ball of $X$. By Theorem \ref{thm:reflexive_iff_weakly_compact_ball}, the set $B_X$ is weakly compact, and hence so is $\lambda \, B_X$. Since $A$ is norm closed and convex, it is weakly closed \cite[Corollary 4, p.12]{Diestel1984}. Therefore $A$, being a weakly closed subset of the weakly compact set $\lambda \, B_X$, is weakly compact.
\end{proof}

\medskip
\subsection{Weak Sequential Compactness}

\begin{theorem}
\label{thm:eberlein_smulian_appendix}
Let $X$ be a reflexive Banach space. Then every bounded sequence in $X$ admits a weakly convergent subsequence.
\end{theorem}

\begin{proof}
Let $(x_n)_{n\ge 1}$ be a bounded sequence in $X$. Then there exists some $\lambda>0$ such that
$$
\|x_n\|_X\leq \lambda, \ \ \forall\,n\geq 1.
$$

\noindent Hence $x_n \in \lambda \, B_X$, for all $n \geq 1$, where $B_X:=\{x\in X:\|x\|_X\leq 1\}$ is the closed unit ball of $X$. Since $X$ is reflexive, Theorem \ref{thm:reflexive_iff_weakly_compact_ball} implies that $B_X$ is weakly compact, and therefore so is $\lambda \, B_X$. By the Eberlein-\v{S}mulian theorem \cite[p.18]{Diestel1984}, weak compactness and weak sequential compactness coincide in Banach spaces. Thus the bounded sequence $(x_n)_{n\ge 1}\subset \lambda \,B_X$ admits a weakly convergent subsequence.
\end{proof}

\medskip
\subsection{Local Convexity}

\begin{definition}
A topological vector space is said to be locally convex if it admits a local base at $0$ consisting of convex sets.
\end{definition}

\medskip

\begin{remark}
Every normed space is a locally convex Hausdorff topological vector space. In particular, $L^p(\Theta;\bbR^2)$ is locally convex and Hausdorff.
\end{remark}

\medskip
\subsection{Radon-Nikodym Property and Bochner Duality}

We now recall the duality theorem for Bochner $L^p$-spaces.

\medskip

\begin{definition}
A Banach space $X$ is said to have the Radon-Nikodym property (RNP) if, for every probability space $(S,\Sigma,\nu)$ and every $X$-valued countably additive vector measure $\zeta$ of bounded variation that is absolutely continuous with respect to $\nu$, there exists a Bochner integrable function $h\in L^1(S;X)$ such that
$$
\zeta(A)=\int_A h\,d\nu, \ \ \forall\,A\in\Sigma.
$$
\end{definition}

\medskip

\begin{remark}
If $X$ is reflexive, then $X$ has the RNP \cite[Chap.\ 4]{diestel1977uhl}.
\end{remark}

\medskip

\begin{theorem}
\label{thm:bochner_duality_general}
\label{th:duality}
Let $(\Theta,\Omega,\mu)$ be a finite measure space, let $1 \leq p<\infty$, let $q\in[1,\infty]$ satisfy $\frac{1}{p} + \frac{1}{q} = 1$, and let $X$ be a Banach space. The dual of $L^p(\Theta, X)$ is given by $L^p(\Theta; X) ^* = L^q(\Theta; X^*)$ if and only if $X^*$ has the RNP.

\medskip

Moreover, when $X^*$ has the RNP, every continuous linear functional $\ell$ on $L^p(\Theta, X)$ can be represented uniquely by some $g \in L^q(\Theta; X^*)$, as follows:
$$
\ell(f) = \int_{\Theta} \langle g(\theta), f(\theta) \rangle \, d\mu(\theta) , \ \ \forall \,  f \in L^p(\Theta;X),
$$
where $ \langle g(\theta), f(\theta) \rangle $ is the dual pairing. 
\end{theorem}

\begin{proof}
See \cite[Chap.\ 4, Theorem 1]{diestel1977uhl}.
\end{proof}

\medskip

We now specialize to the case $X=\bbR^2$.

\medskip

\begin{proposition}
\label{prop:duality}
Let $1\leq p<\infty$, and let $q\in[1,\infty]$ satisfy $\frac{1}{p}+\frac{1}{q}=1$. Then
$$
\l(L^p(\Theta;\bbR^2)\r)^*=L^q(\Theta;\bbR^2).
$$

\noindent More precisely, every continuous linear functional $\ell$ on $L^p(\Theta;\bbR^2)$ can be represented uniquely by some $g=(g_1,g_2)\in L^q(\Theta;\bbR^2)$ such that
$$
\ell(f)
=
\int_\Theta f_1(\theta)\,g_1(\theta)\,d\mu(\theta)
+
\int_\Theta f_2(\theta)\,g_2(\theta)\,d\mu(\theta),
\ \ \forall\,f=(f_1,f_2)\in L^p(\Theta;\bbR^2).
$$

\noindent Moreover,
$$
|\ell(f)|\le \|f\|_{L^p}\,\|g\|_{L^q} < +\infty,
\ \ \forall\,f\in L^p(\Theta;\bbR^2).
$$
\end{proposition}

\medskip

\begin{proof}
Since $\bbR^2$ is finite-dimensional, it is reflexive, and therefore it has the RNP. Therefore, by Theorem \ref{thm:bochner_duality_general}, we have 
$$
\l(L^p(\Theta;\bbR^2)\r)^* = L^q(\Theta;\bbR^2).
$$

\noindent Hence every continuous linear functional $\ell$ on $L^p(\Theta;\bbR^2)$ admits a unique representation
$$
\ell(f)=\int_\Theta \langle f(\theta),g(\theta)\rangle_{\bbR^2}\,d\mu(\theta),
\ \ \forall\,f\in L^p(\Theta;\bbR^2),
$$

\noindent for some unique $g=(g_1,g_2)\in L^q(\Theta;\bbR^2)$. Writing out the Euclidean inner product, we obtain
$$
\ell(f)
=
\int_\Theta f_1(\theta)\,g_1(\theta)\,d\mu(\theta)
+
\int_\Theta f_2(\theta)\,g_2(\theta)\,d\mu(\theta),
\ \ \forall\,f=(f_1,f_2)\in L^p(\Theta;\bbR^2).
$$

Finally, for every $f=(f_1,f_2)\in L^p(\Theta;\bbR^2)$, we have
\begin{align*}
|\ell(f)|
&=
\left|
\int_\Theta \langle f(\theta),g(\theta)\rangle_{\bbR^2}\,d\mu(\theta)
\right|\\
&\leq
\int_\Theta \bigl|\langle f(\theta),g(\theta)\rangle_{\bbR^2}\bigr|\,d\mu(\theta)\\
&\leq
\int_\Theta \|f(\theta)\|_{\bbR^2}\,\|g(\theta)\|_{\bbR^2}\,d\mu(\theta)\\
&\leq
\|f\|_{L^p}\,\|g\|_{L^q}
<+\infty,
\end{align*}

\noindent where the second inequality follows from the Euclidean Cauchy-Schwarz inequality and the third from H\"{o}lder's inequality.
\end{proof}

\bigskip
\section{Proofs of Main Results} \label{Appendixproofs}

\subsection{Proof of Theorem \ref{th:IPOiff}}\label{App:proofIPOiff}

Assume that there exists a probability measure $\eta$ on $(\Theta, \cB(\Theta))$ that is equivalent to $\mu$, and  $\alpha \in (0,1)$ such that the menu of contracts $ \left(I^*_{ \theta} , p^*_{ \theta} \right) _{\theta \in \Theta}$ is optimal for Problem \eqref{eq:IPO_sup}.
For the sake of contradiction, we assume that $ \left(I^*_{ \theta} , p^*_{\theta} \right) _{\theta \in \Theta} \notin \cI\cP\cO$. By Definition \ref{def:IPO} there exists  $ ( I_{\theta}, p_{\theta})_{\theta \in \Theta} \in \cI\cR \cap \cI\cC$ such that for $\mu$-almost every $\theta \in \Theta$,
$$
U_{\theta}(I_{\theta}, p_{\theta}) \geq U_{\theta}(I^*_{\theta}, p^*_{\theta}) 
\ \ 
\hbox{and}
\ \
\int_{\Theta} V_{\theta} (I_{\theta}, p_{\theta}) \, d\mu(\theta) \geq \int_{\Theta} V_{\theta} (I^*_{\theta}, p^*_{\theta}) \, d\mu(\theta). 
$$

\medskip

\noindent In addition, at least one of the following two conditions holds: 
$$
\int_{\Theta} V_{\theta} (I_{\theta}, p_{\theta}) \, d\mu(\theta) >\int_{\Theta} V_{\theta} (I^*_{\theta}, p^*_{\theta}) \, d\mu(\theta)
\ \ 
\hbox{or}
\ \ 
\mu \left( \left\{ \theta \in \Theta \,; \, U_{\theta} (I_{\theta}, p_{\theta}) >  U_{\theta}(I^*_{\theta}, p^*_{\theta}) \,\right\} \right) >0.
$$

\medskip

\noindent We consider the following two cases:
\medskip
\begin{enumerate}
\item Suppose that $ \int_{\Theta}  V_{\theta} (I_{\theta}, p_{\theta}) d \mu (\theta)> \int_{\Theta}  V_{\theta} (I^*_{\theta}, p^*_{\theta}) d \mu(\theta) $.
Since $\eta$ is equivalent to $\mu$, and
$$
U_{\theta}(I_{\theta}, p_{\theta}) \geq U_{\theta}(I^*_{\theta}, p^*_{\theta}), \, \hbox{ $\mu$-a.e.,}
$$
    
\noindent this inequality also holds $\eta$-a.e. Moreover,
$$
\alpha \int_{\Theta} U_{\theta} (I_{\theta}, p_{\theta} ) \, d\eta(\theta) + (1-\alpha) \int_{\Theta} V_{\theta} (I_{\theta}, p_{\theta} ) \, d\mu(\theta) >  \alpha \int_{\Theta} U_{\theta} (I^*_{\theta}, p^*_{\theta} ) \, d\eta(\theta) + (1-\alpha) \int_{\Theta} V_{\theta} (I^*_{\theta}, p^*_{\theta} ) \, d\mu(\theta),
$$
   
\noindent contradicting the optimality of $(I^*_{\theta}, p^*_{\theta} ) _{\theta \in \Theta}$ for Problem \eqref{eq:IPO_sup}. Hence,  $  (I^*_{\theta}, p^*_{\theta} ) _{\theta \in \Theta} \in \cI \cP \cO$.

\bigskip    

\item Suppose now that $ \int_{\Theta}  V_{\theta} (I_{\theta}, p_{\theta}) \,d \mu(\theta) = \int_{\Theta}  V_{\theta} (I^*_{\theta}, p^*_{\theta})\, d \mu(\theta) $,
and hence
$$
\mu \left( \left\{ \theta \in \Theta \,; \,\, U_{\theta} (I_{\theta}, p_{\theta}) >  U_{\theta}(I^*_{\theta}, p^*_{\theta}) \,\right\} \right) >0. 
$$

\noindent It follows that,
$$ 
\eta \left( \left\{ \theta \in \Theta \,; \,\, U_{\theta} (I_{\theta}, p_{\theta}) >  U_{\theta}(I^*_{\theta}, p^*_{\theta}) \,\right\} \right) >0 .
$$

\noindent Therefore, 
$$
\alpha \int_{\Theta} U_{\theta} (I_{\theta}, p_{\theta} ) \, d\eta(\theta) + (1-\alpha) \int_{\Theta} V_{\theta} (I_{\theta}, p_{\theta} ) \, d\mu(\theta) >  \alpha \int_{\Theta} U_{\theta} (I^*_{\theta}, p^*_{\theta} ) \, d\eta(\theta) + (1-\alpha) \int_{\Theta} V_{\theta} (I^*_{\theta}, p^*_{\theta} ) \, d\mu(\theta),
$$

\noindent contradicting the optimality of $  (I^*_{\theta}, p^*_{\theta} ) _{\theta \in \Theta}$ for Problem \eqref{eq:IPO_sup}. Hence,  $  (I^*_{\theta}, p^*_{\theta} ) _{\theta \in \Theta} \in \cI \cP \cO$. \qed
\end{enumerate}

\bigskip
\subsection{Proof of Proposition \ref{IR_characterization}}\label{App:proofIRcharacterization}
Let  $(R_{\theta}, p_{\theta})_{\theta \in \Theta} \in \cI\cR$. Then it follows from Definition \ref{def:IR} that
$$
U_{\theta} ( R_{\theta} , p_{\theta} ) \geq U_{\theta}(L_{\theta}, 0), \,\,\, \forall \, \theta \in \Theta,
 \ \ 
\hbox{and}
 \ \ 
\int_\Theta V_\theta (R_\theta, p_\theta) \, d\mu \geq 0.
$$

\noindent Moreover, 
$$
U_{\theta} ( R_{\theta} , p_{\theta} ) 
= - p_{\theta} - \int_0^{ \bar L}  \left[ 1 - g_{\theta} \big( F_{\theta}(l)  \big)   \right] \,\, \frac{ \partial R_{\theta}(l) }{ \partial l } \, dl, 
$$
and
$$
U_{\theta} (L_{\theta} , 0) = - \int_0^{ \bar L} \left( 1-g_{\theta}\left(F_{\theta}(l)\right) \right)\, dl.
$$

\noindent Using (P1), we obtain:
$$
- p_{\theta} - \int_0^{ \bar L }  \left[ 1 - g_{\theta} \big( F_{\theta}(l)  \big)   \right] \,\, \frac{ \partial R_{\theta}(l) }{ \partial l } \,\, dl  
\geq 
- \int_0^{ \bar L }  \left[  1- g_{\theta}\left(F_{\theta}(l)\right) \right]\,\, dl .
$$

\noindent Hence,
$$
p_{\theta} 
\leq  
\int_0^{ \bar L }  \left[ 1 - g_{\theta} \big( F_{\theta}(l)  \big)   \right] \,\, \left [ 1-\frac{ \partial R_{\theta}(l) }{ \partial l } \right] \,\, dl  .
$$

\medskip

Conversely, consider a menu $(R_{\theta}, p_{\theta})_{\theta \in \Theta}$ that satisfies $\displaystyle\int_\Theta V_\theta (R_\theta, p_\theta) d\mu \geq 0$ and
$$
p_{\theta} 
\leq  
\int_0^{ \bar L }  \left[ 1 - g_{\theta} \big( F_{\theta}(l)  \big)   \right] \,\, \left [ 1-\frac{ \partial R_{\theta}(l) }{ \partial l } \right] \,\, dl.
$$ 

The above inequality can be rewritten as:
$$
p_{\theta} 
\leq 
\int_0^{ \bar L }  \left[ 1 - g_{\theta} \big( F_{\theta}(l)  \big)   \right] \,\,dl 
-
\int_0^{ \bar L }  \left[ 1 - g_{\theta} \big( F_{\theta}(l)  \big)   \right] \,\, \frac{ \partial R_{\theta}(l) }{ \partial l }\,\, dl,
$$

\noindent or equivalently,
$$
p_{\theta} + \int_0^{ \bar L }  \left[ 1 - g_{\theta} \big( F_{\theta}(l)  \big)   \right] \,\, \frac{ \partial R_{\theta}(l) }{ \partial l }\,\, dl
\leq 
\int_0^{ \bar L }  \left[ 1 - g_{\theta} \big( F_{\theta}(l)  \big)   \right] \,\,dl   ,
$$

\noindent which implies that
$$
U_{\theta} ( R_{\theta} , p_{\theta} ) \geq U_{\theta}(L_{\theta}, 0), \,\,\,\, \forall \theta \in \Theta .
$$

\medskip

\noindent Consequently, $(R_\theta,p_\theta)_{\theta \in \Theta} \in \cI\cR$ since it satisfies (P1) and (P2) of Definition \ref{def:IR}.  \qed

\bigskip
\subsection{Proof of Proposition \ref{IC_characterization}}
Consider an incentive compatible menu of contracts $( R_{\theta} , p_{\theta} ) _{ \theta \in \Theta} \in \cI \cC$. First, we know that for $\theta \in \Theta$,
\begin{equation}
\label{eqStar}
U_{\theta} ( R_{\theta} , p_{\theta} )
=
- p_{\theta} - \int_0^{ \bar L }  \left[ 1 - g_{\theta} \big( F_{\theta}(l)  \big)   \right] \,\, \frac{ \partial R_{\theta}(l) }{ \partial l } \,\, dl .
\end{equation}

\noindent We can also express the utility of a type-$\theta$ agent using $(R_{\theta ^\prime}, p_{\theta^ \prime})$, where $\theta^\prime  \in \Theta,\,\, \theta ^\prime\neq \theta$: 
\begin{align*}
U_{\theta} ( R_{\theta^\prime} , p_{\theta^\prime} )
&= - p_{\theta^\prime} - \int_0^{ \bar L }  \left[ 1 - g_{\theta} \big( F_{\theta}(l)  \big)   \right] \,\, \frac{ \partial R_{\theta^\prime}(l) }{ \partial l } \, dl  .
\end{align*}

\noindent Moreover, it follows from Remark \ref{Re:chain_rule} that:
\begin{align*}
\left| \frac{\partial U_{\theta} (R_{\theta^\prime}, p_{\theta^\prime})}{\partial \theta} \right| 
&= \left| \int_0^{\bar L} 
\left[ \frac{ \partial g_{\theta} }{ \partial \theta }(F_{\theta}(l)) + g^\prime _{\theta} (F_{\theta}(l)) \frac{\partial F_{\theta}}{\partial \theta}(l) \right] \, \frac{ \partial R_{\theta^\prime}(l) }{ \partial l } \,\, dl
\right|.
\end{align*}

\noindent Since $R_{\theta^\prime } \in \cR$, and by Assumption \ref{Ass:cdf_family} and Assumption \ref{Ass:distortion_family}, we obtain
$$
\left| 
\frac{\partial U_{\theta} (R_{\theta^\prime}, p_{\theta^\prime})}{\partial \theta} 
\right| \leq
\int_0^{\bar L } 
\left| \left[
\frac{ \partial g_{\theta} }{ \partial \theta }(F_{\theta}(l)) 
+ 
g^\prime_{\theta} (F_{\theta}(l)) \frac{\partial F_{\theta}}{\partial \theta}(l)  \,\, 
\right] 
\frac{ \partial R_{\theta^\prime}(l) }{ \partial l } \ 
\right|  dl
\leq  ( c + c' \delta) \bar L < +\infty. 
$$

\medskip

\noindent Hence, $U_{ \theta } ( R_{\theta^\prime}, p_{\theta^\prime})$ is Lipschitz continuous in $\theta$.
By the envelope theorem (e.g.,  \cite{milgrom2002envelope}), for any $\theta \in \Theta$, we have:
\begin{equation}
\label{eqStarStar}
\begin{split}
U_{\theta}(R_{\theta}, p_{\theta}) 
&= U_{ \underline{ \theta }} (R_{ \underline{ \theta} }, p_{ \underline{ \theta } })  + \int_{\underline{\theta} } ^ {\theta} \frac{\partial U_{s'}(R_s, p_s) }{ \partial s'}\bigg|_{s'=s} \,\, ds \\
&= - p_{\underline{ \theta }} - \int_0^{ \bar L}  \left[ 1 - g_{ \underline{ \theta }}  \big( F_{\underline{ \theta }}(l)  \big)   \right] \,\, \frac{ \partial R_{\underline{ \theta } }(l) }{ \partial l } \,\, dl 
+
\int_{\underline{\theta} } ^ {\theta} \int_0 ^{\bar L}  \left[ \frac{\partial g_s}{ \partial s} (F_s(l)) + g'_s(F_s(l)) \frac{\partial F_s(l)}{\partial s} \right] \frac{\partial R_s(l)}{\partial l } \, dl \, ds.
\end{split}
\end{equation}

\noindent Equating \eqref{eqStar} and \eqref{eqStarStar} yields
\begin{align*}
p_{\theta} = 
p_{\underline {\theta} } \,\, &+\int_0^{ \bar L}  \left[ 1 - g_{ \underline{ \theta }} \big( F_{\underline{ \theta }}(l)  \big)   \right] \,\, \frac{ \partial R_{\underline{\theta}}(l) }{ \partial l } \,\, dl 
-
\int_{ \underline{\theta}} ^ { \theta} \int_0 ^{\bar L}  \left[ \frac{\partial g_s}{ \partial s} (F_s(l)) + g'_s(F_s(l)) \frac{\partial F_s(l)}{\partial s} \right] 
\frac{\partial R_s(l)}{\partial l } \, dl \, ds \\
&\quad  -
\int_0^{ \bar L }  \left[ 1 - g_{\theta} \big( F_{\theta}(l)  \big)   \right] \,\, \frac{ \partial R_{\theta}(l) }{ \partial l } \,\, dl.  
\end{align*} \qed

\bigskip
\subsection{Proof of Proposition \ref{prop:IC_iff}}
We start by assuming that $ \{ R_{\theta} \}_{ \theta \in \Theta} $ is submodular and $ \{ p_{ \theta } \}_{ \theta \in \Theta}$ satisfies \eqref{eq:premium}. That is,
\begin{align*}
p_{\theta} =  p_{\underline {\theta} } \,\, 
&+
\int_0^{ \bar L}  \left[ 1 - g_{ \underline{ \theta }} \big( F_{\underline{ \theta }}(l)  \big)   \right] \,\, \frac{ \partial R_{\underline{\theta}}(l) }{ \partial l } \,\, dl  - \int_{ \underline{\theta}} ^ { \theta} \int_0 ^{\bar L}  \left[ \frac{\partial g_s}{ \partial s} (F_s(l)) + g'_s(F_s(l))  \frac{\partial F_s(l)}{\partial s} \right] \frac{\partial R_s(l)}{\partial l } \, dl \, ds \\
&-
\int_0^{ \bar L }  \left[ 1 - g_{\theta} \big( F_{\theta}(l)  \big)   \right] \,\, \frac{ \partial R_{\theta}(l) }{ \partial l } \,\, dl.  
\end{align*}

\noindent We show that the menu $(R_{\theta}, p_{\theta})_{\theta \in \Theta}$ is incentive compatible. To do so, we aim to show that for $ \theta,\theta^\prime \in \Theta$, and $\theta^\prime \neq \theta$, the following holds:
$$
U_{\theta} (R_{\theta}, p_{\theta})  \geq U_{\theta}(R_{\theta^\prime} , p_{\theta^\prime}) .
$$

\medskip

We first consider $\theta, \theta ^ \prime \in \Theta$ such that, $\theta < \theta ^\prime$. We have: 
\begin{align*}
U_{\theta}(R_{\theta^\prime} , p_{\theta^\prime})  
&=   
- p_{\theta^\prime} - \int_0^{ \bar L }  \left[ 1 - g_{\theta} \big( F_{\theta}(l)  \big)   \right] \,\, \frac{ \partial R_{\theta^\prime}(l) }{ \partial l } \,\, dl.
\end{align*}

\noindent Substituting $p_{\theta^\prime}$ by the corresponding expression given by \eqref{eq:premium} for $\theta^\prime \in \Theta$, we obtain the following:
\begin{align*}
U_{\theta}(R_{\theta^\prime} , p_{\theta^\prime})
&= - p_{\underline{\theta}} - \int_0^{ \bar L}  \left[ 1 - g_{ \underline{\theta}} \big( F_{\underline{\theta}}(l)  \big)   \right] \,\, \frac{ \partial R_{\underline{\theta}}(l) }{ \partial l } \,\, dl +  \int_{ \underline{\theta}} ^ { \theta^\prime} \int_0 ^{\bar L}   \left[  \frac{\partial g_s}{ \partial s} (F_s(l)) + g'_s(F_s(l))  \frac{\partial F_s(l)}{\partial s} \right]  \frac{\partial R_s(l)}{\partial l } \, dl \, ds \\
&\quad +  
\int_0^{ \bar L } \left[ g_{\theta}\big(F_{\theta}(l)\big) - g_{\theta^\prime }\big( F_{\theta^ \prime}(l)\big)\right] \frac{ \partial R_{\theta ^\prime}(l) }{ \partial l } \,dl.
\end{align*}

\noindent Since $\underline{\theta} \leq \theta < \theta^\prime$, the third term can be written as
$$
\int_{ \underline{\theta}} ^ { \theta } \int_0 ^{\bar L}   \left[  \frac{\partial g_s}{ \partial s} (F_s(l)) + g'_s(F_s(l))  \frac{\partial F_s(l)}{\partial s} \right]  \frac{\partial R_s(l)}{\partial l } \, dl \, ds 
+
\int_{ \theta} ^ { \theta^\prime} \int_0 ^{\bar L}   \left[  \frac{\partial g_s}{ \partial s} (F_s(l)) + g'_s(F_s(l))  \frac{\partial F_s(l)}{\partial s} \right]  \frac{\partial R_s(l)}{\partial l }\, dl \, ds .
$$

\medskip

\noindent Hence, 
\begin{align*}
U_{\theta}(R_{\theta^\prime} , p_{\theta^\prime})
&= U_{\theta} (R_{\theta}, p_{\theta})   +    \int_{ \theta} ^ { \theta^\prime} \int_0 ^{\bar L}   \left[  \frac{\partial g_s}{ \partial s} (F_s(l)) + g'_s(F_s(l))  \frac{\partial F_s(l)}{\partial s} \right]  \frac{\partial R_s(l)}{\partial l}\, dl \, ds \\
&\qquad +  \int_0^{ \bar L } \left[ g_{\theta}\big(F_{\theta}(l)\big) - g_{\theta^\prime }\big( F_{\theta^ \prime}(l)\big)\right] \frac{ \partial R_{\theta ^\prime}(l) }{ \partial l } \,dl.
\end{align*}

\noindent That is, 
\begin{align*}
U_{\theta}(R_{\theta^\prime} , p_{\theta^\prime})
&= U_{\theta} (R_{\theta}, p_{\theta})   +    \int_{ \theta} ^ { \theta^\prime} \int_0 ^{\bar L}   \left[  \frac{\partial g_s}{ \partial s} (F_s(l)) + g'_s(F_s(l))  \frac{\partial F_s(l)}{\partial s} \right]  \frac{\partial R_s(l)}{\partial l} \, dl \, ds \\
&\qquad +  \int_0^{ \bar L }\int_{\theta}^{\theta^\prime} - \frac{\partial}{\partial s} \big( g_s \circ F_s \big) (l) ds \,\, \frac{ \partial R_{\theta^\prime }(l) }{ \partial l } \, dl\\\
&=  U_{\theta} (R_{\theta}, p_{\theta})   +    \int_{ \theta} ^ { \theta^\prime} \int_0 ^{\bar L}   \left[  \frac{\partial g_s}{ \partial s} (F_s(l)) + g'_s(F_s(l))  \frac{\partial F_s(l)}{\partial s} \right]  \frac{\partial R_s(l)}{\partial l } \, dl \, ds \\
&\qquad - \int_{ \theta} ^ { \theta^\prime} \int_0 ^{\bar L}   \left[  \frac{\partial g_s}{ \partial s} (F_s(l)) + g'_s(F_s(l))  \frac{\partial F_s(l)}{\partial s} \right]  \frac{\partial R_{\theta^\prime}(l)}{\partial l } \, dl \, ds \\
&= U_{\theta} (R_{\theta}, p_{\theta})   +    \int_{ \theta} ^ { \theta^\prime} \int_0 ^{\bar L}   \left[  \frac{\partial g_s}{ \partial s} (F_s(l)) + g'_s(F_s(l))  \frac{\partial F_s(l)}{\partial s} \right] \left[  \frac{\partial R_s(l)}{\partial l } -  \frac{\partial R_{\theta^\prime}(l)}{\partial l } \right]\, dl \, ds \\
&\leq U_{\theta} (R_{\theta}, p_{\theta}) . 
\end{align*}

\medskip

The above inequality holds for the following two reasons. First, because $\{ R_{\theta }\}_{\theta \in \Theta}$ is submodular, it follows that for $\theta < \theta^\prime$, we have $ \frac{\partial R_{\theta}(l)}{\partial l } $ is non-increasing in $\theta$. Hence, for any $s  \in [\theta, \theta^\prime]$: 
$$
\frac{\partial R_s(l)}{\partial l } - \frac{ \partial R_{\theta^\prime}(l) }{ \partial l }  \geq 0.
$$

\noindent Second, we know from Subsection \ref{sub:typeordering} that $\frac{\partial g_s}{ \partial s} (F_s(l)) + g'_s(F_s(l)) \frac{\partial F_s(l)}{\partial s}  \leq 0$. 
Therefore,
$$
\left[   \frac{\partial g_s}{ \partial s} (F_s(l)) + g'_s(F_s(l)) \frac{\partial F_s(l)}{\partial s} \right] \cdot \left[ \frac{\partial R_s(l)}{\partial l } - \frac{ \partial R_{\theta^\prime}(l) }{ \partial l } \right]  \leq 0 . 
$$

\medskip

We can similarly prove that for $\theta < \theta^\prime$, $U_{\theta^\prime} (R_{\theta}, p_{\theta}) \leq U_{\theta^\prime} (R_{\theta^\prime}, p_{\theta^\prime}) $. Therefore, we conclude that  $(R_{\theta}, p_{\theta})_{\theta \in \Theta} \in \cI \cC$. The converse follows immediately from Proposition \ref{IC_characterization}. \qed

\bigskip
\subsection{Proof of Proposition \ref{prop:IR_lowest_type}}
Let $(R_{\theta}, p_{\theta})_{\theta \in \Theta} \in \cI \cC$ be such that $\int_\Theta V_\theta (R_\theta, p_\theta) d\mu \geq 0$. Assume that for the lowest type $\underline{\theta}$, the contract $( R_{\underline{\theta}}, p_{\underline{\theta}})$  satisfies the agent's participation (P1) of Definition \ref{def:IR}. We show that $(R_{\theta}, p_{\theta})_{\theta \in \Theta} \in \cI\cR$. 
Since (P2) of Definition \ref{def:IR} is satisfied, it remains to show that
$$
U_{\theta} (R_{\theta}, p_{\theta}) \geq U_{\theta} (L_{\theta}, 0), \ \forall \theta \in \Theta. 
$$

\noindent We have seen by the envelope theorem that:
\begin{align*}
U_{\theta} (R_{\theta}, p_{\theta}) 
&= U_{ \underline{ \theta } }(R_{ \underline{\theta} }, p_{ \underline{\theta} })  
+ \int_{\underline{\theta}} ^ {\theta} \frac{\partial U_{s '}(R_s, p_s) }{ \partial s'} \bigg|_{s'=s}\, ds \\
&\geq U_{\underline{\theta}} (L_{\underline{\theta}}, 0) +  \int_{\underline{\theta}} ^ {\theta} \frac{\partial U_{s'}(R_s, p_s) }{ \partial s'} \bigg|_{s'=s}\, ds \\
&= U_{\underline{\theta}} (L_{\underline{\theta}}, 0) +  \int_{\underline{\theta}} ^ {\theta} \int_0^{\bar L} \left[   \frac{\partial g_s}{ \partial s} (F_s(l)) + g'_s(F_s(l)) \frac{\partial F_s(l)}{\partial s} \right]  \frac{\partial R_s(l)}{\partial l } \,dl \, ds.
\end{align*}

\noindent Since $  \frac{\partial g_s}{ \partial s} (F_s(l))+
g^\prime _{s} (F_{s}(l)) \frac{\partial F_{s}}{\partial s}(l)  \leq 0 $, and for any $R_{s} \in \cR$,  $ 0 \leq \frac{ \partial R_{s}(l) }{ \partial l }  \leq 1$, we have:
$$
\left[   \frac{\partial g_s}{ \partial s} (F_s(l)) + g'_s(F_s(l)) \frac{\partial F_s(l)}{\partial s} \right]  \frac{\partial R_s(l)}{\partial l } 
\geq  
\frac{\partial g_s}{ \partial s} (F_s(l)) + g'_s(F_s(l)) \frac{\partial F_s(l)}{\partial s}  . 
$$

\noindent Hence,
\begin{align*}
U_{\theta} (R_{\theta}, p_{\theta}) 
&\geq 
U_{\underline{\theta}} (L_{\underline{\theta}}, 0) +  \int_{\underline{\theta}} ^ {\theta}  \int_0^{\bar L } \left[ \frac{\partial g_s}{ \partial s} (F_s(l)) + g'_s(F_s(l)) \frac{\partial F_s(l)}{\partial s}  \right] dl \, ds \\
&=
U_{\underline{\theta}} (L_{\underline{\theta}}, 0) + \int_{\underline{\theta}} ^ {\theta}  \frac{\partial U_{s'} (L_{s},0)} {\partial s'} \bigg|_{s'=s}  ds 
= 
U_{\theta} (L_{\theta}, 0) . 
\end{align*}

\noindent This implies that $(R_{\theta}, p_{\theta})_{\theta \in \Theta} \in \cI \cR$. Conversely, if $(R_{\theta}, p_{\theta})_{\theta \in \Theta} \in \cI \cR$ then it is trivial that $(R_{\underline{\theta}}, p_{\underline{\theta}})$  satisfies the agent's participation (P1) of Definition \ref{def:IR}. \qed

\bigskip
\subsection{Proof of Corollary \ref{cor:IRandIC_iff}}
Consider a collection of submodular retention functions $\{R_{\theta} \}_{\theta \in \Theta}$. Assume that $ \int_\Theta V_\theta (R_\theta, p_\theta)\, d\mu(\theta) \geq 0 $, and $\{ p_{\theta} \}_{\theta \in \Theta}$ satisfies \eqref{eq:premium} with,
$$
p_{ \underline{\theta} }
\leq
\int_0 ^{ \bar L} \left[ \, 1- g_{ \underline{\theta} }( F_{ \underline{ \theta} }(l) ) \, \right] \left[ 1- \frac{ \partial R_{ \underline{ \theta } } (l) }{\partial l} \right] \, dl .
$$

It follows from Proposition \ref{prop:IC_iff} that $(R_{\theta}, p_{\theta})_{\theta \in \Theta} \in \cI\cC $.
It remains to show that $(R_{\theta}, p_{\theta})_{\theta \in \Theta} \in \cI\cR$. Since $\int_\Theta V_\theta (R_\theta, p_\theta)\, d\mu(\theta) \geq 0$ and by Proposition \ref{prop:IR_lowest_type}, it is enough to show that condition (P1) of Definition \ref{def:IR} holds for $(R_{ \underline{\theta} }, p_{ \underline{\theta} } )$.  
We have:
\begin{align*}
U_{ \underline{ \theta } }(R_{ \underline{\theta} }, p_{ \underline{\theta} }) 
&=  -  p_{ \underline{\theta} } - \int_0^{\bar L} \left[ \, 1- g_{ \underline{\theta} }( F_{ \underline{ \theta} }(l) ) \, \right]  \frac{ \partial R_{ \underline{ \theta } } (l) }{\partial l} \, dl  \\
&\geq - \int_0^{\bar L} \left[ \, 1- g_{ \underline{\theta} }( F_{ \underline{ \theta} }(l) ) \, \right] \left[ 1- \frac{ \partial R_{ \underline{ \theta } } (l) }{\partial l} \right] \, dl 
- \int_0^{\bar L} \left[ \, 1- g_{ \underline{\theta} }( F_{ \underline{ \theta} }(l) ) \, \right]  \frac{ \partial R_{ \underline{ \theta } } (l) }{\partial l} \, dl  \\
&= -  \int_0^{\bar L} \left[ \, 1- g_{ \underline{\theta} }( F_{ \underline{ \theta} }(l) ) \, \right] dl 
= U_{ \underline{ \theta } }(L_{ \underline{\theta} },0) .
\end{align*}

\medskip

\noindent Hence, $(R_{\theta}, p_{\theta})_{\theta \in \Theta} \in \cI\cR$. Conversely, assume that $(R_{\theta}, p_{\theta})_{\theta \in \Theta} \in \cI\cR \cap \cI\cC$. It follows from Proposition \ref{prop:IC_iff}, that $\{p_{\theta}\}_{\theta \in \Theta} $ satisfies \eqref{eq:premium}. 
Moreover, by individual rationality we know that
$$
\int_\Theta V_\theta (R_\theta, p_\theta) \,d\mu(\theta) \geq 0.
$$

\noindent Additionally, by Proposition \ref{prop:IR_lowest_type}, $(R_{ \underline{\theta} }, p_{ \underline{\theta} })$ satisfies (P1) of Definition \ref{def:IR}. That is,
$$
U_{ \underline{ \theta } }(R_{ \underline{\theta} }, p_{ \underline{\theta} })
\geq 
U_{ \underline{ \theta } }(L_{ \underline{\theta} },0),
$$

\noindent which implies that
$$
p_{ \underline{\theta} }
\leq
\int_0 ^{ \bar L_{\theta} } \left[ \, 1- g_{ \underline{\theta} }( F_{ \underline{ \theta} }(l) ) \, \right] \left[ 1- \frac{ \partial R_{ \underline{ \theta } } (l) }{\partial l} \right] \, dl.
$$ 
\qed

\bigskip
\subsection{Proof of Proposition \ref{prop:social_welfare_expression}}\label{App:prop:social_welfare_expression}
The social welfare function is given by:
\begin{align*}
W_{\eta, \alpha} \left( (R_{\theta},  p_{\theta})_{\theta \in \Theta} \right) 
&= \alpha \int_{\Theta} U_{\theta} (R_{\theta}, p_{\theta} ) \, d\eta(\theta)  + (1-\alpha) \int_{\Theta} V_{\theta} (R_{\theta}, p_{\theta} ) \, d \mu(\theta)\\
&=  \alpha \int_{\Theta} U_{\theta} (R_{\theta}, p_{\theta} ) \,\frac{dQ_{\eta}(\theta)}{dQ(\theta)} d \mu (\theta) + 
(1-\alpha) \int_{\Theta} V_{\theta} (R_{\theta}, p_{\theta} ) \,  d \mu (\theta).
\end{align*}
That is, 
\begin{align*}
W_{\eta, \alpha} \left( (R_{\theta},  p_{\theta})_{\theta \in \Theta} \right) 
&= \alpha \int_{\Theta} \left[ - p_{\theta} - \int_0^{\bar L} \left[ 1- g_{\theta} (F_{\theta}(l)) \right] \frac{ \partial R_{\theta}(l) }{ \partial l } dl \right]  \frac{dQ_{\eta}(\theta)}{dQ(\theta)}  \,d \mu(\theta)   \\
&\quad +
(1-\alpha)  \int_{\Theta} \left[   p_{\theta} - \int_0^{ \bar L }  \left [ 1 - g^{In}  \left( F_{\theta}(l) \right) \right]\left[ 1 - \frac{ \partial R_{\theta}(l) }{ \partial l } \right]\, dl \right] \,d \mu (\theta) \\
&=  \int_{\Theta} \int_0^{\bar L} \left[ (1-\alpha)\left[ 1 - g^{In}(F_{\theta}(l))  \right]  - \alpha  \left[ 1- g_{\theta} (F_{\theta}(l)) \right]  \frac{dQ_{\eta}(\theta)}{dQ(\theta)} \right] \cdot \frac{ \partial R_{\theta}(l) }{ \partial l } \, dl \, d \mu(\theta) \\
&\quad +
\int_{\Theta} \left[(1-\alpha) - \alpha \frac{dQ_{\eta}(\theta)}{dQ(\theta)}  \right] \cdot p_{\theta} \, d \mu(\theta )  - (1-\alpha) \int_{\Theta} \int_0^{\bar L} \left[ 1 - g^{In}(F_{\theta}(l))  \right] \,dl \, d \mu (\theta).
\end{align*}

\medskip

For every $(R_{\theta}, p_{\theta})_{\theta \in \Theta} \in \cI\cC$ we know that the premium $p_{\theta}$ satisfies \eqref{eq:premium}. Substituting this premium into the social welfare function, we obtain 
\begin{align*}
&W_{\eta, \alpha} \left( (R_{\theta},  p_{\theta})_{\theta \in \Theta} \right) \\
&\quad =
\int_{\Theta} \int_0^{\bar L} \left[  (1-\alpha)\left[ 1 - g^{In}(F_{\theta}(l))  \right]  - \alpha  \left[ 1- g_{\theta} (F_{\theta}(l)) \right]  \frac{dQ_{\eta}(\theta)}{dQ(\theta)} \right] \cdot \frac{ \partial R_{\theta}(l) }{ \partial l } \, dl \, d \mu (\theta) \\
&\quad +
\int_{\Theta} \left[(1-\alpha) - \alpha \frac{dQ_{\eta}(\theta)}{dQ(\theta)}  \right] \cdot p_{\underline{\theta}} \, d\mu(\theta ) \\
&\quad +
\int_{\Theta} \left[(1-\alpha) - \alpha \frac{dQ_{\eta}(\theta)}{dQ(\theta)}  \right] \int_0^{ \bar L}  \left[ 1 - g_{ \underline{ \theta }} \big( F_{\underline{ \theta} } (l)  \big)   \right] \, \frac{ \partial R_{\underline{ \theta }}(l) }{ \partial l } \, dl\, d\mu(\theta ) \\
&\quad - 
\int_{\Theta} \left[(1-\alpha) - \alpha \frac{dQ_{\eta}(\theta)}{dQ(\theta)}  \right] \int_{\underline{\theta}} ^ {\theta} \int_0 ^{\bar L} \left[ \frac{\partial g_s}{ \partial s} (F_s(l)) + g'_s(F_s(l)) \frac{\partial F_s(l)}{\partial s} \right] 
\frac{\partial R_s(l)}{\partial l } \, dl \, ds \,d\mu(\theta) \\
&\quad - 
\int_{\Theta} \left[(1-\alpha) - \alpha \frac{dQ_{\eta}(\theta)}{dQ(\theta)}  \right]  \int_0^{ \bar L }  \left[ 1 - g_{\theta} \big( F_{\theta}(l)  \big)   \right] \, \frac{ \partial R_{\theta}(l) }{ \partial l } \, dl \, d\mu(\theta) \\
&\quad -
(1-\alpha) \int_{\Theta} \int_0^{\bar L} \left[ 1 - g^{In}(F_{\theta}(l))  \right] \,dl d\mu(\theta) \\ &\\
&=
\int_{\Theta} \int_0^{\bar L} \left[  (1-\alpha)\left[ 1 - g^{In}(F_{\theta}(l))  \right]  - \alpha  \left[ 1- g_{\theta} (F_{\theta}(l)) \right]  \frac{dQ_{\eta}(\theta)}{dQ(\theta)} \right] \cdot \frac{ \partial R_{\theta}(l) }{ \partial l } \, dl \, d\mu(\theta) \\
&\quad + 
\int_{\Theta} \left[(1-\alpha) - \alpha \frac{dQ_{\eta}(\theta)}{dQ(\theta)}  \right] \cdot  \left[ p_{\underline{\theta}} +  \int_0^{ \bar L}  \left[ 1 - g_{ \underline{ \theta }} \big( F_{\underline{ \theta} } (l)  \big)   \right] \, \frac{ \partial R_{\underline{ \theta }}(l) }{ \partial l } \, dl \right]\, d\mu(\theta ) \\
&\quad - 
\int_{\Theta} \left[(1-\alpha) - \alpha \frac{dQ_{\eta}(\theta)}{dQ(\theta)}  \right] \int_{\underline{\theta}} ^ {\theta} \int_0 ^{\bar L}  \left[ \frac{\partial g_s}{ \partial s} (F_s(l)) + g'_s(F_s(l)) \frac{\partial F_s(l)}{\partial s} \right] 
\frac{\partial R_s(l)}{\partial l } \, dl \, ds \,d\mu(\theta) \\
&\quad - 
\int_{\Theta} \left[(1-\alpha) - \alpha \frac{dQ_{\eta}(\theta)}{dQ(\theta)}  \right]  \int_0^{ \bar L }  \left[ 1 - g_{\theta} \big( F_{\theta}(l)  \big)   \right] \, \frac{ \partial R_{\theta}(l) }{ \partial l } \, dl \, d\mu(\theta) \\
&\quad - 
(1-\alpha) \int_{\Theta} \int_0^{\bar L} \left[ 1 - g^{In}(F_{\theta}(l))  \right] \,dl \, d\mu(\theta).
\end{align*}

\medskip

\noindent Simplifying the first and fourth terms gives
\begin{align*}
&W_{\eta, \alpha} \left( (R_{\theta}, p_{\theta})_{\theta \in \Theta} \right) \\
&\quad = \int_{\Theta} \int_0^{\bar L } (1-\alpha) \left[ g_{\theta} (F_{\theta}(l))  - g^{In} (F_{\theta}(l))  \right] \frac{ \partial R_{\theta}(l) }{ \partial l } \, dl \, d\mu(\theta) \\ 
&\quad +
\int_{\Theta} \left[(1-\alpha) - \alpha \frac{dQ_{\eta}(\theta)}{dQ(\theta)}  \right] \cdot  \left[ p_{\underline{\theta}} +  \int_0^{ \bar L}  \left[ 1 - g_{ \underline{ \theta }} \big( F_{\underline{ \theta} } (l)  \big)   \right] \, \frac{ \partial R_{\underline{ \theta }}(l) }{ \partial l } \, dl \right]\, d\mu(\theta ) \\
&\quad - 
\int_{\Theta} \left[(1-\alpha) - \alpha \frac{dQ_{\eta}(\theta)}{dQ(\theta)}  \right] \int_{\underline{\theta}} ^ {\theta} \int_0 ^{\bar L}  \left[ \frac{\partial g_s}{ \partial s} (F_s(l)) + g'_s(F_s(l)) \frac{\partial F_s(l)}{\partial s} \right] \frac{\partial R_s(l)}{\partial l } \, dl \, ds \, d\mu(\theta) \\
&\quad -
(1-\alpha) \int_{\Theta} \int_0^{\bar L} \left[ 1 - g^{In}(F_{\theta}(l))  \right] \,dl \,d\mu (\theta) \\
&\\
&= \int_{\Theta} \int_0^{\bar L } (1-\alpha) \left[ g_{\theta} (F_{\theta}(l))  - g^{In} (F_{\theta}(l))  \right] \frac{ \partial R_{\theta}(l) }{ \partial l } \, dl \, d\mu (\theta) \\ 
&\quad+
\int_{\Theta} \left[(1-\alpha) - \alpha \frac{dQ_{\eta}(\theta)}{dQ(\theta)}  \right] \cdot  \left[ p_{\underline{\theta}} +  \int_0^{ \bar L}  \left[ 1 - g_{ \underline{ \theta }} \big( F_{\underline{ \theta} } (l)  \big)   \right] \, \frac{ \partial R_{\underline{ \theta }}(l) }{ \partial l } \, dl \right]\, d\mu (\theta ) \\
&\quad -
(1-\alpha) \int_{\Theta}\int_{\underline{\theta}} ^ {\theta} \int_0 ^{\bar L}  \left[\frac{\partial g_s}{ \partial s} (F_s(l)) + g'_s(F_s(l)) \frac{\partial F_s(l)}{\partial s} \right] \frac{\partial R_s(l)}{\partial l } \, dl \, ds \, d\mu(\theta) \\
&\quad +
\alpha \int_{\Theta}\int_{\underline{\theta}} ^ {\theta} \int_0 ^{\bar L}  \left[ \frac{\partial g_s}{ \partial s} (F_s(l)) + g'_s(F_s(l)) \frac{\partial F_s(l)}{\partial s} \right] \frac{\partial R_s(l)}{\partial l } \, dl \, ds\, dQ_{\eta}(\theta)
\\
&\quad -
(1-\alpha) \int_{\Theta} \int_0^{\bar L} \left[ 1 - g^{In}(F_{\theta}(l))  \right] \,dl \, d\mu(\theta).
\end{align*}

\medskip

\noindent Integrating the third and fourth terms by parts yields
\begin{align*}
&W_{\eta, \alpha} \left( (R_{\theta}, p_{\theta})_{\theta \in \Theta} \right) \\
&\quad=\int_{\Theta} \int_0^{\bar L } (1-\alpha) \left[ g_{\theta} (F_{\theta}(l))  - g^{In} (F_{\theta}(l))  \right] \frac{ \partial R_{\theta}(l) }{ \partial l } \, dl \, d\mu(\theta) \\ 
&\quad + 
\int_{\Theta} \left[(1-\alpha) - \alpha \frac{dQ_{\eta}(\theta)}{dQ(\theta)}  \right] \cdot  \left[ p_{\underline{\theta}} +  \int_0^{ \bar L}  \left[ 1 - g_{ \underline{ \theta }} \big( F_{\underline{ \theta} } (l)  \big)   \right] \, \frac{ \partial R_{\underline{ \theta }}(l) }{ \partial l } \, dl \right]\, d\mu(\theta ) \\
&\quad +
(1-\alpha) \left[ \int_{\underline{\theta}} ^{\theta} \int_0^{\bar L} \left[ \frac{\partial g_s}{ \partial s} (F_s(l)) + g'_s(F_s(l)) \frac{\partial F_s(l)}{\partial s} \right] \frac{\partial R_s(l)}{\partial l } \, dl \, ds  \cdot \bar Q(\theta)  \right] _{\theta=\underline{\theta}}^{\theta = \bar \theta} \\
&\quad -
(1-\alpha) \int_{\Theta} \int_0^{\bar L}  \left[ \frac{\partial g_{\theta}}{ \partial \theta} (F_{\theta}(l)) + g'_{\theta}(F_{\theta}(l)) \frac{\partial F_{\theta}(l)}{\partial \theta} \right] \frac{\partial R_{\theta}(l)}{\partial l } \, dl \bar Q(\theta) d\theta \\
&\quad -
\alpha \left[ \int_{\underline{\theta}} ^ {\theta} \int_0 ^{\bar L}  
\left[ \frac{\partial g_s}{ \partial s} (F_s(l)) + g'_s(F_s(l)) \frac{\partial F_s(l)}{\partial s} \right] 
\frac{\partial R_s(l)}{\partial l } \, dl \, ds \, \bar Q_{\eta} (\theta) \right]_{\theta = \underline{\theta}} ^{\theta = \bar \theta} \\
&\quad + 
\alpha \int_{\Theta} \int_0^{\bar L}  \left[ \frac{\partial g_{\theta}}{ \partial \theta} (F_{\theta}(l)) + g'_{\theta}(F_{\theta}(l)) \frac{\partial F_{\theta}(l)}{\partial \theta} \right] \frac{\partial R_{\theta}(l)}{\partial l } \, dl \, \bar Q_{\eta}(\theta) \, d\theta \\
&\quad -
(1-\alpha) \int_{\Theta} \int_0^{\bar L} \left[ 1 - g^{In}(F_{\theta}(l))  \right] \,dl\, d\mu(\theta) \\
&\\
&=
\int_{\Theta} \int_0^{\bar L } (1-\alpha) \left[ g_{\theta} (F_{\theta}(l))  - g^{In} (F_{\theta}(l))  \right] \frac{ \partial R_{\theta}(l) }{ \partial l } \, dl \, d\mu(\theta) \\ 
&\quad + 
\int_{\Theta} \left[(1-\alpha) - \alpha \frac{dQ_{\eta}(\theta)}{dQ(\theta)}  \right] \cdot  \left[ p_{\underline{\theta}} +  \int_0^{ \bar L}  \left[ 1 - g_{ \underline{ \theta }} \big( F_{\underline{ \theta} } (l)  \big)   \right] \, \frac{ \partial R_{\underline{ \theta }}(l) }{ \partial l } \, dl \right]\, d\mu(\theta ) \\
&\quad -
(1-\alpha) \int_{\Theta} \int_0^{\bar L}  \left[ \frac{\partial g_{\theta}}{ \partial \theta} (F_{\theta}(l)) + g'_{\theta}(F_{\theta}(l)) \frac{\partial F_{\theta}(l)}{\partial \theta} \right] \frac{\partial R_{\theta}(l)}{\partial l } \, dl \,\bar Q(\theta)\, d\theta \\
&\quad + 
\alpha \int_{\Theta} \int_0^{\bar L}  \left[\frac{\partial g_{\theta}}{ \partial \theta} (F_{\theta}(l)) + g'_{\theta}(F_{\theta}(l)) \frac{\partial F_{\theta}(l)}{\partial \theta} \right] \frac{\partial R_{\theta}(l)}{\partial l } \, dl\, \bar Q_{\eta}(\theta) \, d\theta \\
&\quad
- (1-\alpha) \int_{\Theta} \int_0^{\bar L} \left[ 1 - g^{In}(F_{\theta}(l))  \right] \,dl \,d\mu(\theta) .
\end{align*}

\medskip

\noindent Looking at the second term, we can see that $ \left[ p_{\underline{\theta}} +  \int_0^{ \bar L}  \left[ 1 - g_{ \underline{ \theta }} \big( F_{\underline{ \theta} } (l)  \big)   \right] \, \frac{ \partial R_{\underline{ \theta }}(l) }{ \partial l } \, dl \right]$ is independent of $\theta$. Hence,
\begin{align*}
&W_{\eta, \alpha} \left( (R_{\theta}, p_{\theta})_{\theta \in \Theta} \right) \\
&\quad =
\int_{\Theta} \int_0^{\bar L } (1-\alpha) \left[ g_{\theta} (F_{\theta}(l))  - g^{In} (F_{\theta}(l))  \right] \frac{ \partial R_{\theta}(l) }{ \partial l } \, dl \, d\mu(\theta) \\ 
&\quad + 
\left[ p_{\underline{\theta}} +  \int_0^{ \bar L}  \left[ 1 - g_{ \underline{ \theta }} \big( F_{\underline{ \theta} } (l)  \big)   \right] \, \frac{ \partial R_{\underline{ \theta }}(l) }{ \partial l } \, dl \right] \cdot \int_{\Theta} \left[(1-\alpha) - \alpha \frac{dQ_{\eta}(\theta)}{dQ(\theta)}  \right] \, d\mu(\theta ) \\
&\quad -
(1-\alpha) \int_{\Theta} \int_0^{\bar L}  \left[\frac{\partial g_{\theta}}{ \partial \theta} (F_{\theta}(l)) + g'_{\theta}(F_{\theta}(l)) \frac{\partial F_{\theta}(l)}{\partial \theta} \right] 
\frac{\partial R_{\theta}(l)}{\partial l } \, dl \, \bar Q(\theta)\, d\theta \\
&\quad + 
\alpha  \int_{\Theta} \int_0^{\bar L}  \left[ \frac{\partial g_{\theta}}{ \partial \theta} (F_{\theta}(l)) + g'_{\theta}(F_{\theta}(l)) \frac{\partial F_{\theta}(l)}{\partial \theta} \right] 
\frac{\partial R_{\theta}(l)}{\partial l } \, dl \,\bar Q_{\eta}(\theta)\, d\theta  
-
(1-\alpha) \int_{\Theta} \int_0^{\bar L} \left[ 1 - g^{In}(F_{\theta}(l))  \right] \,dl \, d\mu(\theta) \\
&\\
&=
\int_{\Theta} \int_0^{\bar L } (1-\alpha) \left[ g_{\theta} (F_{\theta}(l))  - g^{In} (F_{\theta}(l))  \right] \frac{ \partial R_{\theta}(l) }{ \partial l } \, dl \, d\mu(\theta) \\ 
&\quad + 
\left[ p_{\underline{\theta}} +  \int_0^{ \bar L}  \left[ 1 - g_{ \underline{ \theta }} \big( F_{\underline{ \theta} } (l)  \big)   \right] \, \frac{ \partial R_{\underline{ \theta }}(l) }{ \partial l } \, dl \right] \cdot \left[(1-\alpha) Q(\Theta) -\alpha Q_{\eta}(\Theta) \right] \\
&\quad -
(1-\alpha) \int_{\Theta} \int_0^{\bar L}  \left[ \frac{\partial g_{\theta}}{ \partial \theta} (F_{\theta}(l)) + g'_{\theta}(F_{\theta}(l)) \frac{\partial F_{\theta}(l)}{\partial \theta} \right] 
\frac{\partial R_{\theta}(l)}{\partial l } \, dl \,\bar Q(\theta)\, d\theta \\
&\quad + 
\alpha \int_{\Theta} \int_0^{\bar L}  \left[ \frac{\partial g_{\theta}}{ \partial \theta} (F_{\theta}(l)) + g'_{\theta}(F_{\theta}(l))\frac{\partial F_{\theta}(l)}{\partial \theta} \right] \frac{\partial R_{\theta}(l)}{\partial l } \, dl \, \bar Q_{\eta}(\theta)\, d\theta
- (1-\alpha) \int_{\Theta} \int_0^{\bar L} \left[ 1 - g^{In}(F_{\theta}(l))  \right] \,dl \, d\mu(\theta) \\
&\\
&=
\int_{\Theta} \int_0^{\bar L } (1-\alpha) \left[ g_{\theta} (F_{\theta}(l))  - g^{In} (F_{\theta}(l))  \right] \frac{ \partial R_{\theta}(l) }{ \partial l } \, dl \, d\mu(\theta)  
+ 
(1-2\alpha) \left[ p_{\underline{\theta}} +  \int_0^{ \bar L}  \left[ 1 - g_{ \underline{ \theta }} \big( F_{\underline{ \theta} } (l)  \big)   \right] \, \frac{ \partial R_{\underline{ \theta }}(l) }{ \partial l } \, dl \right] \\
&\quad -
(1-\alpha) \int_{\Theta} \int_0^{\bar L}  \left[ \frac{\partial g_{\theta}}{ \partial \theta} (F_{\theta}(l)) + g'_{\theta}(F_{\theta}(l))\frac{\partial F_{\theta}(l)}{\partial \theta} \right] \frac{\partial R_{\theta}(l)}{\partial l } \, dl \,\bar Q(\theta)\, d\theta \\
&\quad + 
\alpha  \int_{\Theta} \int_0^{\bar L}  \left[ \frac{\partial g_{\theta}}{ \partial \theta} (F_{\theta}(l)) + g'_{\theta}(F_{\theta}(l)) \frac{\partial F_{\theta}(l)}{\partial \theta} \right] 
\frac{\partial R_{\theta}(l)}{\partial l } \, dl \, \bar Q_{\eta}(\theta) \,d\theta 
- (1-\alpha) \int_{\Theta} \int_0^{\bar L} \left[ 1 - g^{In}(F_{\theta}(l))  \right] \,dl \, d\mu(\theta) \\ 
&\\
&= \int_{\Theta} \int_0^{\bar L } (1-\alpha) \left[ g_{\theta} (F_{\theta}(l))  - g^{In} (F_{\theta}(l))  \right] \frac{ \partial R_{\theta}(l) }{ \partial l } \, dl \, d\mu(\theta) 
+ 
(1-2\alpha) \left[ p_{\underline{\theta}} +  \int_0^{ \bar L}  \left[ 1 - g_{ \underline{ \theta }} \big(F_{\underline{ \theta} } (l)\big)\right] \, \frac{ \partial R_{\underline{ \theta }}(l) }{ \partial l } \, dl \right]\\
&\quad -
\int_{\Theta} \int_0^{\bar L}  \left[ \frac{\partial g_{\theta}}{ \partial \theta} (F_{\theta}(l)) + g'_{\theta}(F_{\theta}(l))\frac{\partial F_{\theta}(l)}{\partial \theta} \right] 
\frac{\partial R_{\theta}(l)}{\partial l }  \cdot \left[ \frac{(1-\alpha) \bar Q(\theta) - \alpha \bar Q_{\eta}(\theta)}{q(\theta)}\right]\, dl \, q(\theta) d\theta \\
&\quad - 
(1-\alpha) \int_{\Theta} \int_0^{\bar L} \left[ 1 - g^{In}(F_{\theta}(l))  \right] \,dl \, d\mu(\theta).
\end{align*}

\medskip

\noindent Thus, 
\begin{align*}
W_{\eta, \alpha} \left( (R_{\theta}, p_{\theta})_{\theta \in \Theta} \right) 
&= (1-2\alpha)  \left[ p_{\underline{\theta}} +  \int_0^{ \bar L}  \left[ 1 - g_{ \underline{ \theta }} \big( F_{\underline{ \theta} } (l)  \big)   \right] \, \frac{ \partial R_{\underline{ \theta }}(l) }{ \partial l } \, dl \right] 
-
\int_{\Theta} \int_0^{\bar L} J_{\theta, \eta} (l) \frac{ \partial R_{\theta}(l) }{ \partial l } \, dl \,d\mu(\theta)\\
&\qquad- 
(1-\alpha) \int_{\Theta} \int_0^{\bar L} \left[ 1 - g^{In}(F_{\theta}(l))  \right] \,dl \, d\mu(\theta),
\end{align*}

\noindent where, 
$$
J_{\theta, \eta}(l) = 
(1-\alpha) \left[   g^{In} (F_{\theta}(l)) -  g_{\theta} (F_{\theta}(l)) \right] +  \left[ \frac{\partial g_{\theta}}{ \partial \theta} (F_{\theta}(l)) + g'_{\theta}(F_{\theta}(l)) \frac{\partial F_{\theta}(l)}{\partial \theta} \right] \left[ \frac{(1-\alpha) \bar Q(\theta) - \alpha \bar Q_{\eta}(\theta)}{q(\theta)}\right],
$$

\noindent which can be rewritten as
\begin{align} 
J_{\theta, \eta}(l) &= 
(1-\alpha) \left[  g^{In} (F_{\theta}(l))-  g_{\theta} (F_{\theta}(l)) \right] \nonumber \\
&\quad + 
\left(\frac{\bar Q(\theta) }{ q (\theta) } \right) \left[ \frac{\partial g_{\theta}}{ \partial \theta} (F_{\theta}(l))  + g'_{\theta}(F_{\theta}(l)) \frac{\partial F_{\theta}(l)}{\partial \theta} \right]  \left[ (1-\alpha) - \alpha \frac{ \bar Q_{\eta}(\theta)}{\bar Q(\theta)}\right].
\end{align}\qed

\bigskip
\subsection{Proof of Lemma \ref{le:alpha0}}
\label{App:le:alpha0}

First, using \eqref{eq:frac>1}, we have that $\alpha_0 \leq \frac12$. Since $\theta \mapsto \frac{\bar Q_{\eta}(\theta)}{\bar Q(\theta)}$ is non-decreasing on $[\underline{\theta},\bar{\theta})$, it follows that $\theta \mapsto  (1-\alpha) -  \alpha \frac{\bar Q_{\eta}(\theta)}{ \bar Q(\theta)}$ is non-increasing on $[\underline{\theta},\bar{\theta})$. Moreover, we have by \eqref{eq:frac_at_bar_theta} that:
$$ \left[ (1-\alpha) -  \alpha \frac{\bar Q_{\eta}(\theta)}{ \bar Q(\theta)} \right] _{\theta = \bar \theta} =
\underset{ \theta \to \bar \theta } {\lim}\left[ (1-\alpha) -  \alpha \frac{\bar Q_{\eta}(\theta)}{ \bar Q(\theta)} \right]  
= 1- \alpha - \alpha \frac{q_{\eta}(\bar \theta)}{q( \bar\theta)}.$$

\noindent With this continuous extension, $\theta \mapsto  (1-\alpha) -  \alpha \frac{\bar Q_{\eta}(\theta)}{ \bar Q(\theta)}$ is non-increasing on $\Theta$, which means that for $\theta\in [\underline{\theta}, \bar \theta)$,
$$
1-\alpha -  \alpha \frac{\bar Q_{\eta}(\theta)}{ \bar Q(\theta)} \geq 1- \alpha - \alpha \frac{q_{\eta}(\bar \theta)}{q(\bar \theta)}. 
$$

\noindent At $\theta =\bar \theta$, the continuous extension gives equality. Hence if $0 < \alpha < \alpha_0$, we have: 
$$ 1-\alpha -  \alpha \frac{\bar Q_{\eta}(\theta)}{ \bar Q(\theta)} >0, \ \forall \, \theta \in \Theta.$$ 

On the other hand, suppose that $\alpha \in [\alpha_0, \frac{1}{2}]$. Since $\alpha \leq \frac{1}{2}$, it follows that 
$\frac{1-\alpha}{\alpha}\geq 1$. Moreover, since $\alpha \geq \alpha_0$, we obtain
$\frac{1-\alpha}{\alpha}\leq
\frac{q_\eta(\bar\theta)}{q(\bar\theta)}$, which means that:
$$1
\leq \frac{1-\alpha}{\alpha}
\leq \frac{q_\eta(\bar\theta)}{q(\bar\theta)}.
$$

\noindent Then it follows from \eqref{eq:frac_at_underbar_theta} and \eqref{eq:frac_at_bar_theta} that:
$$
\frac{1-\alpha}{\alpha}
\in
\left[
\frac{\bar Q_\eta(\underline{\theta})}{\bar Q(\underline{\theta})}, \ 
\lim_{\theta\to\bar\theta}\frac{\bar Q_\eta(\theta)}{\bar Q(\theta)}
\right].
$$

\noindent Hence by the intermediate value theorem, there exists $\theta_\alpha\in\Theta$ such that
$\frac{\bar Q_\eta(\theta_\alpha)}
{\bar Q(\theta_\alpha)}
=\frac{1-\alpha}{\alpha}$. \qed

\bigskip
\subsection{Proof of Theorem \ref{th:solution_characterization_IPO}}\label{App:thIPOsolution}
First, we know from Proposition \ref{prop:social_welfare_expression} that $W_{\eta, \alpha} \left( (R_{\theta}, p_{\theta})_{\theta \in \Theta} \right)$ satisfies \eqref{eq:social_welfare} and $J_{\theta, \eta}(l)$ is given in \eqref{eq:Jeta}.  
We consider the following three cases for the values of the social weight $\alpha \in (0,1)$.

\medskip

\subsubsection{ The case where $\alpha \leq \frac{1}{2}$}
The social welfare function $  W_{\eta, \alpha} \left( (R_{\theta}, p_{\theta})_{\theta \in \Theta} \right) $ given in \eqref{eq:social_welfare} is non-decreasing with respect to $p_{\underline{\theta}}$, and then at the optimum, $p^*_{\underline{ \theta}}$ must take its largest value. 
By individual rationality, we conclude from Proposition \ref{IR_characterization} that:
$$
p^*_{\underline{ \theta }} = \int_0 ^{ \bar L} \left[  1- g_{ \underline{\theta} }( F_{ \underline{ \theta } }(l) ) \, \right] \left[ 1- \frac{ \partial R_{ \underline{ \theta } } (l) }{\partial l} \right] \, dl \geq 0.
$$

\noindent Using the premium $p_{\underline{ \theta }} = \int_0 ^{ \bar L} \left[  1- g_{ \underline{\theta} }( F_{ \underline{ \theta } }(l) ) \, \right] \left[ 1- \frac{ \partial R_{ \underline{ \theta } } (l) }{\partial l} \right] \, dl $, the social welfare functions simplifies to: 
\begin{align*}
W_{\eta, \alpha} \left( (R_{\theta}, p_{\theta})_{\theta \in \Theta} \right)   
&= (1- 2\alpha) \cdot \int_0^{\bar L} 
\left[ 1- g_{\underline{\theta}}(F_{\underline{\theta}}(l)) \right] \, dl 
- \int_{\Theta} \int_0^{\bar L} J_{\theta, \eta} (l) \frac{ \partial R_{\theta}(l) }{ \partial l } \, dl \,d \mu(\theta) \nonumber \\
&\qquad-
(1-\alpha) \int_{\Theta} \int_0^{\bar L} \left[ 1 - g^{In}(F_{\theta}(l))  \right] \,dl \, d\mu(\theta).
\end{align*}

\medskip

To find a solution for Problem \eqref{eq:IPO_sup_retention}, we aim to maximize this function pointwise. We first start by analyzing $J_{\theta, \eta}(l)$ given by \eqref{eq:Jeta}. First, since $\alpha \leq \frac{1}{2}$ and by Assumption \ref{ass:distortion_insurer_agent}, the first term is non-negative: 
$$
(1-\alpha) \left[ g^{In} (F_{\theta}(l))  -  g_{\theta} (F_{\theta}(l))\right]  \geq 0 \,.
$$

\noindent Moreover, by Assumption \ref{Ass:cdf_family} and Assumptions \ref{Ass:distortion_family}, we know that
$$
\frac{\bar Q(\theta) }{ q(\theta) } 
\left[ \frac{\partial g_{\theta}}{ \partial \theta} (F_{\theta}(l)) +  g'_{\theta}(F_{\theta}(l)) \frac{\partial F_{\theta}(l)}{\partial \theta} \right] \leq 0 \,.
$$

\medskip

\noindent We know from Lemma \ref{le:alpha0} that if $\alpha \in (0, \alpha_0)$, then 
$$ 1-\alpha -  \alpha \frac{\bar Q_{\eta}(\theta)}{ \bar Q(\theta)} >0, \ \forall \, \theta \in \Theta.$$ 
It follows that
$$
\frac{\bar Q(\theta) }{ q(\theta) } \left[ (1-\alpha) -  \alpha \frac{\bar Q_{\eta}(\theta)}{ \bar Q(\theta)} \right] 
\left[ \frac{\partial g_{\theta}}{ \partial \theta} (F_{\theta}(l)) +  g'_{\theta}(F_{\theta}(l)) \frac{\partial F_{\theta}(l)}{\partial \theta} \right]  \leq 0 \,.
$$

\medskip

\noindent Hence if $\alpha \in \left(0, \alpha_0 \right)$, for a fixed $\theta\in\Theta$, let
$r_\theta(l):=\frac{\partial R_\theta(l)}{\partial l}$. Since admissible retention functions satisfy $0 \leq r_\theta(l) \leq 1$, maximizing the social welfare reduces to
$$
\max_{0\leq r_\theta(\cdot)\leq 1}\left\{-\int_0^{\bar L}J_{\theta, \eta}(l)\,r_\theta(l)\,dl\right\}.
$$

\noindent The pointwise maximizer is achieved when
$$
r^*_\theta(l) = 
\begin{cases}
0 & J_{\theta, \eta}(l) >0 , \\
\in [0,1] & J_{\theta, \eta}(l) = 0 , \\
1 & J_{\theta, \eta}(l) < 0 .
\end{cases}
$$

\medskip

Since $R_\theta^*(0)=0$, then $R_\theta^*(l):=\int_0^l r_\theta^*(s)\,ds$. 
When $\eta$ and $\alpha$ are specified, for any $\theta< \theta ^{\prime}$, and if $J_{\theta, \eta}(l)$ is non-decreasing in $\theta$ for all $l$, then on the sets where $\{J_{\theta, \eta}(l) >0\}$ or  $\{J_{\theta, \eta}(l) <0\}$, the pointwise maximization solution is non-increasing in $\theta$. That is, $r^*_{\theta^\prime}(l) \leq r^*_\theta(l)$ for a.e.\ $l \in [0, \bar L]$, or equivalently
$$
\frac{\partial R^* _{\theta^{\prime} } (l)}{\partial l} \leq \frac{\partial R^* _{\theta} (l)}{\partial l}, \ \text{for a.e.\ $l \in [0, \bar L]$.}
$$

\noindent On the set  $\{J_{\theta, \eta}(l) =0\}$,  the value of $r^*_\theta(l)$ is chosen so that $r^*_\theta(l)$ is non-increasing in $\theta$ for almost every $l$. Hence the collection $\{R^*_{\theta}\}_{\theta \in \Theta}$ is submodular.
Substituting $p^*_{\underline{\theta}}$ in the expression of the optimal premium, we obtain
\begin{align}\label{eq:app_premium}
p^*_\theta &= \int_0^{ \bar L}  \left[ 1 - g_{ \underline{\theta}  } \big( F_{\underline{\theta} }(l)  \big)   \right] \, dl 
- \int_{\underline{\theta}} ^ {\theta} \int_0 ^{\bar L}  \left[ \frac{ \partial g_s}{\partial s } (F_s(l)) + g'_s(F_s(l)) \frac{\partial F_s(l)}{\partial s} \right] \frac{\partial R^*_s(l)}{\partial l } \, dl \, ds  \nonumber  \\
&\quad -
\int_0^{ \bar L }  \left[ 1 - g_{\theta} \big( F_{\theta}(l)  \big)   \right] \,\frac{ \partial R^*_{\theta}(l) }{ \partial l } \, dl.  
\end{align}

\medskip

We now us verify that $p^*_\theta \in \R^+$. By the Fundamental Theorem of Calculus, for each $l\in[0,\bar L]$,
$$
1-g_\theta(F_\theta(l))=1-g_{\underline\theta}(F_{\underline\theta}(l))-\int_{\underline\theta}^{\theta} \frac{\partial }{\partial s} (g_s \circ F_s)(l) \,ds.
$$
Multiplying both sides by $\frac{\partial R_\theta^*(l)}{\partial l} \in [0,1]$, we obtain:
$$
\left(1-g_\theta(F_\theta(l))\right) \frac{\partial R_\theta^*(l)}{\partial l} = 
\left( 1-g_{\underline\theta}(F_{\underline\theta}(l))\right)\frac{\partial R_\theta^*(l)}{\partial l} - \left( \int_{\underline\theta}^{\theta} \frac{\partial }{\partial s} (g_s \circ F_s)(l) \,ds \right) \,\frac{\partial R_\theta^*(l)}{\partial l}.
$$

\noindent Substituting this identity into the expression of \(p_\theta^*\), we obtain
$$
p_\theta^*
= \int_0^{\bar L}
[1-g_{\underline\theta}(F_{\underline\theta}(l))]
\left(1-\frac{\partial R_\theta^*(l)}{\partial l}\right)\,dl + \int_{\underline\theta}^{\theta}\int_0^{\bar L}
\frac{\partial }{\partial s} (g_s \circ F_s)(l)
\left(\frac{\partial R_\theta^*(l)}{\partial l}-\frac{\partial R_s^*(l)}{\partial l}
\right) \,dl\,ds.
$$

\noindent The first term is non-negative. Moreover, by Assumptions \ref{Ass:cdf_family} and \ref{Ass:distortion_family}, we have that $\frac{\partial }{\partial s} (g_s \circ F_s)(l)\leq 0$.
Since $\{R_\theta^*\}_{\theta\in\Theta}$ is submodular, for $s\le\theta$,
$$
\frac{\partial R_\theta^*(l)}{\partial l}
- \frac{\partial R_s^*(l)}{\partial l}
\leq 0, \ \text{for a.e. }l.
$$

\noindent The second term is also non-negative, and hence $
p_\theta^*\geq 0$. 

\medskip
Moreover, it follows from Proposition \ref{prop:IC_iff} that $(R^*_\theta, p^*_\theta)_{\theta \in \Theta} \in \cI\cC$. We next verify that $(R^*_\theta, p^*_\theta)_{\theta \in \Theta} \in \cI\cR$. Since the insurer's participation constraint is satisfied by assumption, it suffices to show that the agent's participation constraint (P1) holds for the lowest type $\underline{\theta}$, by Proposition \ref{prop:IR_lowest_type}. We know that
$$
U_{\underline{\theta}} (R^*_{\underline{\theta}}, p^*_{\underline{\theta}}) = -p^*_{\underline{\theta}} - \int_0^{\bar L } \l[ 1- g_{\underline{\theta}}(F_{\underline{\theta}}(l)) \r] \frac{\partial R^*_{\underline{\theta}}(l)}{\partial l} \, dl. 
$$
Substituting $p^*_{\underline{\theta}}$ by its optimal expression yields:
\begin{align*}
U_{\underline{\theta}} (R^*_{\underline{\theta}}, p^*_{\underline{\theta}}) 
&=- \int_0^{ \bar L}  \left[ 1 - g_{ \underline{ \theta }} \big( F_{\underline{ \theta }}(l)  \big)   \right] \, dl  
= U_{\underline{\theta}} (L_{\underline{\theta}}, 0).
\end{align*}

\noindent Hence, $(R^*_\theta, p^*_\theta)_{\theta \in \Theta} \in \cI\cR \cap \cI\cC$.

\medskip

\subsubsection{The case where $\alpha \in \left[ \alpha_0, \frac{1}{2} \right]$} 
It follows from Lemma \ref{le:alpha0} that there exists $\theta_\alpha \in \Theta$ satisfying $\frac{\bar Q_\eta(\theta_\alpha)}
{\bar Q(\theta_\alpha)}=\frac{1-\alpha}{\alpha}$. Equivalently, there exists $\theta_{\alpha} \in \Theta$, such that $1 - \alpha - \alpha \frac{ \bar Q_{\eta} (\theta_{\alpha} )}{\bar Q ( \theta _{\alpha} ) } =0$. 
Since the function $\theta \mapsto  (1-\alpha) -  \alpha \frac{\bar Q_{\eta}(\theta)}{ \bar Q(\theta)}$ is non-increasing, it follows that:

\begin{enumerate}
\item[(i)]  If $ \theta < \theta_{\alpha}$, then $ 1 - \alpha - \alpha \frac{ \bar Q_{\eta} (\theta) }{\bar Q ( \theta) } \geq 0 $. Therefore, if the function $J_{\theta, \eta}(l)$ is non-decreasing in $\theta$, then the optimal marginal retention follows the form given in \eqref{eq:optimal_marginal_retention_Sec:PO}, and the optimal premium $p^*_\theta$ satisfies \eqref{eq:app_premium}. Similarly to the first case discussed above, $\{R^*_\theta\}_{\theta < \theta_\alpha}$ is submodular.

\medskip

\item[(ii)] If $ \theta \geq \theta_{\alpha}$, then $ 1 - \alpha - \alpha \frac{ \bar Q_{\eta} (\theta) }{\bar Q ( \theta) } \leq 0 $. Therefore $J_{\theta, \eta }(l) \geq 0$ for all $l \in [0, \bar L]$, hence the term $-
J_{\theta,\eta}(l)
\frac{\partial R_\theta(l)}{\partial l}$ is maximized by choosing
$$
\frac{\partial R^*_\theta(l)}{\partial l} = 0, \, \text{for a.e. } l \in [0,\bar L].
$$
Since admissible retention functions are absolutely continuous and satisfy \(R^*_\theta(0)=0\), it follows that
$$R^*_\theta(l)=
R^*_\theta(0)+
\int_0^l \frac{\partial R^*_\theta(s)}{\partial s} ds= 0, \, \forall\, l \in [0,\bar L].
$$

Hence, full coverage is optimal for all types $\theta \geq \theta_\alpha$, $\{R^*_\theta\}_{\theta \geq  \theta_\alpha}$ is submodular and $p^*_\theta$ satisfies \eqref{eq:app_premium}.
\end{enumerate} 

\medskip
The collection
$\{R_\theta^*\}_{\theta\in\Theta}$
is submodular.
Since $p^*_\theta$ satisfies \eqref{eq:app_premium}, it follows from Proposition \ref{prop:IC_iff} that $(R^*_\theta, p^*_\theta)_{\theta \in \Theta} \in \cI\cC$. 
It remains to check the individual rationality of $(R^*_\theta, p^*_\theta)_{\theta \in \Theta}$. 
\begin{align*}
U_{\underline{\theta}} (R^*_{\underline{\theta}}, p^*_{\underline{\theta}}) 
&=-p^*_{\underline{\theta}} - \int_0^{\bar L } \l[ 1- g_{\underline{\theta}}(F_{\underline{\theta}}(l)) \r] \frac{\partial R^*_{\underline{\theta}}(l)}{\partial l} \, dl \\
&= -  \int_0^{ \bar L}  \left[ 1 - g_{ \underline{\theta}  } \big( F_{\underline{\theta} }(l)  \big)   \right] \, dl + \int_0^{\bar L } \l[ 1- g_{\underline{\theta}}(F_{\underline{\theta}}(l)) \r] \frac{\partial R^*_{\underline{\theta}}(l)}{\partial l} \, dl - \int_0^{\bar L } \l[ 1- g_{\underline{\theta}}(F_{\underline{\theta}}(l)) \r] \frac{\partial R^*_{\underline{\theta}}(l)}{\partial l} \, dl \\
&=  U_{\underline{\theta}} (L_{\underline{\theta}}, 0).
\end{align*}

\noindent Since  the agent's participation constraint is satisfied for the lowest type $\underline{\theta}$ and the insurer's participation constraint holds by assumption, it follows from Proposition \ref{prop:IR_lowest_type} that $(R^*_\theta, p^*_\theta)_{\theta \in \Theta} \in \cI\cR$. 

\medskip

\subsubsection{The case where $\alpha > \frac{1}{2}$}
Recall from \eqref{eq:social_welfare} that
\begin{align*}
W_{\eta, \alpha}\left((R_\theta, p_\theta)_{\theta \in \Theta}\right)
&= (1 - 2\alpha) \left[ p_{\underline\theta} + \int_0^{\bar L} \left[1 - g_{\underline\theta}(F_{\underline\theta}(l))\right] \frac{\partial R_{\underline\theta}(l)}{\partial l} \, dl \right] \\
&\quad - \int_\Theta \int_0^{\bar L} J_{\theta, \eta}(l) \, \frac{\partial R_\theta(l)}{\partial l} \, dl \,  d\mu(\theta) \\
&\quad - (1 - \alpha) \int_\Theta \int_0^{\bar L} \left[1 - g^{In}(F_\theta(l))\right] dl \, d\mu(\theta),
\end{align*}
where $J_{\theta, \eta}$ is given in \eqref{eq:Jeta}:
$$
J_{\theta, \eta}(l) = (1 - \alpha) \left[g^{In}(F_\theta(l)) - g_\theta(F_\theta(l))\right] + \left(\frac{\bar Q(\theta)}{q(\theta)}\right) \left[\frac{\partial g_\theta}{\partial \theta}(F_\theta(l)) + g'_\theta(F_\theta(l)) \frac{\partial F_\theta(l)}{\partial \theta}\right] \left[(1 - \alpha) - \alpha \frac{\bar Q_\eta(\theta)}{\bar Q(\theta)}\right].
$$

\medskip

\noindent The last term of $W_{\eta, \alpha}$ is independent of the choice of retention and premium. Since $\alpha > \frac{1}{2}$, we have $1 - 2\alpha < 0$, and by definition $1 - g_{\underline\theta}(F_{\underline\theta}(l)) \geq 0$ for every $l \in [0, \bar L]$. Therefore the term
$$
(1 - 2\alpha) \int_0^{\bar L} \left[1 - g_{\underline\theta}(F_{\underline\theta}(l))\right] \frac{\partial R_{\underline\theta}(l)}{\partial l} \, dl
$$

\noindent is pointwise maximized by choosing $\frac{\partial R_{\underline\theta}(l)}{\partial l} = 0$, for every $l \in [0, \bar L]$.

\medskip

Moreover, since $\alpha > \frac{1}{2}$ and $\bar Q_\eta(\theta)/\bar Q(\theta) \geq 1$, we have
$$
(1 - \alpha) - \alpha \frac{\bar Q_\eta(\theta)}{\bar Q(\theta)} < 0, \ \forall \, \theta \in \Theta.
$$

\noindent Under Assumption \ref{Ass:cdf_family}-\eqref{ass:F_i2} and Assumption \ref{Ass:distortion_family}-\eqref{ass:g_i3},
$$
\frac{\partial g_\theta}{\partial \theta}(F_\theta(l)) + g'_\theta(F_\theta(l)) \frac{\partial F_\theta(l)}{\partial \theta} \leq 0,
$$

\noindent and since $\bar Q(\theta)/q(\theta) \geq 0$, the second term in $J_{\theta, \eta}(l)$ is nonnegative. The first term is also nonnegative by Assumption \ref{ass:distortion_insurer_agent}. Consequently,
$$
J_{\theta, \eta}(l) \geq 0, \ \forall \, \theta \in \Theta, \ \forall \, l \in [0, \bar L].
$$

\noindent Since admissible retention functions satisfy $0 \leq \frac{\partial R_\theta(l)}{\partial l} \leq 1$, a.e., the term $-
J_{\theta,\eta}(l)
\frac{\partial R_\theta(l)}{\partial l}$ is maximized by choosing $\frac{\partial R_\theta^*(l)}{\partial l}=0$. That is, the pointwise map
$$
r \mapsto - J_{\theta, \eta}(l) \, r, \ \hbox{ for } r \in [0, 1],
$$

\noindent is maximized by $r = 0$, for every $\theta \in \Theta$ and every $l \in [0, \bar L]$.

\medskip

Since each admissible retention function is absolutely continuous, and since $R_\theta^*(0)=0$, it follows that for all $\theta\in\Theta$ and all $l\in[0,\bar L]$, we have
$$
R^*_\theta(l)
=
R^*_\theta(0)
+
\int_0^l
\frac{\partial R^*_\theta(s)}{\partial s}\,ds
=
0.
$$

\medskip

It remains to characterize the premium schedule. Since $R_\theta^*=0$ for every $\theta \in\Theta$, it follows that for any $\theta,\theta'\in\Theta$, we have
$$
U_\theta(R_\theta^*,p_\theta^*)
=
U_\theta(0,p_\theta^*)
=
-p_\theta^*
\ \ \hbox{and} \ \ U_\theta(R_{\theta'}^*,p_{\theta'}^*)
=
U_\theta(0,p_{\theta'}^*)
=
-p_{\theta'}^*.
$$

\noindent By incentive compatibility, we have $U_\theta(0,p_\theta^*)
\geq U_\theta(0,p_{\theta'}^*)$, that is, $p_\theta^*\leq p_{\theta'}^*$. Reversing the roles of $\theta$ and $\theta'$ gives $p_{\theta'}^*\leq p_\theta^*$, and so
$$
p_\theta^*=p_{\theta'}^*,
\ \ \forall\,\theta,\theta'\in\Theta.
$$
Thus, there exists some  $p^*\in\bbR_+$ such that
$$
p_\theta^*=p^*,
\ \
\forall\,\theta\in\Theta.
$$
Letting
$$
K
:=
\int_\Theta\int_0^{\bar L}
\left[
1-g^{In}(F_\theta(l))
\right]\,dl\,d\mu(\theta),
$$
we obtain
$$
\int_\Theta V_\theta(0,p^*)\,d\mu(\theta)
=
\int_\Theta \l\{p^*
-
\int_0^{\bar L}
\left[
1-g^{In}(F_\theta(l))
\right]\,dl\r\} \,d\mu(\theta)
=
p^*-K,
$$

\noindent since $\mu(\Theta)=1$. Therefore, the insurer's participation constraint is equivalent to
$$p^* \geq K.$$

\medskip

\noindent On the other hand, under full coverage and pooling premia, each type $\theta$ receives the utility
$$U_{\theta} (0, p^*)
=- p^*
$$
Therefore, the participation constraint of a type-$\theta$ agent becomes
$$-p^* 
\geq 
U_{\theta} (L_{\theta}, 0)
=  - \int_0^{ \bar L }  \left [ 1 - g_{\theta}  \left( F_{\theta}(l) \right) \right] \ dl,
$$
or
$$p^* 
\leq 
\int_0^{ \bar L }  \left [ 1 - g_{\theta}  \left( F_{\theta}(l) \right) \right] \ dl.
$$
Let 
$$
M
:= \int_0^{\bar L} \left[1 - g_{\underline\theta}(F_{\underline\theta}(l))\right] \, dl,
$$
By Assumptions \ref{Ass:cdf_family}-(2) and \ref{Ass:distortion_family}-(3), a sufficient condition for the above inequality to hold for each $\theta \in \Theta$ is:
$$p^* 
\leq 
M
$$
Now, by hypothesis, we have $K \leq M$, and so any choice of $p^*$ in $[K,M]$ leads to an IC-IR menu $(0, p^*)_{\theta \in \Theta}$.

\medskip

We now determine which feasible pooling premium is optimal. Let
$$
W^{FC}_{\eta,\alpha}(p)
:=
W_{\eta,\alpha}\bigl((0,p)_{\theta\in\Theta}\bigr),
\ \ \forall \, p\in[K,M].
$$
For the menu \((0,p)_{\theta\in\Theta}\), we have
$$
U_\theta(0,p)=-p,
\ \ \forall\,\theta\in\Theta,
$$
and therefore
$$
\int_\Theta U_\theta(0,p)\,d\eta(\theta)=-p,
$$
because $\eta(\Theta)=1$. On the insurer's side, we have $
V_\theta(0,p)
=
p-\int_0^{\bar L}\left[1-g^{In}(F_\theta(l))\right]\,dl,
$
and so
$$
\int_\Theta V_\theta(0,p)\,d\mu(\theta)
=
p-K.
$$
Hence
\begin{align*}
W^{FC}_{\eta,\alpha}(p)
&=
\alpha\int_\Theta U_\theta(0,p)\,d\eta(\theta)
+
(1-\alpha)\int_\Theta V_\theta(0,p)\,d\mu(\theta) \\
&=
-\alpha p+(1-\alpha)(p-K)
=
(1-2\alpha)p-(1-\alpha)K.
\end{align*}

\noindent Since $\alpha>\frac{1}{2}$, $W^{FC}_{\eta,\alpha}(p)$ is strictly decreasing in $p$. It follows that, among all feasible full-coverage pooling premia $p\in[K,M]$, the unique welfare-maximizing premium is the smallest feasible one, namely $p=K$. Thus the optimal full-coverage pooling menu is
$$
(R_\theta^*,p_\theta^*)_{\theta\in\Theta}
=
(0,K)_{\theta\in\Theta}.
$$
Since all types receive the same full-coverage contract $(0, p^*)$, then for every $\theta \in \Theta$, $U_\theta (0, p^*) =- p^*$. Hence, incentive compatibility of the optimal menu holds trivially. At this premium,
$$
\int_\Theta V_\theta(0,p^*)\,d\mu(\theta)
=
p^*-K
=
0,
$$
and so the insurer's participation constraint binds. It remains to verify that the agent's participation holds so that individual rationality of the optimal menu is satisfied. We have that
$$
U_\theta(0, p^*) = -p^* = - \int_\Theta\int_0^{\bar L}
\left[
1-g^{In}(F_\theta(l))
\right]\,dl\,d\mu(\theta)\geq - \int_0^{\bar L} \left[1 - g_{\underline\theta}(F_{\underline\theta}(l))\right] \, dl . 
$$

Since the composite function $g_\theta \circ F_\theta$ is non-increasing in $\theta$ by Remark \ref{Re:chain_rule}, it follows that $g_\theta(F_\theta(l)) \leq g_{\underline{\theta}}(F_{\underline{\theta}}(l))$ for all $\theta \in \Theta$. Hence, 
$$
U_\theta(0, p^*) \geq - \int_0^{\bar L} \left[1 - g_{\underline\theta}(F_{\underline\theta}(l))\right] \, dl \geq - \int_0^{\bar L} \left[1 - g_{\theta}(F_{\theta}(l))\right] \, dl
= U_\theta(L_\theta,0), 
$$
which completes the proof.\qed

\bigskip
\subsection{Proof of Lemma \ref{le:IRimplication}} \label{App:proofleIRimplication}

Consider $(R^*_\theta, p^*_\theta)_{\theta \in \Theta}$ characterized in Theorem \ref{th:solution_characterization_IPO}. The insurer's aggregate utility is given by:
\begin{align*}
\int_\Theta V_\theta (R^*_\theta , p^*_\theta) \, d\mu(\theta) =
\int_\Theta p^*_\theta \,d\mu(\theta) - \int_\Theta \int_0^{\bar L} \left[ 1 - g^{In}(F_\theta(l))\right] \left( 1- r^*_\theta(l) \right) \,dl \, d\mu(\theta), 
\end{align*}

\noindent where the optimal premium $p^*_\theta$ satisfies \eqref{eq:optimal_premium}. We consider the following three cases for the optimal marginal retention function.

\medskip

\begin{enumerate}
\item For $r^*_\theta(l)\in [0,1]$ for a.e. $l \in [0, \bar L]$, the insurer's participation constraint is satisfied if the following holds
\begin{align*}
\int_\Theta \int_0^{\bar L} \left[ 1 - g^{In}(F_\theta(l))\right] \left( 1- r^*_\theta(l) \right) \,dl \, d\mu(\theta)\leq 
\int_\Theta p^*_\theta \,d\mu(\theta),
\end{align*}
where,
\begin{align*}
\int_\Theta p^*_\theta \, d\mu (\theta) 
&= \int_0^{ \bar L}  \left[ 1 - g_{ \underline{\theta}  } \big(F_{\underline{\theta} }(l)  \big)\right] \, dl  -\int_\Theta\int_{\underline{\theta}} ^ {\theta} \int_0 ^{\bar L} \left[ \frac{ \partial g_s}{\partial s } (F_s(l)) + g'_s(F_s(l)) \frac{\partial F_s(l)}{\partial s} \right]r^*_s(l) \, dl \, ds \, d\mu(\theta)  \\
&\quad -\int_\Theta\int_0^{ \bar L }  \left[ 1 - g_{\theta} \big( F_{\theta}(l)  \big)   \right] \, r^*_\theta(l) \, dl \, d\mu(\theta). 
\end{align*}

\medskip
\item For $r^*_\theta \equiv 0$, substituting $p^*_\theta$ by \eqref{eq:optimal_premium}, the insurer's aggregate utility reduces to
$$
\int_\Theta V_\theta (R^*_\theta , p^*_\theta) \, d\mu(\theta) = \int_0^{ \bar L}  \left[ 1 - g_{ \underline{\theta}  } \big(F_{\underline{\theta} }(l)  \big)\right] \, dl - \int_\Theta \int_0^{\bar L} \left[ 1 - g^{In}(F_\theta(l))\right]\, dl \, d\mu(\theta).
$$

If the following condition holds 
$$
\int_0^{ \bar L}  \left[ 1 - g_{ \underline{\theta}  } \big(F_{\underline{\theta} }(l)  \big)\right] \, dl \geq \int_\Theta \int_0^{\bar L} \left[ 1 - g^{In}(F_\theta(l))\right]\, dl \, d\mu(\theta), 
$$

\noindent then the insurer's participation constraint is satisfied. That is, $\int_\Theta V_\theta (R^*_\theta , p^*_\theta) \, d\mu(\theta) \geq 0$. 

\medskip

\item For $r^*_\theta \equiv 1$, then using the optimal premium $p^*_\theta$ given in \eqref{eq:optimal_premium} for $\theta \in \Theta$, the insurer's utility reduces to the following: 
\begin{align*}
V_\theta (R^*_\theta , p^*_\theta) 
& = \int_0^{ \bar L}  \left[ g_\theta(F_\theta(l)) - g_{ \underline{\theta}  } \big(F_{\underline{\theta} }(l)  \big)\right] \, dl -\int_{\underline{\theta}} ^ {\theta} \int_0 ^{\bar L} \left[ \frac{ \partial g_s}{\partial s } (F_s(l)) + g'_s(F_s(l)) \frac{\partial F_s(l)}{\partial s} \right] \, dl \, ds \\
&= \int_0^{ \bar L}  \int_{\underline{\theta}}^{\theta} \frac{\partial }{\partial s} (g_s \circ F_s)(l) ds \, dl  -\int_{\underline{\theta}} ^ {\theta} \int_0 ^{\bar L} \left[ \frac{ \partial g_s}{\partial s } (F_s(l)) + g'_s(F_s(l)) \frac{\partial F_s(l)}{\partial s} \right] \, dl \, ds \\
&= 0. 
\end{align*}

Hence the insurer's participation constraint binds. That is, $\int_\Theta V_\theta (R^*_\theta, p^*_\theta) \, d\mu(\theta) =0$. \qed
\end{enumerate}

\bigskip
\subsection{Proof of Proposition \ref{prop:insurerutilityoptimalmenu}} \label{App:proof_Vthetaoptimum}

Consider the optimal menu of contracts $(R^*_\theta, p^*_\theta)_{\theta \in \Theta}$, characterized in Theorem \ref{th:solution_characterization_IPO}. We know that
\begin{align*}
V_\theta (R^*_\theta,p^*_\theta) = p^*_\theta - \int_0^{\bar L} \left[1-g^{In}(F_\theta(l))\right] \left(1 - r^*_\theta(l) \right) \, dl,
\end{align*}
where $r^*_\theta(l): = \frac{\partial R^*_\theta(l)}{\partial l}$ for all $\theta \in \Theta$. Consider $\theta, \theta^\prime \in \Theta$ such that $\theta < \theta^\prime$. First, using the Fundamental Theorem of Calculus, 
\begin{align*}
p_{\theta'}^*-p_\theta^*
&=
\int_0^{\bar L}
\left[ 1- g_\theta(F_\theta(l)) \right] \left(r_\theta^*(l)-r_{\theta'}^*(l)\right)\,dl \\ 
&\quad +
\int_\theta^{\theta'}\int_0^{\bar L}
\left[\frac{\partial g_s}{\partial s}(F_s(l))
+
g_s'(F_s(l))\frac{\partial F_s(l)}{\partial s}\right]\left(r_{\theta'}^*(l)-r_s^*(l)\right)\,dl\,ds .
\end{align*}

\noindent Moreover,
\begin{align*}
V_{\theta'}(R_{\theta'}^*,p_{\theta'}^*)
-
V_\theta(R_\theta^*,p_\theta^*)
&=
\int_0^{\bar L}
\left[g^{In}(F_\theta(l))-g_\theta(F_\theta(l))\right]
\left(r_\theta^*(l)-r_{\theta'}^*(l)\right)\,dl \\
&\quad+
\int_\theta^{\theta'}\int_0^{\bar L}
\left[\frac{\partial g_s}{\partial s}(F_s(l))
+
g_s'(F_s(l))\frac{\partial F_s(l)}{\partial s}\right]\left(r_{\theta'}^*(l)-r_s^*(l)\right)\,dl\,ds \\
&\quad-
\int_\theta^{\theta'}\int_0^{\bar L}
g^{In\prime}(F_s(l))
\left(-\frac{\partial F_s(l)}{\partial s}\right)
\left(1-r_{\theta'}^*(l)\right)\,dl\,ds.
\end{align*}

\noindent Consider now the following cases.

\smallskip
\begin{enumerate}
\item If $r^*_\theta(l) \in [0,1]$, for a.e.\ $l \in [0,\bar L]$, we know that $g^{In}(F_\theta(l))-g_\theta(F_\theta(l)) \geq 0$ by Assumption \ref{ass:distortion_insurer_agent}, the collection $\{R^*_\theta\}_{\theta \in \Theta}$ is submodular, and the terms $\frac{\partial }{\partial s} (g_s \circ F_s)(l)$ and $\frac{\partial F_s(l)}{\partial s}$ are non-positive. If the following condition holds
\begin{align*}
&\int_0^{\bar L}
[g^{In}(F_\theta(l))-g_\theta(F_\theta(l))]
(r_\theta^*(l)-r_{\theta'}^*(l))\,dl +
\int_\theta^{\theta'}
\int_0^{\bar L}
\frac{\partial }{\partial s} (g_s \circ F_s)(l)
(r_{\theta'}^*(l)-r_s^*(l))
\,dl\,ds \\
&\geq
\int_\theta^{\theta'}
\int_0^{\bar L}
g^{In\prime}(F_s(l))
\left(-\frac{\partial F_s(l)}{\partial s}\right)
(1-r_{\theta'}^*(l))
\,dl\,ds,
\end{align*}

\noindent then $V_{\theta'}(R^*_{\theta'},p^*_{\theta'})\geq V_\theta(R^*_\theta,p^*_\theta)$, for any $\theta, \theta^\prime \in \Theta$ such that $\theta<\theta^\prime$. Hence the insurer's utility is non-decreasing in $\theta$.

\medskip
\item If $r_\theta^* \equiv 0$ for all $\theta$, then the first two terms vanish and the last term is non-positive. Hence the insurer's utility is non-increasing in $\theta$.

\medskip
\item If $r_\theta^*\equiv 1$ for all $\theta$, then by Lemma \ref{le:IRimplication}, $V_\theta(R_\theta^*,p_\theta^*)=0$ for all $\theta \in \Theta$.
\end{enumerate}

\bigskip
\subsection{Proof of Proposition \ref{prop:properties_optimalmenu}}\label{App:proofpropoptimalmenu}
Consider the optimal menu of contracts $(R^*_\theta, p^*_\theta)_{\theta \in \Theta}$, characterized in Theorem \ref{th:solution_characterization_IPO}. 

\medskip

\noindent (1) Since $\{R^*_\theta\}_{\theta \in \Theta}$ is submodular, it follows that $R^*_\theta$ decreases with $\theta$. 
Moreover, consider $\theta, \theta^\prime \in \Theta$ such that $\theta<\theta^\prime$. It follows from incentive compatibility that:
$$
U_\theta(R_\theta^*,p_\theta^*)
\geq
U_\theta(R_{\theta^\prime}^*,p_{\theta^\prime}^*).
$$
Hence,
$$
-p_\theta^*
-\int_0^{\bar L}
\left[1-g_\theta(F_\theta(l))\right]
r^*_\theta(l)\,dl
\geq
-p_{\theta^\prime}^*
-\int_0^{\bar L}
\left[1-g_\theta(F_\theta(l))\right]
r^*_{\theta^\prime}(l)\,dl.
$$

\noindent Rearranging, we obtain
$$
p_{\theta'}^*-p_\theta^*
\geq
\int_0^{\bar L}
\left[1-g_\theta(F_\theta(l))\right]
\left(r^*_\theta(l) - r^*_{\theta^\prime}(l) \right)\,dl.
$$

Since $\{R_\theta^*\}_{\theta\in\Theta}$ is submodular, and since $1-g_\theta(F_\theta(l))\geq 0$, we obtain $p_{\theta'}^*-p_\theta^*\geq 0$.
Hence $p_\theta^*$ is non-decreasing in $\theta$.

\bigskip

\noindent (2) For $\theta = \bar \theta $, by Assumption \ref{ass:distortion_insurer_agent}, we have:
$$
J_{\theta , \eta}(l) \bigg|_{\theta = \bar \theta}=  (1-\alpha) \left[  g^{In} (F_{\bar \theta}(l)) -  g_{\bar \theta} (F_{ \bar \theta}(l))\right] \geq 0, \ \forall l \in [0, \bar L],
$$

Suppose that $F_{\bar\theta}(l)\in(0,1)$ for a.e.\ $l \in (0, \bar L)$, and
$g^{In} (t) >  g_{\bar \theta} (t)$ for all $t \in (0, 1)$. Then $J_{\theta, \eta}(l)\big|_{\theta = \bar \theta} >0$ for a.e.\ $l\in(0, \bar L)$, which implies that
$r^*_{\bar \theta}(l) =0$ for a.e.\ $l \in (0, \bar L)$. Moreover, for every $l \in [0, \bar L]$, 
$$
R^*_{\bar \theta}(l) = \int_0^l r^*_{\bar \theta}(s) \,ds =0.
$$

\bigskip

\noindent (3) For $\theta = \underline{\theta}$, the optimal premium given in \eqref{eq:optimal_premium} reduces to
\begin{align*}
p^*_{\underline{\theta}} &= \int_0^{ \bar L}  \left[ 1 - g_{ \underline{\theta}  } \big( F_{\underline{\theta} }(l)  \big)\right] \left[1 - r^*_{\underline{\theta}}(l)\right]\, dl.  
\end{align*}

\noindent Hence, the lowest type's end-of-period utility is given by
\begin{align*}
U_{\underline{\theta}}(R^*_{\underline{\theta}}, p^*_{\underline{\theta}}) 
&= - p^*_{\underline{\theta}} - \int_0^{ \bar L} \left[ 1 - g_{\underline{\theta}} \big( F_{\underline{\theta}}(l)  \big) \right]\, r^*_{\underline{\theta}}(l) \, dl 
= - \int_0^{ \bar L}  \left[ 1 - g_{ \underline{\theta}  } \big( F_{\underline{\theta} }(l)  \big)\right]\, dl 
= U_{\underline{\theta}} ( L_{\underline{\theta}}, 0). 
\end{align*}

\bigskip

\noindent (4) It follows from the envelope theorem that the partial derivative of $U_\theta (R^*_\theta, p^*_\theta)$ with respect to $\theta$ is given by
\begin{align*}
\Phi(\theta) : =\frac{\partial U_\theta (R^*_\theta, p^*_\theta)}{\partial \theta}
&= \int_0^{\bar L}  \left[ \frac{ \partial g_\theta}{\partial \theta } (F_\theta(l)) + g'_\theta(F_\theta(l))\frac{\partial F_\theta(l)}{\partial \theta}\right] r^*_\theta(l) \, dl \leq 0.
\end{align*}

\noindent Let 
$$
A_\theta(l) = \frac{\partial }{\partial \theta} (g_\theta \circ F_\theta)(l),  \ \text{for all $l \in [0, \bar L]$}.
$$

\noindent Consider $\theta,\theta^\prime \in \Theta$ such that $\theta < \theta^\prime$. We have:
\begin{align*}
\Phi (\theta) - \Phi (\theta^\prime)
&=\int_0^{\bar L} A_\theta(l) \, r^*_\theta(l)\, dl - \int_0^{\bar L} A_{\theta ^\prime}(l) \, r^*_{\theta^\prime}(l) \, dl .
\end{align*}

\noindent Since $\{R^*_\theta\}_{\theta \in \Theta}$ is submodular and $A_\theta(l) \leq 0$ by Remark \ref{Re:chain_rule}, then 
$A_\theta(l)\,r^*_\theta(l) \leq A_\theta(l) \, r^*_{\theta^\prime} (l)$ for a.e. $l \in [0, \bar L]$.
Hence, 
\begin{align*}
\Phi(\theta) - \Phi(\theta^\prime)
&\leq \int_0^{\bar L} \left[ A_\theta(l) - A_{\theta^\prime}(l)\right] \, r^*_{\theta^\prime} (l) \, dl.
\end{align*}

\noindent Moreover, 
\begin{align*}
\frac{\partial}{\partial \theta}A_\theta(l)= \frac{\partial^2}{\partial \theta^2} (g_\theta \circ F_\theta)(l)
&=\frac{\partial }{\partial \theta} \left[ \frac{ \partial g_\theta}{\partial \theta } (F_\theta(l)) + g'_\theta(F_\theta(l))\frac{\partial F_\theta(l)}{\partial \theta} \right] \\
&=
\frac{\partial^2 g_\theta}{\partial \theta^2}(F_\theta(l))
+
2\frac{\partial^2 g_\theta}{\partial \theta \, \partial t}(F_\theta(l))
\frac{\partial F_\theta(l)}{\partial \theta}
+
g_\theta''(F_\theta(l))
\left(\frac{\partial F_\theta(l)}{\partial \theta}\right)^2
+
g_\theta'(F_\theta(l))
\frac{\partial^2 F_\theta(l)}{\partial \theta^2}.
\end{align*}

\noindent Suppose that the following conditions hold:
\medskip
\begin{enumerate}
\item[(i)] The function $g : (\theta, t) \mapsto g_\theta (t)$ is submodular. That is,
$$
\frac{\partial ^2 g_\theta (t)}{\partial \theta \partial t} \leq 0. 
$$

\item[(ii)] The function $\theta \mapsto g_{\theta}(t)$ is convex in $\theta$, for all $t$. That is,
$$
\frac{\partial ^2 g_\theta(t)}{\partial \theta ^2} \geq 0, \ \forall t. 
$$

\item[(iii)] The function $t \mapsto g_{\theta}(t)$ is convex in $t$, for all $\theta \in \Theta$. That is, 
$$
\frac{\partial ^2 g_\theta(t)}{\partial t ^2} \geq 0, \ \forall \theta. 
$$

\item[(iv)] The function $\theta \mapsto F_\theta(l)$ is convex in $\theta$, for all $l$. That is, 
$$
\frac{\partial ^2 F_\theta(l)}{\partial \theta ^2} \geq 0, \, \forall l. 
$$
\end{enumerate} 

It follows from the above conditions that the composite function $(g_\theta \circ F_\theta) (l)$ is convex in $\theta$ for all $l$. Therefore its derivative in $\theta$ is non-decreasing in $\theta$ for all $l$. Specifically, for $\theta < \theta^\prime$, $A_\theta(l) \leq A_{\theta^\prime}(l)$,  for all $l \in [0, \bar L]$, or equivalently,
$$\frac{\partial }{\partial \theta} (g_\theta \circ F_\theta)(l)  \leq \frac{\partial }{\partial \theta} (g_{\theta^\prime} \circ F_{\theta ^\prime})(l), \ \forall  l \in [0, \bar L].$$

\noindent Since $r^*_\theta(l) \in [0,1]$ for all $\theta \in \Theta$ and for almost every $l \in [0, \bar L]$, we obtain $\Phi(\theta) - \Phi(\theta^\prime) \leq 0$, or equivalently:
\begin{align*}
\frac{\partial U_\theta (R^*_\theta, p^*_\theta)}{\partial \theta} - \left[ \frac{\partial U_\theta (R^*_\theta, p^*_\theta)}{\partial \theta} \right]_{\theta = \theta^\prime}
&\leq 0.
\end{align*}

\noindent That is, $\theta \mapsto U_\theta (R^*_\theta, p^*_\theta)$ is convex if conditions (i) to (iv) hold. 
\qed

\bigskip
\subsection{Proof of Proposition \ref{prop:insureraggutility}}
First, we know that 
$$
\int_\Theta V_\theta(R_\theta, p_\theta) \, d\mu(\theta) = \int_\Theta \l[ p_{\theta} - \int_0^{\bar L} \left[ 1 - g^{In} (F_{\theta}(l))\right] \left( 1 - \frac{\partial R_{\theta}(l)}{\partial l }\right) \, dl \r] \, d\mu(\theta).
$$

\noindent Since $(R_{\theta}, p_{\theta})_{\theta \in \Theta} \in \cI\cC$, then substituting the premium by \eqref{eq:premium} we obtain:
\begin{align*}
\int_\Theta V_\theta(R_\theta, p_\theta )\, d\mu(\theta)
&= p_{\underline{\theta} } + \int_0^{ \bar L}  \left[ 1 - g_{ \underline{ \theta }} \big( F_{\underline{ \theta} } (l)  \big)   \right] \,\, \frac{ \partial R_{\underline{ \theta }}(l) }{ \partial l } \, dl \\
&\quad - \int_\Theta\int_{\underline{\theta}} ^ {\theta} \int_0 ^{\bar L}  \left[ \frac{\partial g_s}{ \partial s} (F_s(l)) + g'_s(F_s(l)) \frac{\partial F_s(l)}{\partial s} \right] \frac{\partial R_s(l)}{\partial l } \, dl \, ds\,d\mu (\theta) \nonumber \\
& \quad -
\int_\Theta \int_0^{ \bar L }  \left[ 1 - g_{\theta} \big( F_{\theta}(l)  \big)   \right] \,\, \frac{ \partial R_{\theta}(l) }{ \partial l } \, dl\,d\mu(\theta) \\
&\quad - \int_\Theta \int_0^{\bar L} \left[ 1 - g^{In} (F_{\theta}(l))\right] \left( 1 - \frac{\partial R_{\theta}(l)}{\partial l }\right) \, dl  \, d\mu(\theta). 
\end{align*}

Integrating by parts similarly to the proof of Theorem \ref{th:solution_characterization_IPO} yields the desired expression of $\int_\Theta V_\theta(R_\theta, p_\theta )\, d\mu(\theta)$. \qed

\bigskip
\subsection{Proof of Theorem \ref{th:solution_characterization}}\label{App:proofthsolcharact2}
By Proposition \ref{prop:insureraggutility}, for every incentive compatible menu $(R_\theta,p_\theta)_{\theta\in\Theta}\in \cI\cC$, the insurer's aggregate utility can be written as  \begin{align*}
\int_\Theta V_\theta(R_\theta,p_\theta)\,d\mu(\theta)
&=
p_{\underline\theta}+
\int_0^{\bar L}
[1-g_{\underline\theta}(F_{\underline\theta}(l))]
\frac{\partial R_{\underline\theta}(l)}{\partial l}\,dl \\
&\quad-
\int_\Theta\int_0^{\bar L}
[1-g^{In}(F_\theta(l))]\,dl\,d\mu(\theta)
-\int_\Theta\int_0^{\bar L}J_\theta^I(l)
\frac{\partial R_\theta(l)}{\partial l}
\,dl\,d\mu(\theta).
\end{align*}

The first term is increasing in $p_{\underline\theta}$. Hence at the optimum, $p_{\underline\theta}$ is chosen at its largest value compatible with the lowest type's participation constraint. By Proposition \ref{IR_characterization}, 
$$
p_{\underline\theta}^*=\int_0^{\bar L}
[1-g_{\underline\theta}(F_{\underline\theta}(l))]
\left(1-\frac{\partial R_{\underline\theta}^*(l)}{\partial l}\right)\,dl.
$$

\noindent Substituting this expression into the aggregate utility gives
\begin{align}
\int_\Theta V_\theta(R_\theta,p_\theta)\,d\mu(\theta)
&=
\int_0^{\bar L}
[1-g_{\underline\theta}(F_{\underline\theta}(l))]\,dl
-\int_\Theta\int_0^{\bar L}
[1-g^{In}(F_\theta(l))]\,dl\,d\mu(\theta) \nonumber\\
&\quad -\int_\Theta\int_0^{\bar L}J_\theta^I(l)\frac{\partial R_\theta(l)}{\partial l}\,dl\,d\mu(\theta).
\label{eq:insurer_objective_reduced}
\end{align}

\noindent Therefore, maximizing the insurer's aggregate utility reduces to maximizing
$$
-\int_\Theta\int_0^{\bar L}J_\theta^I(l)
\frac{\partial R_\theta(l)}{\partial l}\,dl\,d\mu(\theta).
$$

For fixed $\theta\in\Theta$, let
$r_\theta(l):=\frac{\partial R_\theta(l)}{\partial l}$. Since admissible retention functions satisfy $0 \leq r_\theta(l) \leq 1$, maximizing the insurer's aggregate utility reduces to
$$
\max_{0\leq r_\theta(\cdot)\leq 1}\left\{-\int_0^{\bar L}J_\theta^I(l)\,r_\theta(l)\,dl\right\}.
$$

The pointwise maximizer is achieved when 
$$
r_\theta^*(l) =
\begin{cases}
0, & J_\theta^I(l)>0,\\
\in[0,1], & J_\theta^I(l)=0,\\
1, & J_\theta^I(l)<0.
\end{cases}
$$

Let $R_\theta^*(l):=\int_0^l r_\theta^*(s)\,ds$, then $R_\theta^*$ is absolutely continuous, $R_\theta^*(0)=0$ and $0\leq \frac{\partial R_\theta^*(l)}{\partial l}\leq1$. Hence, $R^*_\theta \in \cR$. By assumption, $J_\theta^I(l)$ is non-decreasing in $\theta$ for all $l$. On the sets where $\{J_\theta ^I(l) <0\}$ and $\{ J_\theta ^I(l) >0\}$, the pointwise maximizer $r^*_\theta(l)$ is non-increasing in $\theta$. On the set $\{J_\theta ^I(l) =0\}$, the value of $r^*_\theta(l)$ is chosen so that $r^*_\theta(l)$ is non-increasing in $\theta$ for almost every $l$. Therefore, for $\theta, \theta^\prime \in \Theta$ such that $\theta<\theta^\prime$, we have:
$$r_{\theta'}^*(l)\leq r_\theta^*(l),
\ \text{for a.e. } l\in[0,\bar L].$$

\noindent Hence the collection $\{R_\theta^*\}_{\theta\in\Theta}$ is submodular. The optimal premium satisfies
\begin{align*}
p^*_{\theta} 
&= \int_0^{ \bar L}  \left[ 1 - g_{ \underline{\theta}  } \big( F_{\underline{\theta} }(l)  \big)   \right] \, dl  - \int_{\underline{\theta}} ^ {\theta} \int_0 ^{\bar L}  \left[ \frac{ \partial g_s}{\partial s } (F_s(l)) + g'_s(F_s(l)) \frac{\partial F_s(l)}{\partial s}  \right] \frac{\partial R^*_s(l)}{\partial l } \, dl \, ds\\
&\quad -
\int_0^{ \bar L }  \left[ 1 - g_{\theta} \big( F_{\theta}(l)  \big)   \right] \, \frac{ \partial R^*_{\theta}(l) }{ \partial l } \, dl.  
\end{align*}

Similarly to the proof of Theorem \ref{th:solution_characterization_IPO}, we can show that $p^*_\theta \geq 0$ for all $\theta \in \Theta$. Additionally, by the submodularity of $\{R_\theta^*\}_{\theta\in\Theta}$ and since $p^*_\theta$ satisfies \eqref{eq:optimal_premium}, it follows from Proposition \ref{prop:IC_iff} that $(R_\theta^*,p_\theta^*)_{\theta\in\Theta}\in\mathcal{IC}$. Moreover,  
\begin{align*}
U_{\underline{\theta}} (R^*_{\underline{\theta}}, p^*_{\underline{\theta}}) 
&=- \int_0^{ \bar L}  \left[ 1 - g_{ \underline{ \theta }} \big( F_{\underline{ \theta }}(l)  \big)   \right] \, dl  
= U_{\underline{\theta}} (L_{\underline{\theta}}, 0),
\end{align*}
and $\int_\Theta V_\theta (R^*_\theta, p^*_\theta)\, d\mu(\theta) \geq 0$, ensuring that $(R^* _{\theta} , p^*_{\theta}) \in \cI\cR$ since the lowest type satisfies the agent's participation constraint, by Proposition \ref{prop:IR_lowest_type}.\qed

\bigskip
\subsection{Proof of Proposition \ref{prop:objective_agent}}
First, we know that
\begin{align*}
\int_\Theta U_\theta (R_\theta, p_\theta) \,d\mu(\theta) 
&= \int_\Theta\l[ - p_{\theta} - \int_0^{ \bar L}  \left[ 1 - g_{\theta} \big( F_{\theta}(l)  \big)  \right] \, \frac{ \partial R_{\theta}(l) }{ \partial l } \, dl \r] \, d\mu(\theta). 
\end{align*}

Since $(R_\theta, p_\theta)_{\theta \in \Theta} \in \cI\cC$ then $p_\theta$ satisfies \eqref{eq:premium} by Proposition \ref{IC_characterization}. Hence,
\begin{align*}
\int_\Theta U_\theta (R_\theta, p_\theta) \,d\mu(\theta) 
&= - p_{\underline{\theta} } - \int_0^{ \bar L}  \left[ 1 - g_{ \underline{ \theta }} \big( F_{\underline{ \theta} } (l)  \big)   \right] \,\, \frac{ \partial R_{\underline{ \theta }}(l) }{ \partial l } \, dl  \\
&\quad + \int_\Theta \int_{\underline{\theta}} ^ {\theta} \int_0 ^{\bar L}  \left[ \frac{\partial g_s}{ \partial s} (F_s(l)) + g'_s(F_s(l)) \frac{\partial F_s(l)}{\partial s} \right] \frac{\partial R_s(l)}{\partial l } \, dl \,  ds \, d\mu(\theta). 
\end{align*}

\noindent Integrating the third term by parts yields:
\begin{align*}
\int_\Theta U_\theta (R_\theta, p_\theta) \,d\mu(\theta) 
&=- p_{\underline{\theta} } - \int_0^{ \bar L}  \left[ 1 - g_{ \underline{ \theta }} \big( F_{\underline{ \theta} } (l)  \big)   \right] \,\, \frac{ \partial R_{\underline{ \theta }}(l) }{ \partial l } \, dl  \\
&\quad- 
\int_\Theta\int_0^{\bar L}
- \left[\frac{\partial g_\theta}{\partial \theta}(F_\theta(l))+g'_\theta(F_\theta(l))\frac{\partial F_\theta(l)}{\partial \theta}
\right]
\frac{\partial R_\theta(l)}{\partial l}
\bar Q(\theta)
\,dl\,d\theta.
\end{align*}

\noindent Hence, the above expression can be rewritten as follows:
\begin{align*}
\int_\Theta U_\theta (R_\theta, p_\theta) \,d\mu(\theta) 
&=- p_{\underline{\theta} } - \int_0^{ \bar L}  \left[ 1 - g_{ \underline{ \theta }} \big( F_{\underline{ \theta} } (l)  \big)   \right] \,\, \frac{ \partial R_{\underline{ \theta }}(l) }{ \partial l } \, dl - \int_\Theta \int_0^{\bar L} J^A_\theta(l) \frac{\partial R_\theta(l)}{\partial l} \, dl \, d\mu(\theta),
\end{align*}

\noindent where $J^A_\theta(l) = -\frac{\bar Q(\theta)}{q(\theta)} \left[\frac{\partial g_\theta}{\partial \theta}(F_\theta(l))+g'_\theta(F_\theta(l))\frac{\partial F_\theta(l)}{\partial \theta}
\right] \geq0$, since $\frac{\partial g_{\theta}}{ \partial \theta} (F_{\theta}(l)) + g'_{\theta}(F_{\theta}(l))\frac{\partial F_{\theta}(l)}{\partial \theta} \leq 0$.\qed

\bigskip
\subsection{Proof of Proposition \ref{prop:solution_characterization_a1}}
We know from Proposition \ref{prop:objective_agent} that all retention-dependent terms are non-positive. 
For $\theta = \underline{\theta}$, we define  
$$\tilde J^A_{\underline{\theta}} (l) :=  \left[ 1- g_{\underline{\theta}}(F_{\underline{\theta}}(l)) \right] + J^A_{\underline{\theta}} (l) \geq J^A_{\underline{\theta}} (l) \geq 0.$$
The optimal retention function for the lowest risk type $\underline{\theta}$, satisfies:
$$
R_{\underline{\theta} }^* \in \underset{  R_{ \underline{\theta} } \in \cR }{ \arg \max  } - \int_0^{\bar L} \tilde J^A_{\underline{\theta}} (l) \frac{ \partial R_{  \underline{\theta} }(l) }{ \partial l }\, dl . 
$$

\noindent For all other types $\theta \in \Theta$, the optimal retention satisfies
$$
R^*_{\theta} \in \underset{R_{ \theta} \in \cR }{ \arg \max} \ - \int_0^{\bar L} J^A_{\theta}(l) \frac{ \partial R_{ \theta}(l)}{\partial l} \,dl,
$$

\noindent where $J^A_\theta(l) \geq 0$. Hence the above problem can be maximized by choosing $ \frac{ \partial R^*_{\theta }(l) }{ \partial l } =0 $, for a.e.\ $l \in [0, \bar L]$, which implies that $R^*_\theta(l)=0$ for all $l$ and $\theta \in \Theta$. 

It remains to determine the premium. Under full coverage, we have that 
$U_\theta(0,p^*_\theta)=-p^*_\theta$.
Incentive compatibility therefore implies that the premium is pooling:
$$p^*_\theta=p^*, \ \text{ for all } \theta\in\Theta.$$

For $(0, p^*)_{\theta \in \Theta}$, the insurer's participation constraint yields $p^* \geq K$. 
On the other hand, the agent's participation constraint yields
$$ p ^*\leq \int_0^{\bar L} \left[1-g_\theta(F_\theta(l))\right]\,dl, \, \forall \theta\in\Theta.$$

\noindent Since the composite function $g_\theta(F_\theta(l))$ is non-increasing in $\theta$, then   
$p^* \leq M$. 
Using condition \eqref{eq_case3_condition}, any choice of $ p^* \in [K, M]$ leads to an IC-IR menu $(0, p^*)_{\theta \in \Theta}$. 

Let us now determine which feasible pooling premium is optimal. For the menu $(0, p)_{\theta \in \Theta}$, since $\mu(\Theta) =1$, then 
$$ \int_\Theta U_\theta (0, p) \, d\mu(\theta) = -p, $$

\noindent which is decreasing in $p$. It follows that, among all feasible full-coverage pooling premia $p\in[K,M]$, the unique welfare-maximizing premium is the smallest feasible one, namely $p=K$. Thus the optimal full-coverage pooling menu is $(R_\theta^*,p_\theta^*)_{\theta\in\Theta}=
(0,K)_{\theta\in\Theta}$.
Since all types receive the same full-coverage contract $(0, p^*)$, it follows that for every $\theta \in \Theta$, $U_\theta (0, p^*) =- p^*$. Hence, incentive compatibility of the optimal menu holds trivially. At this premium,
$$
\int_\Theta V_\theta(0,p^*)\,d\mu(\theta)
=
p^*-K
=
0,
$$
and so the insurer's participation constraint binds. It remains to verify that the agent's participation holds so that individual rationality of the optimal menu is satisfied. We have that
$$
U_\theta(0, p^*) = -p^* = - \int_\Theta\int_0^{\bar L}
\left[
1-g^{In}(F_\theta(l))
\right]\,dl\,d\mu(\theta)\geq - \int_0^{\bar L} \left[1 - g_{\underline\theta}(F_{\underline\theta}(l))\right] \, dl . 
$$

Since the composite function $g_\theta \circ F_\theta$ is non-increasing in $\theta$ by Remark \ref{Re:chain_rule}, it follows that $g_\theta(F_\theta(l)) \leq g_{\underline{\theta}}(F_{\underline{\theta}}(l))$ for all $\theta \in \Theta$. Hence, 
$$
U_\theta(0, p^*) \geq - \int_0^{\bar L} \left[1 - g_{\underline\theta}(F_{\underline\theta}(l))\right] \, dl \geq - \int_0^{\bar L} \left[1 - g_{\theta}(F_{\theta}(l))\right] \, dl
= U_\theta(L_\theta,0), 
$$
which completes the proof.\qed

\bigskip
\section{Welfare Support of Pareto Optima} \label{AppendixPOsupport}

In this section, we consider the case of Yaari Dual Utilities, as in Section \ref{sec:DU}, and we provide two partial converse statements to Theorem \ref{th:IPOiff} in this case. We first introduce the following lemma which ensures that the function $u$ that maps each type to the pair of the agent's utility function and the insurer's aggregate utility belongs to a Bochner $L^p$ space.

\medskip

\begin{lemma}\label{le:assumsatisfied}
Suppose that for every menu of contracts $
(R_\theta,p_\theta)_{\theta\in\Theta}\in \cI\cR\cap \cI\cC$, the maps
$$
\theta \mapsto U_\theta(R_\theta,p_\theta)
\ \ \hbox{and} \ \ 
\theta \mapsto V_\theta(R_\theta,p_\theta)
$$
are $\cB(\Theta)$-measurable. Then, for every $p\in(1,\infty)$, there exists a constant $M<\infty$ such that for every
$(R_\theta,p_\theta)_{\theta\in\Theta}\in \cI\cR\cap \cI\cC$,
the map
$$\theta \mapsto 
u(\theta):=
\l(
U_{\theta}(R_{\theta},p_{\theta}),
\int_{\Theta}V_{\vartheta}(R_{\vartheta},p_{\vartheta})\,d\mu(\vartheta)
\r)
$$
is in $L^p(\Theta;\bbR^2)$ and satisfies $\|u\|_{L^p}\leq M$.
\end{lemma}

\medskip

\begin{proof}
Fix a menu $(R_\theta,p_\theta)_{\theta\in\Theta}\in \cI\cR\cap \cI\cC$, and let $
\Pi:=\int_\Theta V_\vartheta(R_\vartheta,p_\vartheta)\,d\mu(\vartheta)$. We show that the map
$$
\theta \mapsto
u(\theta):=\l(U_\theta(R_\theta,p_\theta),\Pi\r)
$$

\noindent is in $L^p(\Theta;\bbR^2)$, with a bound independent of the menu. First, recall that
$$
U_\theta(R_\theta,p_\theta)
=
-p_\theta
-
\int_0^{\bar L}
\l[1-g_\theta(F_\theta(l))\r]
\frac{\partial R_\theta(l)}{\partial l}\,dl.
$$

\medskip

By assumption, the map $\theta \mapsto U_\theta(R_\theta,p_\theta)$ is measurable. Since $\Pi$ is a constant, it follows that $u$ is $\bbR^2$-valued and strongly measurable. Next, by Proposition \ref{IR_characterization}, IR implies that for every $\theta\in\Theta$,
$$
p_\theta
\leq
\int_0^{\bar L}
\l[1-g_\theta(F_\theta(l))\r]
\l(1-\frac{\partial R_\theta(l)}{\partial l}\r)\,dl.
$$
Since $0\leq 1-g_\theta(F_\theta(l))\leq 1$ and $0\leq 1-\frac{\partial R_\theta(l)}{\partial l}\leq 1$, a.e., we obtain
$$
p_\theta\le \bar L,
\ \ \forall\,\theta\in\Theta.
$$

\noindent By admissibility of contracts, we also have $p_\theta\geq 0$. Hence
$$
0\leq p_\theta\leq \bar L,
\ \ \forall\,\theta\in\Theta.
$$

\medskip

Additionally, since $0\leq 1-g_\theta(F_\theta(l))\leq 1$ and 
$0 \leq \frac{\partial R_\theta(l)}{\partial l}\leq 1$, a.e., it follows that
$$
\l|
\int_0^{\bar L}
\l[1-g_\theta(F_\theta(l))\r]
\frac{\partial R_\theta(l)}{\partial l}\,dl
\r|
\leq \bar L.
$$
Therefore,
$$
|U_\theta(R_\theta,p_\theta)|
\leq |p_\theta|+\bar L
\leq 2 \, \bar L,
\ \ \forall\,\theta\in\Theta.
$$

\medskip

Similarly,
$$
V_\theta(R_\theta,p_\theta)
=
p_\theta
-
\int_0^{\bar L}
\l[1-g^{In}(F_\theta(l))\r]
\l(1-\frac{\partial R_\theta(l)}{\partial l}\r)\,dl.
$$
By IR, we have $\Pi \geq 0$. Moreover,
$$
V_\theta(R_\theta,p_\theta)\leq p_\theta\leq \bar L,
\ \ \forall\,\theta\in\Theta,
$$
and hence, since $\mu(\Theta) =1$, we obtain
$$
0\leq \Pi\leq \int_\Theta p_\theta\,d\mu(\theta)\leq \bar L.
$$

\medskip

Consequently, for every $\theta\in\Theta$, we have
$$
\|u(\theta)\|_{\bbR^2}
=
\sqrt{U_\theta(R_\theta,p_\theta)^2+\Pi^2}
\leq |U_\theta(R_\theta,p_\theta)|+|\Pi|
\leq 3\,\bar L.
$$
Hence,
$$
\|u\|_{L^p}^p
=
\int_\Theta \|u(\theta)\|_{\bbR^2}^p\,d\mu(\theta)
\leq
(3\,\bar L)^p\,\mu(\Theta)
=
(3\,\bar L)^p.
$$
Therefore,
$$
\|u\|_{L^p}\leq 3\,\bar L.
$$
\end{proof}

\medskip

The following result provides an approximate welfare functional support to any given incentive efficient menu. 

\medskip

\begin{theorem}
\label{th:IPO_Yaari_direct_support}

Suppose that the set 
$$
K: = \left\{ u \in L^p (\Theta; \R^2) : u(\theta) = \left( U_{\theta} (I_{\theta}, p_{\theta}) ,\int_{\Theta} V_{\vartheta}(I_{\vartheta}, p_{\vartheta}) \,d \mu(\vartheta)\right) , \text{ for some } \, (I_{\theta}, p_{\theta})_{\theta \in \Theta} \in \cI\cR \cap \cI\cC \right\}.
$$

\noindent is closed in $L^p(\Theta;\bbR^2)$. If $(I^*_\theta,p^*_\theta)_{\theta\in\Theta}\in \cI\cP\cO$, then for every $\varepsilon > 0$, there exist some $\phi_\varepsilon \in L^q_+(\Theta,\mu)$ and some $\beta_\varepsilon \in \bbR_+$ such that for every feasible menu $m=
(I_\theta,p_\theta)_{\theta\in\Theta}\in \cI\cR\cap \cI\cC$,
\begin{equation}
\begin{split}
&\int_\Theta \phi_\varepsilon(\vartheta)\,U_\vartheta(I_\vartheta,p_\vartheta)\,d\mu(\vartheta)
+
\beta_\varepsilon \, \int_\Theta V_\vartheta(I_\vartheta,p_\vartheta)\,d\mu(\vartheta)\\
&\qquad\qquad <
\int_\Theta \phi_\varepsilon(\vartheta)\,U_\vartheta(I^*_\vartheta,p^*_\vartheta)\,d\mu(\vartheta)
+
\beta_\varepsilon\,\left(\int_\Theta V_\vartheta(I^*_\vartheta,p^*_\vartheta)\,d\mu(\vartheta)+\varepsilon\right).
\end{split}
\label{eq:eps_support_ineq}
\end{equation}
\end{theorem}

\medskip

\begin{proof}
For each feasible menu $m=(I_\theta,p_\theta)_{\theta\in\Theta}\in \cI\cR\cap \cI\cC$, let 
$$
\Pi_m:=\int_\Theta V_{\vartheta}(I_{\vartheta}, p_{\vartheta}) \,d \mu(\vartheta),
$$
and define the function $u_m$ by:
$$
u_m(\theta) :=
\l(
U_\theta(I_\theta,p_\theta),
\Pi_m
\r), \ \ \forall \, \theta \in \Theta.
$$

\medskip

First, we show that $K$ is convex and weakly compact. Consider a finite collection of feasible menus 
$$
m^{(k)}=(I_\theta^{(k)},p_\theta^{(k)})_{\theta\in\Theta}\in \cI\cR\cap \cI\cC,
\ \ \hbox{for} 
\ k=1,\ldots,m,
$$

\medskip

\noindent and let $\beta_1,\ldots,\beta_m \in[0,1]$ be such that $\sum_{k=1}^m \beta_k=1$. Let
$$
\bar I_\theta:=\sum_{k=1}^m \beta_k I_\theta^{(k)}
\ \ \hbox{and} \ \ 
\bar p_\theta:=\sum_{k=1}^m \beta_k p_\theta^{(k)}.
$$

\noindent 
Since $I_\theta^{(k)}\in \cI$ for all $k$, the function $\bar I_\theta$ is non-decreasing, $1$-Lipschitz, and satisfies $\bar I_\theta(0)=0$. Hence $\bar I_\theta\in \cI$, for every $\theta \in \Theta$. Fix $\theta,\theta'\in\Theta$. Since each $I_{\theta'}^{(k)}$ is non-decreasing, the vector of random variables
$$\left(
I_{\theta'}^{(1)}(L_\theta),
\ldots,
I_{\theta'}^{(m)}(L_\theta)
\right)$$

\noindent is pairwise comonotone.

\medskip

For an integrable random variable $Z$, let
$$
Y_\theta(Z):=\int_0^1 F_Z^{-1}(t)\,dg_\theta(t)
\ \ \hbox{and} \ \ 
Y^I(Z):=\int_0^1 F_Z^{-1}(t)\,dg^{In}(t).
$$
The insured's and insurer's utilities from a contract $(R,p)\in \cR\times\bbR_+$ can then be written as
$$
U_\theta(R,p)=Y_\theta(-p-R(L_\theta))
\ \ \hbox{and} \ \ 
V_\theta(R,p)=Y^I(p-L_\theta+R(L_\theta)).
$$

\noindent The translation invariance, positive homogeneity, and comonotonic additivity of the Choquet integral implies that
\begin{align*}
U_\theta(\bar I_{\theta'},\bar p_{\theta'})
&=
-\bar p_{\theta'}
+
 Y_\theta\left(-L_\theta+\bar I_{\theta'}(L_\theta)\right)\\
&=
-\sum_{k=1}^m \beta_k \, p_{\theta'}^{(k)}
+
Y_\theta\left(
\sum_{k=1}^m \beta_k \, \left(-L_\theta+I_{\theta'}^{(k)}(L_\theta)\right)
\right)\\
&=
\sum_{k=1}^m \beta_k \, U_\theta(I_{\theta'}^{(k)},p_{\theta'}^{(k)}).
\end{align*}

\noindent Similarly, for the insurer, we have
$$
V_\theta(\bar I_{\theta'},\bar p_{\theta'})
=
\sum_{k=1}^m \beta_k \, V_\theta(I_{\theta'}^{(k)},p_{\theta'}^{(k)}).
$$

\noindent Since $m^{(k)}$ is IC for each $k$, we have
$$
U_\theta(I_\theta^{(k)},p_\theta^{(k)})
\geq
U_\theta(I_{\theta'}^{(k)},p_{\theta'}^{(k)}),
\ \ \forall\,\theta,\theta',\ \forall\,k.
$$

\noindent Multiplying by $\beta_k$ and summing over $k$ yields
$$
U_\theta(\bar I_\theta,\bar p_\theta)\ge U_\theta(\bar I_{\theta'},\bar p_{\theta'}),
\ \ \forall\,\theta,\theta',
$$

\noindent and so the menu $\bar m :=(\bar I_{\theta},\bar p_{\theta})_{\theta\in\Theta}$ is IC.

\medskip

Similarly, averaging the IR inequalities for the insured and the insurer shows that the averaged menu $\bar M$ is also IR. Therefore,
$$
\bar m:=(\bar I_\theta,\bar p_\theta)_{\theta\in\Theta}\in \cI\cR\cap \cI\cC,
\ \ \hbox{and} \ \ 
u_{\bar m}
=
\sum_{k=1}^m \beta_k \, u_{m^{(k)}} \in L^p(\Theta;\bbR^2).
$$
Hence $K$ is convex.

\medskip

Now, by Lemma \ref{le:assumsatisfied}, the set $K$ is bounded in $L^p(\Theta;\bbR^2)$. Since $1<p<\infty$, the space $L^p(\Theta;\bbR^2)$ is reflexive. Therefore, since $K$ is convex, closed, and bounded, it follows from Corollary \ref{cor:closed_bounded_convex_weakly_compact} that $K$ is weakly compact.

\medskip

Next, let $\Pi^\star:=\displaystyle\int_\Theta V_\vartheta(I^*_\vartheta,p^*_\vartheta)\,d\mu(\vartheta)$, and for $\varepsilon>0$, let
$$
\cV_\varepsilon
:=
\left\{
\phi=(\phi_1,\phi_2)\in L^p(\Theta;\bbR^2):
\phi_1(\theta)\geq U_\theta(I^*_\theta,p^*_\theta), \  \mu\hbox{-a.e.},
\ \, 
\phi_2(\theta)\geq \Pi^\star+\varepsilon, \ \mu\hbox{-a.e.}
\right\}.
$$

\noindent The set $\cV_\varepsilon$ is convex and norm closed in $L^p(\Theta;\bbR^2)$. Since every norm-closed convex subset of a Banach space is weakly closed \cite[Corollary 4, p.12]{Diestel1984}, it follows that $\cV_\varepsilon$ is weakly closed. We claim that $K\cap \cV_\varepsilon=\varnothing$. Indeed, if $u_m\in K\cap \cV_\varepsilon$, then for the feasible menu $m=(I_\theta,p_\theta)_{\theta\in\Theta}\in \cI\cR\cap \cI\cC$ corresponding to $u_m$, we have
$$
U_\theta(I_\theta,p_\theta)\geq U_\theta(I^*_\theta,p^*_\theta), \ \mu\hbox{-a.e.},
\ \ \hbox{and} \ \
\int_\Theta V_{\vartheta}(I_{\vartheta}, p_{\vartheta}) \,d \mu(\vartheta)
\geq
\Pi^\star+\varepsilon
>
\Pi^\star,
$$

\noindent contradicting the fact that $m^\star$ is IPO. Hence,
$$
K\cap \cV_\varepsilon=\varnothing.
$$

Since $K$ is weakly compact and convex, and $\cV_\varepsilon$ is weakly closed and convex, the strict separation theorem \cite[Theorem 7.3.4, p.\ 131]{Jarchow1981} gives a nonzero weakly continuous linear functional
$$
\psi_\varepsilon\in \bigl(L^p(\Theta;\bbR^2)\bigr)^*
=L^q(\Theta;\bbR^2)
$$
such that
\begin{equation}
\label{eq:strict_sep}
\sup_{u\in K}\,\psi_\varepsilon(u)
<
\inf_{\phi\in \cV_\varepsilon}\,\psi_\varepsilon(\phi).
\end{equation}

\noindent Furthermore, by Proposition \ref{prop:duality}, there exists a unique $\delta^\varepsilon=(\delta_1^\varepsilon,\delta_2^\varepsilon)\in L^q(\Theta;\bbR^2)$, with $\delta^\varepsilon\neq 0$, such that 
$$
\psi_\varepsilon(f)
= \int_{\Theta} f_1(\vartheta) \, \delta_1^\varepsilon(\vartheta) \,d\mu(\vartheta) + \int_{\Theta} f_2(\vartheta) \, \delta_2^\varepsilon(\vartheta)\, d\mu(\vartheta),
$$

\noindent for all $f = (f_1, f_2) \in L^p(\Theta;\R^2)$. 

\medskip

We now show that $\delta_1^\varepsilon \geq 0$ and $\delta_2^\varepsilon\geq 0$, $\mu$-a.e. For each $\theta \in \Theta$, let
$$
\phi_{\min}^\varepsilon(\theta)
:=
\l(U_\theta(I^*_\theta,p^*_\theta),\,\Pi^\star+\varepsilon\r).
$$

\noindent Then $\phi_{\min}^\varepsilon\in \cV_\varepsilon$. Moreover, if $h=(h_1,h_2)\in L^p(\Theta;\bbR^2)$ is such that $h_1\geq 0$ and $h_2\geq 0$, $\mu$-a.e., then $\phi_{\min}^\varepsilon+h\in \cV_\varepsilon$. Let $A:=\l\{\theta \in \Theta: \delta_1^\varepsilon(\theta) < 0\r\}$ and $h:=(\mathbf 1_A,0)$. Suppose, by way of contradiction, that $
\mu(A)>0$. Then, for every $t>0$, we  have 
$$
\phi_{\min}^\varepsilon + t \, h
=
\l(
\l(\phi_{\min}^\varepsilon\r)_1 + t \, \mathbf \, 1_A, \,
\l(\phi_{\min}^\varepsilon\r)_2
\r)
\in \cV_\varepsilon.
$$

\noindent 
Moreover,
\begin{align*}
\psi_\varepsilon(\phi_{\min}^\varepsilon+t\,h)
&=
\int_\Theta \Bigl((\phi_{\min}^\varepsilon)_1(\vartheta)+t\,\mathbf 1_A(\vartheta)\Bigr) \, \delta_1^\varepsilon(\vartheta)\,d\mu(\vartheta)
+
\int_\Theta (\phi_{\min}^\varepsilon)_2(\vartheta)\,\delta_2^\varepsilon(\vartheta)\,d\mu(\vartheta)\\
&=
\int_\Theta (\phi_{\min}^\varepsilon)_1(\vartheta)\,\delta_1^\varepsilon(\vartheta)\,d\mu(\vartheta)
+
\int_\Theta (\phi_{\min}^\varepsilon)_2(\vartheta)\,\delta_2^\varepsilon(\vartheta)\,d\mu(\vartheta)
+
t \, \int_\Theta \mathbf 1_A(\vartheta)\,\delta_1^\varepsilon(\vartheta)\,d\mu(\vartheta)\\
&=
\psi_\varepsilon(\phi_{\min}^\varepsilon)
+
t \, \int_A \delta_1^\varepsilon \, d\mu.
\end{align*}

\noindent Since $\int_A \delta_1^\varepsilon\,d\mu < 0$, it follows that 
$$
\psi_\varepsilon(\phi_{\min}^\varepsilon+t h)
\ \underset{t\to\infty} \longrightarrow \ 
-\infty,
$$

\noindent which contradicts the inequality
$$
\psi_\varepsilon(\phi) 
\geq \underset{\varphi \,\in\, \cV_\varepsilon} \inf \, \psi_\varepsilon(\varphi)
> \underset{u\,\in\, K} \sup \, \psi_\varepsilon(u)
> -\infty,
\ \ \forall \, \phi \in \cV_\varepsilon,
$$

\noindent where the last inequality follows from $|\psi_\varepsilon(f)| <+\infty$ for all $f \in L^p(\Theta;\R^2)$, by Proposition \ref{prop:duality}. Hence $\delta_1^\varepsilon \in L^q_+(\Theta,\mu)$. The same argument, using $\tilde h=(0,\mathbf 1_A)$, shows that $\delta_2^\varepsilon \in L^q_+(\Theta,\mu)$.

\medskip

Now, since $\delta_1^\varepsilon \geq 0$ and $\delta_2^\varepsilon \geq 0$, $\mu$-a.e., the functional
$$
\psi_\varepsilon(f)
=
\int_\Theta f_1(\vartheta)\,\delta_1^\varepsilon(\vartheta)\,d\mu(\vartheta)
+
\int_\Theta f_2(\vartheta)\,\delta_2^\varepsilon(\vartheta)\,d\mu(\vartheta)
$$
is monotone with respect to the pointwise order on \(L^p(\Theta;\bbR^2)\). Now, let $\phi = (\phi_1,\phi_2)\in \cV_\varepsilon$. By definition of $\cV_\varepsilon$, we have $\phi_1(\theta)\geq (\phi_{\min}^\varepsilon)_1(\theta)$ and $\phi_2(\theta)\geq (\phi_{\min}^\varepsilon)_2(\theta)$, $\mu$-a.e. Hence, $\psi_\varepsilon(\phi)\geq \psi_\varepsilon(\phi_{\min}^\varepsilon)$. Since $\phi_{\min}^\varepsilon\in \cV_\varepsilon$, it follows that
$$
\underset{\phi \,\in\, \cV_\varepsilon} \inf \, \psi_\varepsilon(\phi)
=
\psi_\varepsilon(\phi_{\min}^\varepsilon).
$$

\noindent Moreover, the strict separation inequality in \eqref{eq:strict_sep} gives
$$
\underset{u\in K} \sup \, \psi_\varepsilon(u)
<
\underset{\phi\in \cV_\varepsilon} \inf \, \psi_\varepsilon(\phi),
$$
and hence
$$
\underset{u\in K} \sup \, \psi_\varepsilon(u)
<
\psi_\varepsilon(\phi_{\min}^\varepsilon).
$$

For any feasible menu $m=(I_\theta,p_\theta)_{\theta\in\Theta}$ with $u_m \in K$, we have $u_m(\theta) = \l(U_\theta(I_\theta,p_\theta), \Pi_m\r)$, for all $\theta \in \Theta$, where $\Pi_m=\int_\Theta V_\vartheta(I_\vartheta,p_\vartheta)\,d\mu(\vartheta)$.
Therefore,
$$\psi_\varepsilon(u_m)
=
\int_\Theta \delta_1^\varepsilon(\vartheta)\,U_\vartheta(I_\vartheta,p_\vartheta)\,d\mu(\vartheta)
+
\beta_\varepsilon \, \Pi_m,
$$
where
$$
\beta_\varepsilon:=\int_\Theta \delta_2^\varepsilon(\vartheta)\,d\mu(\vartheta)\geq 0.
$$
Similarly,
$$
\psi_\varepsilon(\phi_{\min}^\varepsilon)
=
\int_\Theta \delta_1^\varepsilon(\vartheta)\,U_\vartheta(I^*_\vartheta,p^*_\vartheta)\,d\mu(\vartheta)
+
\beta_\varepsilon \, (\Pi^\star+\varepsilon).
$$

\noindent Thus,
\begin{equation*}
\begin{split}
&\int_\Theta \delta_1^\varepsilon(\vartheta)\,U_\vartheta(I_\vartheta,p_\vartheta)\,d\mu(\vartheta)
+
\beta_\varepsilon \, \int_\Theta V_\vartheta(I_\vartheta,p_\vartheta)\,d\mu(\vartheta)\\
&\qquad\qquad <
\int_\Theta \delta_1^\varepsilon(\vartheta)\,U_\vartheta(I^*_\vartheta,p^*_\vartheta)\,d\mu(\vartheta)
+
\beta_\varepsilon\,\left(\int_\Theta V_\vartheta(I^*_\vartheta,p^*_\vartheta)\,d\mu(\vartheta)+\varepsilon\right).
\end{split}
\end{equation*}
Letting $\phi_\varepsilon:=\delta_1^\varepsilon$ concludes the proof.
\end{proof}

\medskip

The second result gives an exact welfare functional support to any incentive efficient menu, but under some stronger structural assumptions.

\medskip

\begin{theorem}
\label{thm:direct_exact_support}
Let $m^\star=(I_\theta^\star,p_\theta^\star)_{\theta\in\Theta}\in \cI\cR\cap \cI\cC$, and let
$$
S_{m^\star}:=K-u_{m^\star}
=
\{u-u_{m^\star}:u\in K\}
\subset L^p(\Theta;\bbR^2).
$$ 

\noindent Suppose that there exists an open convex cone
$D\subset L^p(\Theta;\bbR^2)$
such that the following hold:
\smallskip
\begin{enumerate}
\item[(i)] $S_{m^\star} \cap D=\varnothing$;
\medskip
\item[(ii)] $(0,1)\in D$;
\medskip
\item[(iii)] For every $A\in \cB(\Theta)$ with $\mu(A)>0$, we have $(\mathbf 1_A,0)\in D$.
\end{enumerate}

\medskip

\noindent Then there exist some $\phi\in L^q_+(\Theta,\mu)$ and $\beta>0$, with $\phi>0$, $\mu$-a.e., such that
$$
m^\star
\in
\underset{(I_\theta,p_\theta)_{\theta\in\Theta}\,\in\, \cI\cR\cap \cI\cC}{\argsup}
\left\{
\int_\Theta \phi(\vartheta)\,U_\vartheta(I_\vartheta,p_\vartheta)\,d\mu(\vartheta)
+
\beta \, \int_\Theta V_\vartheta(I_\vartheta,p_\vartheta)\,d\mu(\vartheta)
\right\}.
$$

\noindent Consequently, letting
$$
m_\phi:=\int_\Theta \phi(\vartheta)\,d\mu(\vartheta)>0,
\ \ 
d\eta:=\frac{\phi}{m_\phi}\,d\mu,
\ \hbox{ and } \ 
\alpha:=\frac{m_\phi}{m_\phi+\beta}\in(0,1),
$$
the measure $\eta$ is equivalent to $\mu$, and
$$
m^\star
\in
\underset{(I_\theta,p_\theta)_{\theta\in\Theta}\,\in\, \cI\cR\cap \cI\cC} \argsup
\l\{
\alpha\int_\Theta U_\vartheta(I_\vartheta,p_\vartheta)\,d\eta(\vartheta)
+
(1-\alpha)\int_\Theta V_\vartheta(I_\vartheta,p_\vartheta)\,d\mu(\vartheta)
\r\}.
$$
\end{theorem}

\medskip

\begin{proof}
Since $K$ is convex, so is $S_{m^\star}$. By assumption, there exists an open convex cone
$D\subset L^p(\Theta;\bbR^2)$ such that $S_{m^\star}\cap D=\varnothing$. Since $0=u_{m^\star}-u_{m^\star}\in S_{m^\star}$, the set $S_{m^\star}$ is nonempty. Because $S_{m^\star}$ and $D$ are nonempty disjoint convex subsets of the locally convex Hausdorff space $L^p(\Theta;\bbR^2)$, and since $D$ is open, it follows from the Hahn-Banach separation theorem \cite[Theorem 7.3.2, p.\ 130]{Jarchow1981} that there exists a nonzero continuous linear functional $\psi\in \l(L^p(\Theta;\bbR^2)\r)^*=L^q(\Theta;\bbR^2)$ and some $\alpha_0 \in \bbR$ such that
$$
\psi(s)\leq \alpha_0 < \psi(d),
\ \  \forall\,s\in S_{m^\star},\ \forall\,d\in D.
$$
Since \(0\in S_{m^\star}\), we have
$$
0=\psi(0)\leq \alpha_0.
$$

Next, we show that $\psi(d)>0$, for all $d\in D$. Indeed, fix $d\in D$. Since $D$ is a cone, we have $\lambda \, d\in D$, for every $\lambda>0$. Therefore, the separation inequality above gives
$$
\alpha_0<\psi(\lambda \, d) = \lambda \, \psi(d),
\ \ \forall\,\lambda>0.
$$

\noindent If $\psi(d)<0$, then $\lambda \,\psi(d) \to -\infty$ as $\lambda \to +\infty$. That is, for each $R \in \bbR$, there exists some $\lambda_R > 0$ such that $\lambda \, \psi(d) < R$, for all $\lambda \geq \lambda_R$. In particular, for $R=\alpha_0$, there exists $\lambda_{\alpha_0} > 0$ such that $
\lambda \, \psi(d) < \alpha_0$, for all $ \lambda \geq \lambda_{\alpha_0}$, which contradicts the inequality above.
Hence $\psi(d)\geq 0$. If $\psi(d)=0$, then the separation inequality above gives $\alpha_0<0$, contradicting the fact that $\alpha_0\geq 0$. Therefore, 
$$\psi(d)>0, \ \ \forall \, d \in D.$$

Now, for any $d\in D$ and any $\lambda>0$, we have $
\alpha_0 < \lambda \, \psi(d)$. Letting $\lambda\downarrow 0$, we obtain $\alpha_0\leq 0$. Since $\alpha_0\geq 0$, it follows that $
\alpha_0=0$. Hence,
$$
\psi(s)\leq 0<\psi(d),
\ \ \forall\,s\in S_{m^\star},\ \forall\,d\in D.
$$

By Proposition \ref{prop:duality}, there exists a unique $(\phi,\widetilde \phi)\in L^q(\Theta;\bbR^2)$ such that
$$
\psi(f)
=
\int_\Theta f_1(\vartheta)\,\phi(\vartheta)\,d\mu(\vartheta)
+
\int_\Theta f_2(\vartheta)\,\widetilde\phi(\vartheta)\,d\mu(\vartheta),
\ \  \forall\,f=(f_1,f_2)\in L^p(\Theta;\bbR^2).
$$

\noindent Let $\beta := \displaystyle\int_\Theta \widetilde\phi(\vartheta)\,d\mu(\vartheta)$. Since $(0,1)\in D$ by assumption, we have
$$ 0< \psi(0,1) =\beta,$$
and so $\beta>0$. Next, fix $A\in \cB(\Theta)$ such that $\mu(A)>0$. Since $(\mathbf 1_A,0) \in D$ by assumption, we have
$$
0<\psi(\mathbf 1_A,0)=\int_A \phi(\vartheta)\,d\mu(\vartheta).
$$
Since this holds for every $A\in \cB(\Theta)$ such that $\mu(A)>0$, it follows that
$$
\phi >0, \ \mu\text{-a.e.}
$$

Finally, let $u_m\in K$, and let $m = (I_{\theta}, p_{\theta})_{\theta \in \Theta} \in \cI\cR \cap \cI\cC$ be the associated menu, such that $u(\theta) = \l( U_{\theta} (I_{\theta}, p_{\theta}) ,\int_{\Theta} V_{\vartheta}(I_{\vartheta}, p_{\vartheta}) \,d \mu(\vartheta)\r)$, for all $\theta \in \Theta$. Then
$$
u_m -u_{m^\star}\in S_{m^\star},
$$
and therefore $\psi(u_m-u_{m^\star})\leq 0$ by the separation inequality. Since $\psi$ is linear, we obtain
$$
\psi(u_m)\leq \psi(u_{m^\star}),
$$
that is, letting $\Pi_m := \int_{\Theta} V_{\vartheta}(I_{\vartheta}, p_{\vartheta}) \,d \mu(\vartheta)$ and $\Pi^\star_m := \int_{\Theta} V_{\vartheta}(I^\star_{\vartheta}, p^\star_{\vartheta}) \,d \mu(\vartheta)$, 
\begin{align*}
\int_\Theta \phi(\vartheta)\,U_\vartheta(I_\vartheta,p_\vartheta)\,d\mu(\vartheta)
+
\beta\,\Pi_m
\leq
\int_\Theta \phi(\vartheta)\,U_\vartheta(I_\vartheta^\star,p_\vartheta^\star)\,d\mu(\vartheta)
+
\beta\,\Pi^\star.
\end{align*}

\noindent Consequently,
$$
m^\star
\in
\underset{(I_\theta,p_\theta)_{\theta\in\Theta}\,\in\, \cI\cR\cap \cI\cC} \argsup
\l\{
\int_\Theta \phi(\vartheta)\,U_\vartheta(I_\vartheta,p_\vartheta)\,d\mu(\vartheta)
+
\beta \, \int_\Theta V_\vartheta(I_\vartheta,p_\vartheta)\,d\mu(\vartheta)
\r\}.
$$

\medskip

Finally, since $\phi>0$, $\mu$-a.e., let
$$
m_\phi:=\int_\Theta \phi(\vartheta)\,d\mu(\vartheta)>0,
\ \ 
d\eta:=\frac{\phi}{m_\phi}\,d\mu,
\ \hbox{ and } \ 
\alpha:=\frac{m_\phi}{m_\phi+\beta}\in(0,1).
$$

\noindent Then $\eta$ is a probability measure equivalent to \(\mu\), and for every feasible menu $m=(I_\theta,p_\theta)_{\theta\in\Theta} \in \cI\cR\cap \cI\cC$, we have
\begin{align*}
&\int_\Theta \phi(\vartheta)\,U_\vartheta(I_\vartheta,p_\vartheta)\,d\mu(\vartheta)
+
\beta\int_\Theta V_\vartheta(I_\vartheta,p_\vartheta)\,d\mu(\vartheta)\\
&\qquad\qquad=
(m_\phi+\beta)
\left[
\alpha\int_\Theta U_\vartheta(I_\vartheta,p_\vartheta)\,d\eta(\vartheta)
+
(1-\alpha)\int_\Theta V_\vartheta(I_\vartheta,p_\vartheta)\,d\mu(\vartheta)
\right].
\end{align*}

\noindent Since multiplication by the positive constant $m_\phi+\beta$ does not change the set of maximizers, this concludes the proof.
\end{proof}

\vspace{0.8cm}
\bibliographystyle{ecta}
\bibliography{biblio}
\vspace{0.7cm}

\end{document}